\RequirePackage{fix-cm}
\documentclass[aps,pra,superscriptaddress,twocolumn,10pt,bibliography]{revtex4-2}
\usepackage{anyfontsize}

\usepackage{amsmath,amsfonts,amssymb,amstext,amsthm}
\usepackage{mathtools}
\usepackage{mathrsfs}
\usepackage{nccmath}
\usepackage{blkarray}
\usepackage{bm}

\allowdisplaybreaks

\usepackage{graphicx}
\usepackage[dvipsnames,svgnames]{xcolor}

\usepackage{tikz}
\usetikzlibrary{arrows.meta,backgrounds,calc,decorations.pathmorphing,fit,matrix,positioning}

\usepackage{array}
\usepackage{makecell}
\usepackage{afterpage}
\usepackage{pifont}

\usepackage{paralist}
\usepackage[most]{tcolorbox}
\usepackage{comment}
\usepackage[normalem]{ulem}

\newcounter{restated}
\newenvironment{restated}[2]{%
  \stepcounter{restated}%
  \expandafter\def\csname the#1\endcsname{\ref*{#2}}%
  \expandafter\def\csname theH#1\endcsname{restated.\arabic{restated}}%
  \def\thetheorem{\ref*{#2}}%
  \addtocounter{theorem}{-1}%
  \def\restatedenv{#1}%
  \begin{#1}}{\expandafter\end\expandafter{\restatedenv}}

\usepackage[page,header]{appendix}

\usepackage{algpseudocode}

\newcounter{protocol}
\renewcommand{\theprotocol}{\arabic{protocol}}

\newenvironment{protocolbox}[2]{%
    \refstepcounter{protocol}%
    \par\medskip
    \noindent\rule{\linewidth}{0.4pt}\par
    \vspace{0em}
    \noindent\textbf{Protocol~\theprotocol. #2}%
    \label{#1}%
    \par\vspace{-0.6em}
    \noindent\rule{\linewidth}{0.4pt}\par
    \vspace{0.15em}
}{%
    \par\vspace{0.25em}
    \noindent\rule{\linewidth}{0.4pt}\par
    \medskip
}

\algrenewcommand\algorithmicindent{1em}

\makeatletter
\newcommand{\StateP}[1]{%
  \State
  \parbox[t]{\dimexpr\linewidth-\ALG@tlm\relax}{\raggedright #1\strut}%
}

\newcommand{\StatexP}[1]{%
  \Statex
  \parbox[t]{\dimexpr\linewidth-\ALG@tlm\relax}{\raggedright #1\strut}%
}

\makeatother

\usepackage[sort&compress]{natbib}
\usepackage{bibunits}

\makeatletter
\def\stocname{Contents}
\def\supp@numberline#1{\makebox[2.6em][l]{#1}}

\newcommand\l@suppsec[2]{%
  \par\addvspace{4pt}\noindent
  \begingroup
    \let\numberline\supp@numberline
    \bfseries #1%
  \endgroup
  \nobreak\leaders\hbox to .6em{\hfil.\hfil}\hfill
  \nobreak\hbox{#2}\par
}

\newcommand\l@suppsub[2]{%
  \par\noindent\hspace*{2.6em}
  \begingroup
    \let\numberline\supp@numberline
    #1%
  \endgroup
  \nobreak\leaders\hbox to .6em{\hfil.\hfil}\hfill
  \nobreak\hbox{#2}\par
}

\newcommand\supptableofcontents{%
  {\let\addcontentsline\@gobblethree \section*{\stocname}}%
  \begingroup
    \let\toc@pre\@empty
    \let\toc@post\@empty
    \let\l@section\l@suppsec
    \let\l@subsection\l@suppsub
    \@starttoc{stoc}%
  \endgroup
}
\makeatother

\usepackage[linktocpage,colorlinks]{hyperref}
\hypersetup{
    colorlinks = true,
    citecolor  = YellowOrange,
    linkcolor  = RoyalBlue,
    urlcolor   = RedViolet,
}

\usepackage{cleveref}

\newtheorem{theorem}{Theorem}
\newtheorem{proposition}[theorem]{Proposition}
\newtheorem{lemma}[theorem]{Lemma}
\newtheorem{corollary}[theorem]{Corollary}
\newtheorem{definition}[theorem]{Definition}

\theoremstyle{definition}

\AddToHook{env/proposition/begin}{\crefalias{theorem}{proposition}}
\AddToHook{env/lemma/begin}{\crefalias{theorem}{lemma}}
\AddToHook{env/corollary/begin}{\crefalias{theorem}{corollary}}
\AddToHook{env/definition/begin}{\crefalias{theorem}{definition}}
\AddToHook{env/question/begin}{\crefalias{theorem}{question}}
\AddToHook{env/remark/begin}{\crefalias{theorem}{remark}}
\AddToHook{env/example/begin}{\crefalias{theorem}{example}}
\AddToHook{env/conjecture/begin}{\crefalias{theorem}{conjecture}}

\crefname{equation}{Eq.}{Eqs.}
\Crefname{equation}{Eq.}{Eqs.}

\crefname{section}{Sec.}{Secs.}
\Crefname{section}{Sec.}{Secs.}

\crefname{subsection}{Sec.}{Secs.}
\Crefname{subsection}{Sec.}{Secs.}

\crefname{figure}{Fig.}{Figs.}
\Crefname{figure}{Fig.}{Figs.}

\crefname{table}{Table}{Tables}
\Crefname{table}{Table}{Tables}

\crefname{algorithm}{Algorithm}{Algorithms}
\Crefname{algorithm}{Algorithm}{Algorithms}

\crefname{protocol}{Protocol}{Protocols}
\Crefname{protocol}{Protocol}{Protocols}

\crefname{theorem}{Theorem}{Theorems}
\Crefname{theorem}{Theorem}{Theorems}

\crefname{proposition}{Proposition}{Propositions}
\Crefname{proposition}{Proposition}{Propositions}

\crefname{lemma}{Lemma}{Lemmas}
\Crefname{lemma}{Lemma}{Lemmas}

\crefname{corollary}{Corollary}{Corollaries}
\Crefname{corollary}{Corollary}{Corollaries}

\crefname{definition}{Definition}{Definitions}
\Crefname{definition}{Definition}{Definitions}

\crefname{question}{Question}{Questions}
\Crefname{question}{Question}{Questions}

\crefname{remark}{Remark}{Remarks}
\Crefname{remark}{Remark}{Remarks}

\crefname{example}{Example}{Examples}
\Crefname{example}{Example}{Examples}

\crefname{conjecture}{Conjecture}{Conjectures}
\Crefname{conjecture}{Conjecture}{Conjectures}

\crefname{suppsection}{Supplementary Note}{Supplementary Notes}
\Crefname{suppsection}{Supplementary Note}{Supplementary Notes}

\crefname{suppsubsection}{Supplementary Note}{Supplementary Notes}
\Crefname{suppsubsection}{Supplementary Note}{Supplementary Notes}

\newcommand{\ket}[1]{|#1\rangle}

\newcommand{\bra}[1]{\langle#1|}
\newcommand{\ketbra}[2]{|#1\rangle\langle#2|}
\newcommand{\braket}[2]{\langle #1 \vert #2 \rangle}

\DeclareMathOperator{\tr}{Tr}

\newcommand{\rid}{{\mathrm{id}}}
\newcommand{\bR}{\mathbb{R}}
\newcommand{\bC}{\mathbb{C}}

\newcommand{\bI}{\mathbb{I}}
\newcommand{\bE}{\mathbb{E}}

\newcommand{\cD}{{\mathcal{D}}}
\newcommand{\cE}{{\mathcal{E}}}
\newcommand{\cF}{{\mathcal{F}}}

\newcommand{\cH}{{\mathcal{H}}}
\newcommand{\cI}{{\mathcal{I}}}

\newcommand{\cK}{{\mathcal{K}}}
\newcommand{\cL}{{\mathcal{L}}}
\newcommand{\cM}{{\mathcal{M}}}
\newcommand{\cN}{{\mathcal{N}}}
\newcommand{\cO}{{\mathcal{O}}}
\newcommand{\cP}{{\mathcal{P}}}

\newcommand{\cR}{{\mathcal{R}}}
\newcommand{\cS}{{\mathcal{S}}}
\newcommand{\cT}{{\mathcal{T}}}
\newcommand{\cU}{{\mathcal{U}}}
\newcommand{\cV}{{\mathcal{V}}}

\newcommand{\rF}{{\mathrm{F}}}

\newcommand{\rI}{{\mathrm{I}}}

\newcommand{\rP}{{\mathrm{P}}}

\newcommand{\rR}{{\mathrm{R}}}

\newcommand{\sA}{{\mathsf{A}}}
\newcommand{\sB}{{\mathsf{B}}}
\newcommand{\sC}{{\mathsf{C}}}
\newcommand{\sD}{{\mathsf{D}}}
\newcommand{\sE}{{\mathsf{E}}}

\newcommand{\sM}{{\mathsf{M}}}

\newcommand{\sR}{{\mathsf{R}}}

\newcommand{\sX}{{\mathsf{X}}}

\newcommand{\scE}{\mathscr{E}}

\newcommand{\bmE}{\bm{E}}
\newcommand{\bmF}{\bm{F}}
\newcommand{\bmB}{\bm{B}}
\newcommand{\bmq}{\bm{q}}
\newcommand{\bmM}{\bm{M}}

\newcommand{\f}{\frac}

\newcommand{\til}{\tilde}
\newcommand{\widtil}{\widetilde}

\newcommand{\supp}{\mathrm{supp}}

\newcommand{\diag}{\mathrm{diag}}

\newcommand{\rvec}{\mathrm{vec}}
\newcommand{\rspan}{\mathrm{span}}

\newcommand{\rrange}{\mathrm{range}}

\newcommand{\lag}{\langle}
\newcommand{\rag}{\rangle}

\newtcolorbox{dashedbox}[1][]{
  enhanced,
  sharp corners,
  boxrule=0.5pt,
  colback=white,
  colframe=black,
  fonttitle=\bfseries,
  title=#1,
  dash pattern=on 2pt off 2pt,
  borderline={0.5pt}{0pt}{black!60!white,dashed},
  before skip=5pt, after skip=5pt,
  boxsep=5pt
}

\begin{document}

\title{When Classical Correlations Certify Entanglement Recovery}

\author{Takeru~Utsumi}
\email{takeru-utsumi@g.ecc.u-tokyo.ac.jp}
\affiliation{Graduate School of Arts and Sciences, The University of Tokyo, 3-8-1 Komaba, Meguro-ku, Tokyo 153-8902, Japan}
\author{Yota~Tachibana}
\affiliation{Department of Mathematical Informatics, Graduate School of Informatics, Nagoya University, Furo-cho, Chikusa-ku, Nagoya, Aichi 464-8601, Japan}
\author{Yoshifumi~Nakata}
\affiliation{Department of Computer Science, School of Computing, Institute of Science Tokyo, 4259 Nagatsuta-cho, Midori-ku, Yokohama, Kanagawa 226-8501, Japan}
\affiliation{Yukawa Institute for Theoretical Physics, Kyoto University, Kitashirakawa Oiwake-cho, Sakyo-ku, Kyoto 606-8502, Japan}
\author{Takaya~Matsuura}
\affiliation{RIKEN Center for Quantum Computing (RQC), 2-1 Hirosawa, Wako, Saitama 351-0198, Japan}
\author{Ryuji~Takagi}
\affiliation{Graduate School of Arts and Sciences, The University of Tokyo, 3-8-1 Komaba, Meguro-ku, Tokyo 153-8902, Japan}
\author{Francesco~Buscemi}
\affiliation{Department of Mathematical Informatics, Graduate School of Informatics, Nagoya University, Furo-cho, Chikusa-ku, Nagoya, Aichi 464-8601, Japan}
\author{Masato~Koashi}
\affiliation{Photon Science Center, Graduate School of Engineering, The University of Tokyo, 7-3-1 Hongo, Bunkyo-ku, Tokyo 113-8656, Japan}

\date{\today}

\begin{abstract}

Entanglement and state distinguishability have long been central topics in quantum foundations. While each has developed into a rich subject in its own right, they can be connected through measurements in complementary bases. This connection provides insights into quantum and classical correlations and underlies many information-processing tasks, most notably quantum error correction (QEC). However, it has remained largely open whether the connection between entanglement and distinguishability extends to general measurements without assuming complementarity. We answer this question in the affirmative: the key is irreducibility induced by measurements, a much more relaxed condition than complementarity. Specifically, for POVMs satisfying an irreducibility condition, we establish a quantitative relation between the infidelity of one-sided local transformation into a maximally entangled state and the failure probability in state discrimination. We then characterize when both errors can vanish simultaneously. Our results have direct applications to entanglement distillation and QEC. Their errors can be certified by estimating classical input-output correlations from experimentally accessible local measurements, without requiring complementarity. The results also extend to quantum measurement theory. We derive novel trade-off relations that bridge two historically distinct approaches: information gain--irreversibility and observable noise--disturbance. Our connection between entanglement and distinguishability via irreducibility thus offers a common framework for entanglement distillation, QEC, and measurement trade-offs.

\end{abstract}
\maketitle

\section{Introduction}
\label{sec:introduction}
Correlation is a central concept in physics, from classical to quantum theory.
Purely quantum correlations---most notably entanglement---reveal fundamental non-classical features of nature, as exemplified by violations of Bell inequalities~\cite{Bell1964OntheEinsteinPodolskyRosenparadox, Clauser1969ProposedExperimenttoTestLocal, Brunner2014bellnonlocality}, and serve as operational resources for quantum computation~\cite{Raussendorf2001OneWayQuantumComputer, Jozsa2003entanglementquantumcomputationalspeedup, Horodecki2009quantumentanglement}, quantum communication~\cite{Schumacher1996sendingentanglement, lloyd1997capacity, barnum1998inftrans, Barnum2000OnquantumFidel, Shor2003QuantumChannelsHowtoFind}, and quantum sensing and metrology~\cite{Giovannetti2004QuantumEnhancedMeasurements, Giovannetti2011Advancesquantummetrology, Degen2017Quantumsensing}.
Beyond these direct applications, quantum and classical correlations also provide a powerful framework for formulating quantum information-processing tasks~\cite{devetak2005private, Devetak2005capacityQchannel, devetak2005distillation, Abeyesinghe2009Motherfamilytree, Datta2011apexfamilytree, Wilde2017ConvrersePrivate}.
Such formulations offer practical tools for analysis and implementation and, at the same time, often clarify the underlying conceptual structure.

An illustrative example is quantum error correction (QEC)~\cite{Shor1995schemereducingdecoherence, Steane1996errorcorrectingcodesinquantum, LaflammePerfectQuantumErrorCorrectingCode}. 
From a correlation-based perspective, QEC can be characterized by a maximally entangled state (MES) between the system and the reference system that purifies the logical input~\cite{hayden2008decoupling, dupuis2010decoupling, dupuis2014one}.
Together with the fact that an MES can be characterized by correlations in complementary bases, such as mutually unbiased Pauli-$X$ and $Z$ bases, protecting quantum information can be interpreted as protecting classical information encoded in two complementary bases~\cite{Shor2000simpleproofsecurity, Luo2007Efficientlyimplementablecodes, wilde2010convolutional}.
Indeed, measuring the reference system in either basis reduces the QEC problem to discriminating the corresponding output states after noise~\cite{Chefles2000Quantumstatediscrimination, Barnett2009Quantumstatediscrimination, Bae2015statediscriminationapplications}.
This viewpoint is not merely conceptual but is also operationally useful since the correlation-based interpretation turns a genuinely quantum recovery task into the simultaneous recovery of two complementary classical messages. This is most exemplified in Calderbank--Shor--Steane (CSS) codes~\cite{Calderbank1996goodquantumcodeexists, Steane1996MultipleParticleErrorCorrection}, where logical-$X$ and -$Z$ information is handled separately.

A similar structure, concerning correlations and distinguishability, also appears in quantum measurement theory.
Since Heisenberg’s original proposal~\cite{Heisenberg1927Uberden, WheelerZurek1983QuantumTheoryMeasurement}, uncertainty relations have been studied from various, typically operational, perspectives.
In quantum-correlation-based approaches~\cite{Groenewold1971InformationGain, Ozawa1986InformationGain, Fuchs1999informationgainvsstatedis, Winter2001Compressionofquantummeasurement, Winter2004ExtrinsicIntrinsic, Maccone2007Entropicinformation, Buscemi2008GlobalInformationBalance, Wilde2012imulatingquantummeasurements, Berta2014identifyinginformationgain, Anshu2019ConvexSplitHypothesisTesting}, they are formulated as trade-off relations in terms of correlations between the measurement output and a reference system: how much information is gained and how much quantum correlation is irreversibly lost due to the measurement.
In observable-based approaches~\cite{Ozawa2003universallyvalidreformulation, ozawa2004uncertaintyprinciplequantuminstruments, Ozawa2005Universaluncertainty, OZAWA2004ncertaintynoisedisturbancegeneralizedmeasure}, the focus is instead on the impact of measurements on the observables.
Operationally, the trade-offs between noise and disturbance on observables are quantified by how indistinguishable the eigenstates of the observable become after the measurement, or, equivalently, by the loss of classical information about the eigenstate labels~\cite{Buscemi2014NoiseandDisturbance, Renes2017uncertaintyAnoperationalapproach}.
These two formulations,
one based on correlations with a reference system and the other based on distinguishability of classical labels, are different manifestations of common properties of quantum measurements.

The parallel structures described above suggest a general correspondence between quantum correlations and classical distinguishability of measurement outcomes.
Since its first observation nearly a century ago~\cite{Einstein1935CanQuantumMechanical}, such a correspondence has been studied extensively and is now well understood via complementarity of mutually unbiased bases (MUBs)~\cite{koashi2007complementaritydistillablesecretkey, Renes2013thephysics, renes2016uncertainrelationsandAQEC, Renes2022QuantumInformationCM, Berta2014Entanglementassistedguessingcomplementary}. Beyond complementarity, in the zero-error case, perfect distinguishability of the outcomes of an irreducible family of projective measurements certifies maximal entanglement~\cite{Zhu2021ZeroUncertaintyStates, Sarkar2022Certificationincompatiblemeasurements, ran2026zerouncertaintystatesrelativeobservable}. In the presence of errors, a quantitative correspondence for general measurements beyond complementary bases has remained unclear. Prior works~\cite{Berta2010UncertaintyPrinciplequantummemory, Bergh2021ExperimentallyAccessibleBoundsDistillableEntanglement, nakata2025constructing} provided partial evidence for such a quantitative correspondence in the contexts of uncertainty relations and QEC. These results, however, remain perturbative around MUBs, and the relation in the general regime was left as a conjecture. Resolving this problem clarifies the structure of the measurements underlying the correspondence and its robustness to errors. It is also of practical importance in applications such as QEC, where errors are unavoidable and the measurements available on a physical platform are often neither mutually unbiased nor projective.

In this work, we establish a general correspondence between quantum correlations and the distinguishability of classical labels.
We consider a bipartite state in which one subsystem undergoes a measurement from a prescribed family of positive operator-valued measures (POVMs), and ask when the ability to distinguish the measurement outcomes from the other subsystem certifies that a local operation on that subsystem can recover a maximally entangled state. Our main theorem gives quantitative bounds that convert label-discrimination errors into an error bound for local conversion of the bipartite state into an MES, and characterizes when both errors vanish simultaneously. The measurements need not be mutually unbiased or projective. The relevant condition is instead the irreducibility induced by the POVMs, which roughly means that the action of the measurements cannot be simultaneously decomposed into smaller actions. Under this condition, classical distinguishability is sufficient to quantitatively certify quantum correlations.

This criterion leads to various certification methods that use only experimentally accessible classical data. For entanglement distillation, Alice and Bob can certify how well a maximally entangled state can be distilled by one-shot local operations and classical communication (LOCC), from outcome coincidences of the local measurements available to them. In entanglement distribution, the same idea reduces to a send-and-measure test certifying how faithfully entanglement is distributed through an unknown noisy channel. In QEC, it gives a prepare-and-measure test that certifies decoding performance from the decoding errors of associated classical--quantum (CQ) codes. In all of these applications, the recoverability of genuinely quantum correlations is certified from classical input-output statistics under realistic measurement constraints.

The same perspective also unifies two formulations of quantum measurement uncertainty. We introduce one-shot measures of four quantities: the information gain and irreversibility of a measurement, and the noise in measuring an observable and the disturbance to it. We then derive trade-off relations connecting these quantities. These relations show that the information gain--irreversibility and observable noise--disturbance approaches, which have developed in parallel for a long time, can be viewed as two expressions of the same entanglement--distinguishability relation. 

Overall, our results identify irreducibility, rather than complementarity, as a principle behind the quantitative correlation--distinguishability correspondence. Beyond the applications developed here, this irreducibility-based perspective suggests a broadly applicable route for using experimentally accessible classical data to probe entanglement, recoverability, and measurement trade-offs.

\section{Results}
\label{sec:mainresults}

Physical systems are denoted by $\sA, \sB, \dots$, with $\cH^\sA$ the Hilbert space of $\sA$ and $d_\sA = \dim \cH^\sA$. We write $\sA \cong \sB$ when $\cH^\sA$ and $\cH^\sB$ are isomorphic. Operators on $\cH^\sA$ and maps from $\cH^\sA$ to $\cH^\sB$ are denoted by $T^\sA$ and $\cT^{\sA\to\sB}$, respectively. System labels, as well as identity operators and maps, are omitted when clear from the context. Missing system labels on states denote partial traces, e.g., $\rho^\sA=\tr_\sB[\rho^{\sA\sB}]$. A pure state $\ket{\varphi}$ is also written as $\varphi = \ketbra{\varphi}{\varphi}$.
We write $\bar{\cdot}$ and $\cdot^\top$ for the complex conjugate and transpose in a fixed basis, and $\cdot^\dag$ for the Hermitian conjugate. 
Negative powers of a positive semidefinite operator are defined on its support.
The complete notation is summarized in~\cref{sec:notation in appendix}.

\subsection{Connecting entanglement and distinguishability}
\label{sec:main result on entanglement generation and state discrimination}

We start our analysis in~\cref{sec:quantitative connection} with a general task of local state transformation and connect the error to that of the state discrimination. We then investigate in~\cref{sec:irreducibility maximal entanglement} when the local transformation to an MES is achievable, and in~\cref{sec:vanishing errors} when both errors vanish simultaneously.

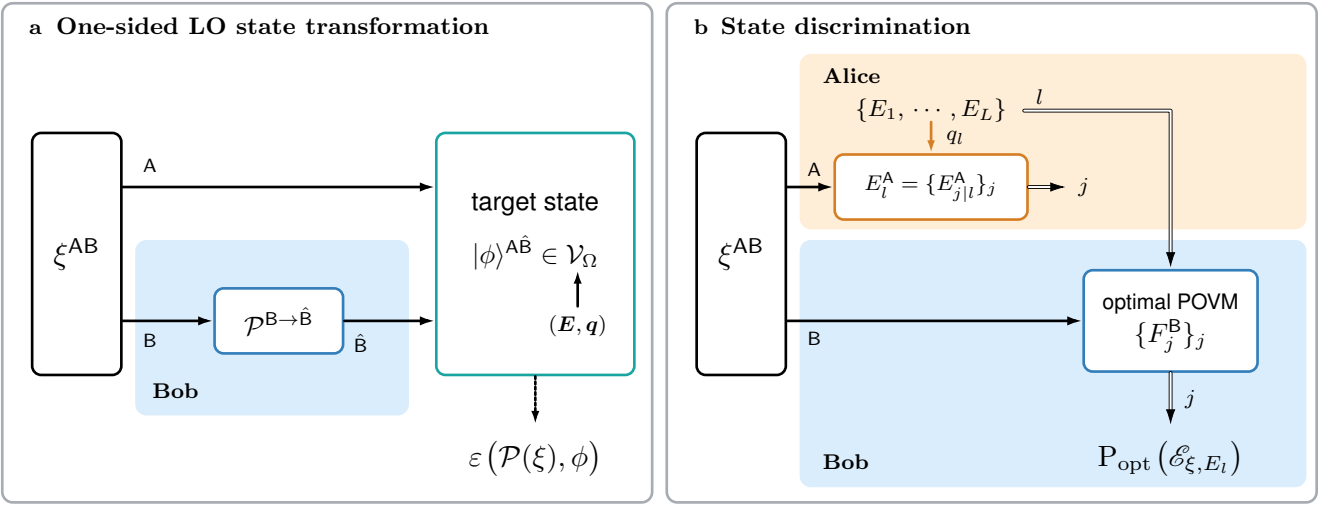
\begin{figure*}[t]
    \centering
    \begingroup
    \def\MainPaperEmbedding{1}
    \ifdefined\MainPaperEmbedding
\def\BeginMainTheoremSetupFigure{}
\def\EndMainTheoremSetupFigure{}
\else
\documentclass[multi=tikzpicture,border=1pt]{standalone}

\usepackage{fix-cm}
\usepackage{amsmath,bm,mathrsfs}
\usepackage{lmodern}
\usepackage[scaled=0.92]{helvet}
\usepackage{tikz}
\usetikzlibrary{arrows.meta,backgrounds,calc}
\DeclareFontFamily{U}{rsfs}{\skewchar\font127}
\DeclareFontShape{U}{rsfs}{m}{n}{<-> rsfs10}{}
\def\BeginMainTheoremSetupFigure{\begin{document}}
\def\EndMainTheoremSetupFigure{\end{document}}
\fi

\providecommand{\sA}{\mathsf{A}}
\providecommand{\sB}{\mathsf{B}}
\providecommand{\cP}{\mathcal{P}}
\providecommand{\cV}{\mathcal{V}}
\providecommand{\rP}{\mathrm{P}}
\providecommand{\scE}{\mathscr{E}}
\providecommand{\ket}[1]{\lvert #1\rangle}

\definecolor{panelborder}{RGB}{166,172,178}
\definecolor{alicefill}{RGB}{255,238,215}
\definecolor{aliceborder}{RGB}{214,126,37}
\definecolor{bobfill}{RGB}{218,237,252}
\definecolor{bobborder}{RGB}{51,126,185}
\definecolor{targetborder}{RGB}{26,164,161}

\newcommand{\MainTheoremSetupFigure}{%
  \def\panelbottom{0.15}%
  \def\stateheight{3.20}%
  \def\stateY{3.435}%
  \def\upperwireoffset{0.885}%
  \def\lowerwireoffset{-0.885}%
  \def\optimalPovmY{2.55}%
  \def\outputArrowY{1.18}%
  \def\resultLabelY{1.08}%
  \def\statefontsize{11.5}%
  \def\statebaselineskip{13}%
  \def\targetSourceDrop{0.48}%
  \def\bobRegionRight{5.57}%
  \def\bobRegionTop{3.62}%
  \def\targetCaptionOffset{0.66}%
  \def\alicePovmWidth{2.55}%
  \def\alicePovmFontSize{8.0}%
  \def\alicePovmBaselineSkip{9.2}%
  \def\aliceOutcomeArrowEnd{14.56}%
  \def\aliceOutcomeLabelX{14.62}%

\begin{tikzpicture}[
  x=0.966cm,
  y=1cm,
  >=Latex,
  font=\fontsize{8.2}{9.6}\selectfont\fontfamily{phv}\selectfont,
  line cap=round,
  line join=round,
  wire/.style={draw=black, line width=1pt, -{Latex[length=2.0mm,width=1.25mm]}},
  classical/.style={draw=black, double, double distance=0.75pt,
                    line width=0.35pt, -{Latex[length=2.0mm,width=1.25mm]}},
  shared/.style={classical},
  dependency/.style={draw=black, line width=1pt,
                     -{Latex[length=1.8mm,width=1.1mm]}},
  box/.style={draw=black, rounded corners=1.2mm, fill=white, line width=1pt,
              align=center, inner sep=2mm},
  panel/.style={draw=panelborder, rounded corners=1.4mm, fill=white,
                line width=1pt},
  smalllabel/.style={font=\fontsize{7.2}{8.2}\selectfont\fontfamily{phv}\selectfont},
  boxcaption/.style={font=\fontsize{8.0}{9.2}\selectfont\fontfamily{phv}\selectfont},
  targetcaption/.style={font=\fontsize{9.6}{10.8}\selectfont\fontfamily{phv}\selectfont},
  boxmath/.style={font=\fontsize{10}{11.5}\selectfont},
  resultmath/.style={font=\fontsize{11}{13}\selectfont}
]

\begin{scope}[on background layer]
  \path[panel] (0,\panelbottom) rectangle (8.90,6.75);
  \path[panel] (9.10,\panelbottom) rectangle (18.00,6.75);
  \path[fill=bobfill, rounded corners=1.5mm]
      (1.82,1.30) rectangle (\bobRegionRight,\bobRegionTop);
\end{scope}

\node[anchor=base west, font=\bfseries\fontsize{8}{9.2}\selectfont] at (0.25,6.35)
  {a};
\node[anchor=base west, font=\bfseries\fontsize{9}{10.5}\selectfont] at (0.62,6.35)
  {One-sided LO state transformation};
\node[anchor=base west, font=\bfseries\fontsize{8}{9.2}\selectfont] at (9.35,6.35)
  {b};
\node[anchor=base west, font=\bfseries\fontsize{9}{10.5}\selectfont] at (9.72,6.35)
  {State discrimination};

\node[box, minimum width=1.18cm, minimum height=\stateheight cm,
      font=\fontsize{\statefontsize}{\statebaselineskip}\selectfont]
      (xia) at (1.02,\stateY) {$\xi^{\sA\sB}$};
\coordinate (aAstart) at ($(xia.east)+(0,\upperwireoffset)$);
\coordinate (aBstart) at ($(xia.east)+(0,\lowerwireoffset)$);

\node[anchor=south west, font=\fontsize{8}{9}\selectfont\bfseries]
      at (1.94,1.42) {Bob};

\node[box, draw=targetborder, fill=white, minimum width=2.62cm,
      minimum height=\stateheight cm] (compare) at (7.30,\stateY) {};
\node[box, draw=bobborder, fill=white, minimum width=1.70cm,
      minimum height=0.86cm, inner sep=1.5mm, boxmath]
      (channel) at ($(xia.east)!0.5!(compare.west)+(0,\lowerwireoffset)$)
      {$\cP^{\sB\to\hat{\sB}}$};
\node[targetcaption] at ($(compare.center)+(0,\targetCaptionOffset)$)
      {target state};
\node[font=\fontsize{10}{11.5}\selectfont, inner sep=0pt]
      (targetformulaH) at ($(compare.center)+(0,0.02)$)
      {$\ket{\phi}^{\sA\hat{\sB}}\in\cV_\Omega$};
\node[font=\fontsize{8}{9}\selectfont, anchor=north]
      (EqSourceH) at ($(targetformulaH.south)+(0.58,-\targetSourceDrop)$)
      {$(\bm E,\bm q)$};
\draw[dependency] (EqSourceH.north) -- ($(targetformulaH.south)+(0.58,0)$);

\draw[wire] (aAstart) -- ($(compare.west)+(0,\upperwireoffset)$);
\draw[wire] (aBstart) -- (channel.west);
\draw[wire] (channel.east) -- ($(compare.west)+(0,\lowerwireoffset)$);
\node[smalllabel, anchor=south]
      at (2.02,{\stateY+\upperwireoffset+0.09}) {$\sA$};
\node[smalllabel, anchor=north]
      at (2.02,{\stateY+\lowerwireoffset-0.06}) {$\sB$};
\node[smalllabel, anchor=north]
      at ($(channel.east)+(0.22,-0.06)$) {$\hat{\sB}$};

\draw[densely dotted, line width=1pt,
      -{Latex[length=1.8mm,width=1.1mm]}]
      (compare.south) -- (7.30,\outputArrowY);
\node[resultmath, anchor=north] at (7.30,\resultLabelY)
      {$\varepsilon\mkern1.5mu\bigl(\cP(\xi),\phi\bigr)$};

\begin{scope}[on background layer]
  \path[fill=alicefill, rounded corners=1.5mm]
      (10.92,3.76) rectangle (17.86,6.08);
  \path[fill=bobfill, rounded corners=1.5mm]
      (10.92,0.35) rectangle (17.86,3.62);
\end{scope}

\node[anchor=north west, font=\fontsize{8}{9}\selectfont\bfseries]
      at (11.12,6.00) {Alice};
\node[anchor=south west, font=\fontsize{8}{9}\selectfont\bfseries]
      at (11.12,0.47) {Bob};

\node[box, minimum width=1.18cm, minimum height=\stateheight cm,
      font=\fontsize{\statefontsize}{\statebaselineskip}\selectfont]
      (xib) at (10.12,\stateY) {$\xi^{\sA\sB}$};
\coordinate (bAstart) at ($(xib.east)+(0,\upperwireoffset)$);
\coordinate (bBstart) at ($(xib.east)+(0,\lowerwireoffset)$);

\node[font=\fontsize{8.8}{10.0}\selectfont, opacity=0]
      (EfamilyAnchor) at (12.72,5.45) {$\{E_1,\,\cdots,E_L\}$};
\node[font=\fontsize{8.8}{10.0}\selectfont]
      (Efamily) at (12.72,5.33) {$\{E_1,\,\cdots,E_L\}$};
\node[box, draw=aliceborder, minimum width=\alicePovmWidth cm,
      minimum height=0.86cm, inner sep=1mm,
      font=\fontsize{\alicePovmFontSize}{\alicePovmBaselineSkip}\selectfont]
      (El) at (12.72,4.32) {$E_l^{\sA}=\{E_{j|l}^{\sA}\}_j$};
\draw[draw=aliceborder, line width=1pt,
      -{Latex[length=1.7mm,width=1.05mm]}]
      (EfamilyAnchor.south) -- node[font=\fontsize{8}{9}\selectfont,
      right, xshift=0.8mm] {$q_l$}
      (El.north);
\draw[wire] (bAstart) -- (El.west);
\node[smalllabel, anchor=south]
      at (11.12,{\stateY+\upperwireoffset+0.03}) {$\sA$};
\node[smalllabel, anchor=north]
      at (11.12,{\stateY+\lowerwireoffset-0.03}) {$\sB$};

\draw[classical] (El.east) -- (\aliceOutcomeArrowEnd,4.32);
\node[font=\fontsize{8}{9}\selectfont, anchor=west]
      at (\aliceOutcomeLabelX,4.32) {$j$};

\node[box, draw=bobborder, fill=white, minimum width=2.30cm,
      minimum height=1.34cm, inner sep=2.2mm]
      (opt) at (16.00,\optimalPovmY) {};
\node[boxcaption] at ($(opt.center)+(0,0.23)$) {optimal POVM};
\node[boxmath] at ($(opt.center)+(0,-0.20)$) {$\{F_j^{\sB}\}_j$};
\draw[wire] (bBstart) -- (opt.west);

\draw[classical] (opt.south) --
      node[font=\fontsize{8}{9}\selectfont, right, xshift=0.8mm] {$j$}
      (16.00,\outputArrowY);
\node[resultmath, anchor=north] at (16.00,\resultLabelY)
      {$\rP_{\mathrm{opt}}\mkern1.5mu\bigl(\scE_{\xi,E_l}\bigr)$};

\coordinate (lsharestartH) at ($(Efamily.east)+(0.12,0)$);
\draw[shared] (lsharestartH) -| ($(opt.north)+(0,0.02)$);
\node[font=\fontsize{8}{9}\selectfont, fill=alicefill, inner sep=1pt,
      anchor=south west] at ($(lsharestartH)+(0.12,0.04)$) {$l$};

\end{tikzpicture}%
}

\BeginMainTheoremSetupFigure
\MainTheoremSetupFigure
\EndMainTheoremSetupFigure
    \endgroup
    \caption{The two tasks on a shared state $\xi^{\sA\sB}$ connected by Theorem~\ref{thm:1}. \textbf{a}, In the one-sided LO state-transformation task, Bob applies a channel $\cP^{\sB\to\hat{\sB}}$ and compares the resulting state with a target state $\ket{\phi}^{\sA\hat{\sB}}\in\cV_\Omega$. \textbf{b}, In the state-discrimination task, Alice chooses a POVM $E_l^{\sA}$ with probability $q_l$, shares $l$ with Bob, and Bob performs an optimal POVM to guess Alice's outcome $j$. Here, $\cV_\Omega$ in \textbf{a} is determined by the family of Alice's POVMs $\bmE$ and the probability distribution $\bmq$ used in \textbf{b}.}
    \label{fig:main-theorem-setup}
\end{figure*}

\subsubsection{Quantitative connection between the two errors}
\label{sec:quantitative connection}

While our main goal is to clarify the connection between entanglement and classical distinguishability, we first consider two tasks on a possibly mixed state $\xi^{\sA\sB}$ shared between Alice and Bob: transforming $\xi^{\sA\sB}$ into a pure state by a one-sided local operation (LO), and discriminating the states induced on $\sB$ by Alice's measurement. See also~\cref{fig:main-theorem-setup}.

In the first task, Bob transforms $\xi^{\sA\sB}$ into a target state $\ket{\phi}^{\sA\hat{\sB}}$ via a quantum channel $\cP^{\sB\to\hat{\sB}}$, i.e., a completely positive and trace-preserving (CPTP) map on $\sB$.
We use the infidelity to quantify the error of this transformation:
\begin{equation}
\label{eq:def quantum error main}
    \varepsilon\big(\cP(\xi), \phi\big)
    \coloneqq 1 - \rF\big(\cP^{\sB\to\hat{\sB}}(\xi^{\sA\sB}), \ketbra{\phi}{\phi}^{\sA\hat{\sB}}\big),
\end{equation}
where $\rF(\rho,\sigma)=\big(\tr\sqrt{\sqrt{\sigma}\rho\sqrt{\sigma}}\big)^2$ is the Uhlmann fidelity.
The target of primary interest will be an MES, but we keep $\ket{\phi}$ general for now.

In the second task, Alice measures $\sA$ of the state $\xi^{\sA\sB}$, and Bob tries to guess the measurement outcome from $\sB$. More concretely, let $L$ be an arbitrary positive integer, let $\bmE=\{E_l^\sA\}_{l=1}^L$ denote a family of POVMs on $\sA$, and $\bmq = \{q_l\}_{l=1}^L$ be a probability distribution with $q_l>0$ for all $l$. Alice chooses a POVM  $E_l^\sA=\{E_{j|l}^{\sA}\}_{j=1}^{J_l}$ with probability $q_l$ and measures $\sA$ of $\xi^{\sA\sB}$, obtaining a measurement outcome $j$ with probability $p_{j|l} = \tr\big[E_{j|l}^\sA \xi^{\sA}\big]$. This measurement on $\sA$ induces an ensemble of states on the system $\sB$: 
\begin{equation}
    \scE_{\xi, E_l} \coloneqq \big\{p_{j|l}; \ \xi_{j|l}^\sB\big\}_{j=1}^{J_l},    
\end{equation}
where $\xi_{j|l}^\sB = \tr_\sA\big[E_{j|l}^\sA \xi^{\sA\sB}\big]/p_{j|l}$. 
From this ensemble, Bob tries to guess Alice's measurement outcome $j$. Assuming that Alice's choice of $l$ is shared with Bob, the optimal guessing probability is given by
\begin{align}
\label{eq:optimal guessing probability main}
    \rP_{\rm opt}(\scE_{\xi, E_l}) \coloneqq \max_{\{F_j^\sB\}} \sum_j p_{j|l} \tr[F_j^\sB \xi_{j|l}^\sB].
\end{align}
Here, the maximization is over all POVMs $\{F_j^\sB\}_j$.
This can be viewed as a state-discrimination task by Bob.

The LO state transformation and state discrimination are particularly fundamental in quantum information processing. The former underlies many important tasks, such as entanglement distillation and approximate QEC~\cite{devetak2005distillation, schumacher2001approximateerrorcorrection}, while the latter is central to numerous CQ tasks, including quantum cryptography~\cite{Bennett1992cryptographynonorthogonalstates, Dusek2000Unambiguousdiscriminationcryptography, Scarani2009securitypracticalqkd} and classical communication over quantum channels~\cite{Holevo1973BoundsQuantityInformationTransmitted, holevo1998capacityclassical, schumacher1997sendingclassicalinfo}. Their broader applications can be found in Ref.~\cite{khatri2024principlemodern}.

The two tasks are, a priori, unrelated except that both start from a common state $\xi^{\sA\sB}$. The LO state transformation depends on the channel $\cP$ and the target $\ket{\phi}$, while state discrimination depends on the family of measurements $\bmE$ and the probability distribution $\bmq$.
Nevertheless, for a suitable choice of the target state in the LO state transformation, these tasks are closely connected.

To clarify the connection, we introduce an operator $\Omega$ from the POVMs $\bmE$ and the probability distribution $\bmq$. Let $\hat{\sB} \cong \sA$.
For each $l$, define $\Omega_l^{\sA\hat{\sB}} \coloneqq \sum_j E_{j|l}^{\sA}\otimes \bar E_{j|l}^{\hat{\sB}}$.
Taking the average over $\bmq$, we define
\begin{align}
\label{eq:def of average of omega}
    \Omega^{\sA\hat{\sB}} \coloneqq \sum_{l=1}^L q_l \Omega_l^{\sA\hat{\sB}}.
\end{align}
This operator satisfies $0 \leq \Omega^{\sA\hat{\sB}} \leq \bI^{\sA\hat{\sB}}$, which implies $\|\Omega\|_\infty \leq 1$, where $\|\cdot\|_\infty$ is the operator norm defined by the largest singular value.
In the following, we refer to the eigenspace $\cV_\Omega$ of $\Omega^{\sA\hat{\sB}}$ with the largest eigenvalue $\|\Omega\|_\infty$ as the \emph{principal eigenspace}. 

The spectral gap of $\Omega^{\sA\hat{\sB}}$, which is the difference between the largest and the second-largest eigenvalues, is denoted by $\nu(\Omega)$. 
Here, we count eigenvalues with multiplicity. If the largest eigenvalue is degenerate, i.e., $\dim \cV_\Omega \geq 2$, then $\nu(\Omega) = 0$.

Our first result connecting the errors in LO state transformation and in state discrimination is as follows. The proof is given in~\cref{sec:proof thm1 new}.

\begin{theorem}
\label{thm:1}
Let $\ket{\phi}^{\sA\hat{\sB}} \in \cV_\Omega$ be a target state of the LO state transformation task.
Then, for any state $\xi^{\sA\sB}$, there exists a channel $\cP^{\sB\to\hat{\sB}}$ such that
\begin{align}
\label{eq:equation of main thm1}
    \nu(\Omega) \varepsilon(\cP(\xi), \phi) 
    \leq \|\Omega\|_\infty - \big\lag \rP_{\rm opt}^2(\scE_{\xi, E})\big\rag,
\end{align}
where $\big\lag \rP_{\rm opt}^2(\scE_{\xi, E}) \big\rag \coloneqq \sum_{l=1}^L q_l \left[\rP_{\rm opt}(\scE_{\xi, E_l})\right]^2$.
\end{theorem}

The result is nontrivial when $\nu(\Omega) \neq 0$, or equivalently, when the principal eigenspace $\cV_\Omega$ is one-dimensional. 
In this case, the target state $\ket{\phi}^{\sA\hat{\sB}}$ is uniquely determined, and dividing~\cref{eq:equation of main thm1} by $\nu(\Omega)$ yields an upper bound on the error of the LO state transformation in terms of the optimal guessing probabilities. Successful state discrimination implies a reliable transformation of $\xi^{\sA\sB}$ into $\ket{\phi}^{\sA\hat{\sB}}$, revealing a connection between the two tasks.

Notably, we explicitly provide a channel achieving~\cref{eq:equation of main thm1}. The channel is called the \emph{pretty good recovery map}~\cite{Berta2014Entanglementassistedguessingcomplementary}, which is defined for a given state $\xi^{\sA\sB}$ as 
\begin{align}
\label{eq:state-based petz map main}
    \cP_\xi^{\sB\to\hat{\sB}}(\cdot) 
    \coloneqq \tr_\sB\big[(\xi^{\hat{\sB}\sB})^{\top_{\hat{\sB}}}(\xi^\sB)^{-1/2}(\cdot)(\xi^\sB)^{-1/2}\big],
\end{align}
where $\xi^{\hat{\sB}\sB}$ is $\xi^{\sA\sB}$ with $\sA$ relabeled as $\hat{\sB} \cong \sA$, and $\top_{\hat{\sB}}$ denotes the partial transpose on $\hat{\sB}$.
This is CPTP on the support of $\xi^\sB = \tr_\sA[\xi^{\sA\sB}]$ and reduces to the well-known Petz map~\cite{petz1986sufficient, petz1988sufficiency} upon expressing $\xi^{\hat{\sB}\sB}$ as the output of a channel acting on a pure state. 
See~\cref{sec:proof thm1 new} for details.

A converse-type statement to~\cref{thm:1}, showing that reliable one-sided LO state transformation bounds the guessing probabilities from below, can also be shown. See~\cref{sec:derivation of converse relation} for the proof.

\begin{proposition}
\label{prop:converse relation}
For any state $\xi^{\sA\sB}$, we have 
\begin{align}
\label{eq:converse relation}
    \big\lag \rP_{\rm opt}(\scE_{\xi, E})\big\rag
    \geq \|\Omega\|_\infty\big(1 - \min_{\cP, \ket{\phi}}\varepsilon(\cP(\xi), \phi)\big),
\end{align}
where $\big\lag \rP_{\rm opt}(\scE_{\xi, E})\big\rag \coloneqq \sum_{l=1}^L q_l \rP_{\rm opt}(\scE_{\xi, E_l})$.
The minimization is taken over all channels $\cP^{\sB \to \hat{\sB}}$ and all states $\ket{\phi}^{\sA\hat{\sB}}$ in the principal eigenspace $\cV_\Omega$ of $\Omega^{\sA\hat{\sB}}$.
\end{proposition}


\subsubsection{Irreducibility and maximally entangled targets}
\label{sec:irreducibility maximal entanglement}

\cref{thm:1} is applicable when the target state $\ket{\phi}^{\sA\hat{\sB}} \in \cV_\Omega$, that is, it is an eigenstate of $\Omega^{\sA\hat{\sB}}$ with the largest eigenvalue. As our primary interest is an MES, we next identify when $\ket{\phi}^{\sA\hat{\sB}}$ is an MES.

To this end, we consider the action of $\Omega^{\sA\hat{\sB}}$ on an MES $\ket{\Phi}^{\sA\hat{\sB}}$. This action can be transferred to Alice’s system alone, as $\Omega^{\sA\hat{\sB}}\ket{\Phi}^{\sA\hat{\sB}} = W^\sA\ket{\Phi}^{\sA\hat{\sB}}$, where
\begin{equation}
    W^\sA \coloneqq \sum_{j, l} q_l \big(E_{j|l}^{\sA}\big)^2.
\end{equation}
This positive operator $W^{\sA}$ satisfies $\|\Omega\|_\infty \leq \|W^{\sA}\|_\infty \leq 1$ (see~\cref{eq:omega w relation}), and as we will see, the saturation of the first inequality is key to identifying a maximally entangled target. Below, the principal eigenspace of $W^\sA$ is denoted by $\cV_W$.

We also introduce algebraic notions associated with a general family of POVMs. 

\begin{definition}[$\bmE$-invariance, $\bmE$-irreducibility]
\label{def:reducibility}
Let $\bmE = \{E_l = \{E_{j|l}\}_j\}_{l=1}^L$ be a family of POVMs acting on a Hilbert space $\cH$.
A subspace $\cK \subseteq \cH$ is $\bmE$-invariant if $E_{j|l} \cK \subseteq \cK$ for all $j$ and $l$. 
Moreover, a nonzero subspace $\cK \subseteq \cH$ is $\bmE$-irreducible if it is $\bmE$-invariant and contains no $\bmE$-invariant subspaces other than $\{0\}$ and $\cK$ itself.
\end{definition}

Since $E_{j|l}$ are Hermitian, the orthogonal complement of any $\bmE$-invariant subspace is again $\bmE$-invariant.
Hence,~\cref{def:reducibility} has the equivalent formulation: a nonzero $\bmE$-invariant subspace $\cK$ is $\bmE$-irreducible if and only if it admits no orthogonal decomposition $\cK=\bigoplus_a\cR_a$ into two or more nonzero $\bmE$-invariant subspaces.

In algebraic terms, the irreducibility in~\cref{def:reducibility} is that of an invariant subspace $\cK$ under the $*$-algebra generated by the POVM elements $\{E_{j|l}\}_{j,l}$, equivalently, the triviality of the corresponding commutant on $\cK$. 
Related structures appear in QEC and measurement theory: the noise operator algebra and its commutant for noiseless subsystems~\cite{Knill2000QuantumErrorCorrectionGeneralNoise, Holbrook2003NoiselessSubsystemsStructureCommutant}, and the irreducibility of a pure-state ensemble for information-disturbance tradeoffs~\cite{Barnum2001reversibleextractionclassicalinformation, Buscemi2009TowardsUnifiedApproach}.

The following theorem gives a condition under which the target state $\ket{\phi}^{\sA\hat{\sB}}$ in \cref{thm:1} is an MES, in terms of $\Omega^{\sA\hat{\sB}}$ and $W^{\sA}$ as well as the algebraic structures of $\bmE$.

\begin{theorem}
\label{prop:qualitative equivalent condition}
For a family of POVMs $\bmE$ and a probability distribution $\bmq$, the following are equivalent:
\begin{itemize}
\item[(I)] The principal eigenspace $\cV_W$ of $W$ contains exactly one $\bmE$-irreducible subspace.
\item[(II)] $\|\Omega\|_\infty = \|W\|_\infty$ and $\nu(\Omega) \neq 0$.
\end{itemize}
Under these conditions, the principal eigenspace $\cV_\Omega$ of $\Omega^{\sA\hat{\sB}}$ is one-dimensional and is spanned by
\begin{align}
\label{eq:subspace mes on k star}
    \ket{\Phi_\star}^{\sA\hat{\sB}} = \f{1}{\sqrt{d_{\cK_\star}}}\sum_j \ket{u_j}^\sA\ket{\bar{u}_j}^{\hat{\sB}},
\end{align}
where $\cK_{\star}$ is the unique $\bmE$-irreducible subspace in $\cV_W$, $\{\ket{u_j}^\sA\}_j$ is any orthonormal basis in $\cK_\star$, and $d_{\cK_\star}$ is the dimension of $\cK_\star$.
\end{theorem}

Note that $\cK_\star$, and thus its dimension $d_{\cK_\star}$, depends not only on the choice of $\bmE$ but also on that of $\bmq$, through $\bmq$ in $\cV_W$. Such dependence does not arise for projection-valued measures (PVMs), where $W = \bI$. The proof of \Cref{prop:qualitative equivalent condition} relies on quantum Perron--Frobenius-type arguments for irreducible maps~\cite{evans1978SpectralPropertiesPositiveMaps, Albeverio1978Frobeniustheorypositivemaps, wolf2012guidedtour}. See~\cref{sec:proof of irreducibility and unique state spectrum equivalence} for the proof and an example of this dependence.

As a direct consequence of~\cref{thm:1,prop:qualitative equivalent condition,prop:converse relation}, we obtain the following corollary.
\begin{widetext}
\begin{corollary}
\label{cor:1}
Suppose the principal eigenspace $\cV_W$ contains exactly one $\bmE$-irreducible subspace $\cK_\star$, or equivalently, $\|\Omega\|_\infty = \|W\|_\infty$ and $\nu(\Omega) \neq 0$.
Then, for any state $\xi^{\sA\sB}$, the pretty good recovery map $\cP_\xi^{\sB\to\hat{\sB}}$ defined by~\cref{eq:state-based petz map main} satisfies
\begin{align}
\label{eq:corollary bound}
    1 - \f{\big\lag \rP_{\rm opt}(\scE_{\xi, E}) \big\rag}{\|W\|_\infty} 
    \leq \varepsilon\big(\cP_\xi(\xi), \Phi_\star\big)
    \leq \frac{\|W\|_\infty - \big\lag \rP_{\rm opt}^2(\scE_{\xi, E}) \big\rag}{\nu(\Omega)},
\end{align}
where $\ket{\Phi_\star}^{\sA\hat{\sB}}$ is the MES on $\cK_\star\otimes\bar{\cK}_\star$, which carries $\log d_{\cK_\star}$ ebits of entanglement.
\end{corollary}
\end{widetext}

\cref{cor:1} bounds the error of extracting entanglement from a given state $\xi^{\sA\sB}$ by a one-sided LO in terms of the guessing probabilities in state discrimination. First, Alice chooses $\bmE$ and $\bmq$ such that the principal eigenspace $\cV_W$ contains a unique $\bmE$-irreducible subspace $\cK_\star$, preferably as large as possible for more entanglement to be certified. 
Once $\bmE$ and $\bmq$ are fixed, Alice measures $\sA$ using $E_l$ with probability $q_l$, and then, Bob performs the corresponding state-discrimination task on $\sB$. 
Due to~\cref{eq:corollary bound}, the guessing probabilities bound the infidelity of extracting $\log d_{\cK_\star}$ ebits via the pretty good recovery map.
Thus, high guessing probabilities ensure reliable extraction of entanglement, and vice versa up to the factor $\|W\|_\infty$.

For given $\bmE$ and $\bmq$, whether $\cV_W$ contains exactly one $\bmE$-irreducible subspace can be verified using condition~(II) of~\cref{prop:qualitative equivalent condition}.
Its first equality, $\|\Omega\|_\infty = \|W\|_\infty$, requires the inequality $\|\Omega\|_\infty \leq \|W\|_\infty$ to be saturated. This is the case for many measurements of interest, including all PVMs and symmetric POVMs such as SIC-POVMs~\cite{Renes2004Symmetricinformationallycomplete} and trine POVMs~\cite{Barnett2009Quantumstatediscrimination}, for which $W \propto \bI$ and hence $\cV_W = \cH^\sA$. It then remains to check whether $\nu(\Omega) \neq 0$. 
For example, this is satisfied by measurements in MUBs. In contrast, $\nu(\Omega) = 0$ for the measurements in the bases $\{\ket{1}, \ket{2}, \ket{3}\}$ and $\big\{\f{1}{\sqrt{2}}(\ket{1}+\ket{2}), \f{1}{\sqrt{2}}(\ket{1}-\ket{2}), \ket{3}\big\}$ ($L=2$), which share $\ket{3}$ and are block-diagonal in $\rspan\{\ket{1}, \ket{2}\} \oplus \rspan\{\ket{3}\}$. Consequently, $\cV_W = \cH^\sA$ contains two $\bmE$-irreducible subspaces, $\rspan\{\ket{1}, \ket{2}\}$ and $\rspan\{\ket{3}\}$.

When $\|\Omega\|_\infty = \|W\|_\infty$ but $\nu(\Omega) = 0$, the target state $\ket{\phi}^{\sA \hat{\sB}}$ is not uniquely determined as an MES. In this case, the largest $\bmE$-invariant subspace $\cK$ in $\cV_W$ is the direct sum $\cK = \bigoplus_a \cR_a$ of two or more $\bmE$-irreducible subspaces $\cR_a$. All POVM elements $E_{j|l}$ are block-diagonal with respect to $\{\cR_a\}_a$. Such measurements induce the same ensemble $\scE_{\xi, E_l}$ for the MES $\ket{\Phi_\star} = \sum_a \sqrt{d_{\cR_a}/d_{\cK}}\ket{\Phi_a}$ and for the incoherent mixture $\sum_a (d_{\cR_a}/d_{\cK})\ketbra{\Phi_a}{\Phi_a}$ of the block MESs $\ket{\Phi_a}$. Even perfect state discrimination cannot tell the two apart, so entanglement across the blocks cannot be certified.

Nevertheless, our results remain informative blockwise even in this case. 
Alice first projects onto the irreducible blocks $\cR_a$ via the measurement $\{\Pi_a, \bI - \Pi_a\}$, obtaining the conditional state $\xi_a = (\Pi_a \otimes \bI)\xi(\Pi_a \otimes \bI)/r_a$ with probability $r_a = \tr[\Pi_a \xi]$.
Since each $\cR_a$ is $\bmE$-invariant, the restricted operators $\big\{E_{j|l}\big|_{\cR_a}\big\}_j$ form a POVM on $\cR_a$, and the spectral gap $\nu(\Omega_a)$ of $\Omega_a = \sum_{j, l} q_l E_{j|l}\big|_{\cR_a} \otimes \bar{E}_{j|l}\big|_{\bar{\cR}_a}$ must be nonzero. 
For the block-local MES $\ket{\Phi_a}$,~\cref{thm:1} implies that
\begin{align}
\label{eq:block bound}
    \varepsilon(\cP_{\xi_a}(\xi_a), \Phi_a) \leq \f{\|\Omega_a\|_\infty - \big\lag \rP_{\rm opt}^2(\scE_{\xi_a, E})\big\rag}{\nu(\Omega_a)}.
\end{align}
Although a global MES $\ket{\Phi_\star}$ cannot be certified when $\nu(\Omega) = 0$, each block MES $\ket{\Phi_a}$ can still be certified probabilistically.

In~\cref{sec:construction povm from given state}, we consider the inverse direction, where the target state is fixed first and then compatible POVMs are constructed. This is possible for a bipartite pure state with any Schmidt coefficients, extending our result beyond the MES. In~\cref{sec:equivalent conditions positive spectral gap}, we give \emph{quantitative} bounds on the spectral gap in terms of the structure of the underlying Hilbert space.


\subsubsection{Simultaneous vanishing of the two errors}

\label{sec:vanishing errors}

By combining the results so far, we can characterize the regime in which both the errors in local transformation to the MES and in state discrimination vanish simultaneously. The derivation is in~\cref{sec:vanishing error equivalence}.

\begin{corollary}
\label{thm:vanishing error equivalence}
Suppose that the principal eigenspace $\cV_W$ of $W^\sA$ contains exactly one $\bmE$-irreducible subspace $\cK_\star$.
Then, the following are equivalent:
\begin{itemize}
    \item[(i)] $\|W^\sA\|_\infty = 1$.
    \item[(ii)] For every $l=1,\dots,L$, $\big\{E_{j|l}^\sA\big|_{\cK_\star}\big\}_j$ forms a PVM on $\cK_\star$.
    \item[(iii)] For every system $\sB$ and every state $\xi^{\sA\sB}$, 
    \begin{align}
    \label{eq:vanish equivalence equation}
    \varepsilon_{\rm opt}^{\Phi_\star}(\xi) = 0  \,\Longrightarrow\,
    \rP_{\rm opt}(\scE_{\xi,E_l}) = 1 \ \text{for all } l = 1, \dots, L.
    \end{align}
\end{itemize}
Here, $\big|_{\cK_\star}$ in (ii) denotes the restriction to $\cK_\star$, $\ket{\Phi_\star}^{\sA\hat{\sB}}$ is the MES on $\cK_\star \otimes \bar{\cK}_\star$, and $\varepsilon_{\rm opt}^{\Phi_\star}(\xi) \coloneqq \min_{\cP}\varepsilon\big(\cP(\xi),\Phi_\star\big)$, where the minimization is taken over all channels $\cP^{\sB\to\hat{\sB}}$.
\end{corollary}

Thus, under the condition that $\cV_W$ contains exactly one $\bmE$-irreducible subspace $\cK_\star$, the connection between entanglement and state distinguishability is characterized by how the POVMs in $\bmE$ act on $\cK_\star$. As shown in (ii) and (iii), the implication in~\cref{eq:vanish equivalence equation} holds when all the POVMs are projective on $\cK_\star$. The converse of~\cref{eq:vanish equivalence equation} holds regardless of conditions~(i) and~(ii), since $\rP_{\rm opt}(\scE_{\xi, E_l}) = 1$ for all $l$ makes the right-hand side of~\cref{eq:equation of main thm1} nonpositive, and $\nu(\Omega) \neq 0$ then forces $\varepsilon_{\rm opt}^{\Phi_\star}(\xi) = 0$. The two errors therefore vanish together, for every state $\xi^{\sA\sB}$, if and only if the POVMs are projective on $\cK_\star$.

We illustrate these conditions with two examples, one satisfying (i) and one violating it. First, consider binary projective measurements $E_l = \{\ketbra{e_l}{e_l},\, \bI - \ketbra{e_l}{e_l}\}$ for $l=1, 2, 3$, where $\ket{e_l} = \f{1}{\sqrt{2}}(\ket{l}+\ket{l+1})$ and $\ket{4}=\ket{1}$. They act irreducibly on a three-dimensional space, which constitutes the unique $\bmE$-irreducible subspace in $\cV_W$, with $\|W\|_\infty = 1$.
\cref{thm:vanishing error equivalence} implies that $\varepsilon_{\rm opt}^{\Phi_\star}(\xi) = 0$ if and only if all $\rP_{\rm opt}(\scE_{\xi, E_l}) = 1$.
By contrast, the trine POVM~\cite{Barnett2009Quantumstatediscrimination} admits a unique two-dimensional $\bmE$-irreducible subspace even for $L=1$, but violates~(i) since $W = 2\bI/3$. Indeed, at $\xi = \ketbra{\Phi_\star}{\Phi_\star}$ one achieves $\varepsilon_{\rm opt}^{\Phi_\star}(\xi) = 0$, while $\rP_{\rm opt}(\scE_{\xi, E})$ never exceeds $2/3$ (see also~\cref{eq:upper bound guess prob} for a general upper bound). Ideal local conversion to the MES no longer implies perfect discrimination.

The following is a direct consequence of \cref{thm:vanishing error equivalence}, specialized to the case $L=1$. The derivation is given in~\cref{sec:vanishing error equivalence}.

\begin{corollary}
\label{prop:L_1_trivial}
Suppose that the principal eigenspace $\cV_W$ contains exactly one $\bmE$-irreducible subspace $\cK_\star$.
If $\|W\|_\infty = 1$ and $L = 1$, then $d_{\cK_\star} = 1$; that is, $\ket{\Phi_\star}^{\sA\hat{\sB}}$ is a product state. Consequently, for $d_{\cK_\star} \geq 2$, the implication in~\cref{eq:vanish equivalence equation} requires $L \geq 2$.
\end{corollary}

\cref{prop:L_1_trivial} shows that simultaneous vanishing of the two errors for an entangled target requires at least two POVMs: for $L = 1$, conditions~(i)--(iii) of \cref{thm:vanishing error equivalence} can hold only in the trivial case $d_{\cK_\star} = 1$, where $\ket{\Phi_\star}^{\sA\hat{\sB}}$ is not entangled.

Before concluding this section, we discuss related literature. Quantitative correspondences between entanglement recovery and state discrimination have been established for complementary measurements~\cite{koashi2007complementaritydistillablesecretkey, Renes2013thephysics, renes2016uncertainrelationsandAQEC, Renes2022QuantumInformationCM, Berta2014Entanglementassistedguessingcomplementary}. Nakata et al.~\cite{nakata2025constructing} extended them to any two bases by constructing a decoder for general QEC, yet their error bound remains nonzero for non-MUBs even when the classical decoding errors vanish. Entropic uncertainty relations lower-bound the distillable entanglement for any two bases but certify maximal entanglement only for MUBs~\cite{Berta2010UncertaintyPrinciplequantummemory, Bergh2021ExperimentallyAccessibleBoundsDistillableEntanglement}. The zero-error case beyond complementarity has been explored using irreducible families of PVMs~\cite{Zhu2021ZeroUncertaintyStates, ran2026zerouncertaintystatesrelativeobservable}, and error-tolerant versions based on quantum steering are known for specific MUB-based measurements~\cite{Sarkar2022Certificationincompatiblemeasurements, Orthey2025Certifyingclassesmeasurements}. \Cref{prop:qualitative equivalent condition,thm:vanishing error equivalence} contain the equivalences of Refs.~\cite{Zhu2021ZeroUncertaintyStates, ran2026zerouncertaintystatesrelativeobservable} as special cases, while our bounds apply to arbitrary states and to families of POVMs satisfying the condition in~\cref{prop:qualitative equivalent condition}. In particular, the zero-error equivalence extends to POVMs precisely when they are projective on the irreducible subspace $\cK_\star$.


\subsection{Applications in quantum information processing}
\label{sec:Applications in quantum information processing}

An immediate application of our results is entanglement certification~\cite{Pallister2018OptimalVerificationEntangledStates, Zhu2019optimalverificationfidelityestimation, Li2019efficientverification, Yu2022StatisticalMethodsStateVerification, Li2025HighDimensionalEntanglementwitnessed, Li2026QuantumStateVerification}, which has broad theoretical and practical significance. Here, we explain the implications of our results in related tasks.

In~\cref{SSS:Verifyin_Entanglement}, we provide a simple LOCC protocol for checking the amount of distillable entanglement that works even when the available operations are restricted. We demonstrate in~\cref{SSS:demonstration_ent_distill} the strength of the protocol in concrete settings. In~\cref{SSS:QEC}, a similar approach is used to guarantee the performance of quantum error-correcting codes (QECCs).

\subsubsection{Certifying entanglement distillation by available measurements}
\label{SSS:Verifyin_Entanglement}

Given a shared bipartite state $\xi^{\sA\sB}$, Alice and Bob aim to determine how much entanglement can be distilled via LOCC. A naive approach is to use LOCC tomography, but it may be experimentally demanding and computationally intractable. Since current quantum platforms can implement only a restricted class of operations, it is practically important to construct a certification protocol that accommodates experimental limitations~\cite{Schneeloch2018Quantifyinghighdimensionalentanglement, ArnonFriedman2019DeviceIndependentCertification, Bergh2021ExperimentallyAccessibleBoundsDistillableEntanglement}.

Based on our results, we construct a simple LOCC protocol that certifies a distillable amount of entanglement even under restricted operations. The protocol is specified by a family of Alice's POVMs $\bmE=\{E_l^{\sA}\}_{l=1}^L$, a family of Bob's POVMs $\bmF=\{F_l^{\sB}\}_{l=1}^L$, and a probability distribution $\bmq = \{q_l\}_{l=1}^L$ over the choice of POVM. These choices can reflect possible experimental restrictions. 
Using $\bmE$ and $\bmq$, we have
\begin{equation}
    \Omega^{\sA\hat{\sB}} = \sum_{j,l} q_l E^{\sA}_{j|l} \otimes \bar{E}^{\hat{\sB}}_{j|l}, \ \ \ W^\sA = \sum_{j,l} q_l \big(E^{\sA}_{j|l}\big)^2.
\end{equation}
The only condition we require is that the principal eigenspace $\cV_W$ of $W$ contains exactly one $\bmE$-irreducible subspace, denoted by $\cK_\star$. By~\cref{prop:qualitative equivalent condition}, this is equivalent to $\|\Omega\|_\infty = \| W \|_{\infty}$ and $\nu(\Omega)\neq 0$.

Alice and Bob measure $\sA$ and $\sB$ of $\xi^{\sA\sB}$ by POVMs $E_l^{\sA} = \{ E_{j|l}^{\sA} \}_{j=1}^{J_l}$ and $F_l^{\sB} = \{F_{j|l}^{\sB} \}_{j=1}^{J_l}$, respectively, with probability $q_l$, and check if their outcomes coincide. By repeating this, they evaluate
\begin{equation}
    C_{\mathrm{coin}}(\bmE, \bmF | \xi)  = \sum_{l=1}^L q_l p_l^2,
\end{equation} 
where $p_l = \sum_{j=1}^{J_l} \tr\big[(E_{j|l}^{\sA} \otimes F_{j|l}^{\sB}) \xi^{\sA\sB}\big]$ is the coincidence probability. Note that $C_{\mathrm{coin}}(\bmE, \bmF | \xi)\leq \lag \rP_{\rm opt}^2(\scE_{\xi, E})\rag$ as Bob's POVMs are fixed in this protocol.
Following~\cref{cor:1}, they can certify that $\log d_{\cK_\star}$ ebits can be distilled with infidelity $\varepsilon_{\mathrm{distill}}$ bounded as
\begin{align}
\label{Eq:upper_e_distill}
    \varepsilon_{\mathrm{distill}} \leq \f{\|\Omega\|_\infty - C_{\mathrm{coin}}(\bmE, \bmF | \xi)}{\nu(\Omega)}. 
\end{align}

This idea is, however, somewhat idealized, because exact evaluation of $p_l$ requires infinitely many copies of $\xi^{\sA\sB}$. This leads us to propose~\cref{Protocol:LOCC_Ent_Check} for estimating $C_{\mathrm{coin}}(\bmE, \bmF | \xi)$ from finite copies of the state. It is a particularly simple LOCC protocol, based on fixed non-adaptive local measurements and classical communication to check the coincidence of their outcomes. 

\begin{figure}[t]
\begin{protocolbox}{Protocol:LOCC_Ent_Check}{Checking Distillable Entanglement}
    \begin{algorithmic}[1]
        \Statex \textbf{Given:} 
        \begin{compactenum}
            \item Copies of a state $\xi^{\sA\sB}$ shared between Alice $\sA$ and Bob $\sB$,
            \item a family of Alice's POVMs $\bmE=\{E_l^\sA\}_{l=1}^L$,
            \item a family of Bob's POVMs $\bmF=\{F_l^\sB\}_{l=1}^L$,
            \item a probability distribution $\bmq = \{q_l\}_{l=1}^L$,
            \item an accuracy parameter $\kappa > 0$ and a failure probability $\delta \in (0,1)$.
        \end{compactenum}
        \Statex \textbf{Condition:} The principal eigenspace $\cV_W$ of $W$ contains exactly one $\bmE$-irreducible subspace $\cK_\star$.
        \Statex \textbf{Output:}  An estimate of $C_{\mathrm{coin}}(\bmE, \bmF | \xi)$.
        \Statex
        
        \State Set $m = \lceil \f{1}{2 \kappa^2} \log \f{2}{\delta} \rceil$.
        \For{$t=1,\ldots,m$}
            \StateP{Alice samples $l$ with probability $q_l$ and tells it to Bob.}
            \StateP{Alice measures her subsystems of the two copies of $\xi^{\sA\sB}$ using the POVM $E_l^{\sA} = \{E_{j|l}^\sA\}_{j=1}^{J_l}$ and obtains outcomes $(j_{a_1}, j_{a_2})$.}
            \StateP{Bob measures his subsystems of the two copies using the POVM $F_l^\sB= \{F_{j|l}^\sB\}_{j=1}^{J_l}$ and obtains outcomes $(j_{b_1}, j_{b_2})$.}
            \StateP{Alice and Bob classically communicate. If $j_{a_1}=j_{b_1}$ and $j_{a_2}=j_{b_2}$, set $c_t = 1$. Otherwise, set $c_t=0$.}
        \EndFor
        \State Output $\widehat C_{\rm coin} = \f{1}{m}\sum_{t=1}^m c_t$.
    \end{algorithmic}
\end{protocolbox}
\end{figure}

 Since $\widehat{C}_{\mathrm{coin}}$ in \cref{Protocol:LOCC_Ent_Check} is unbiased, i.e., $\bE\big[\widehat{C}_{\mathrm{coin}}\big] = C_{\mathrm{coin}}(\bmE,\bmF | \xi)$, and $c_t \in \{0,1\}$, the Hoeffding inequality~\cite{Hoeffding1963ProbabilityInequalitiesSums} guarantees that the output of~\cref{Protocol:LOCC_Ent_Check}, $\widehat{C}_{\mathrm{coin}}$, is $\kappa$-close to $C_{\mathrm{coin}}(\bmE,\bmF | \xi)$ except with probability at most $\delta$. Importantly, this is achieved by the number of copies of $\xi^{\sA \sB}$ that is dependent only on the accuracy parameter $\kappa$ and the failure probability $\delta$, but is completely independent of the state $\xi$ and the number of POVMs $L$.

From~\Cref{cor:1}, \Cref{Protocol:LOCC_Ent_Check} provides a one-shot entanglement-distillation guarantee and a concrete LOCC procedure.

\begin{corollary}[Certifying LOCC distillation of $\log{d_{\cK_\star}}$ ebits]
\label{cor:entanglement_distillation}
Let $\kappa>0$ and $\delta \in (0,1)$, and let $\widehat{C}_{\mathrm{coin}}$ be the outcome of~\cref{Protocol:LOCC_Ent_Check}, which uses $\cO(\kappa^{-2} \log \delta^{-1})$ copies of $\xi^{\sA\sB}$. Then, the LOCC entanglement-distillation protocol described below transforms $\xi^{\sA\sB}$ into an MES carrying $\log{d_{\cK_\star}}$ ebits with infidelity $\varepsilon_{\mathrm{distill}}$ satisfying, with probability at least $1-\delta$,
\begin{align}
    \varepsilon_{\mathrm{distill}} 
    \leq \f{\|\Omega\|_\infty - \widehat{C}_{\mathrm{coin}} + \kappa}{\nu(\Omega)}.
\end{align}
\end{corollary}

The distillation protocol in~\cref{cor:entanglement_distillation} is one-way and non-adaptive. Alice applies a local unitary $U$ that transforms a basis of $\cK_\star$ to the computational basis, and tells Bob what unitary she has applied. Bob then applies the pretty good recovery map $\cP_\xi$ in Eq.~\eqref{eq:state-based petz map main} followed by $\bar{U}$.

While~\cref{Protocol:LOCC_Ent_Check} can be executed with any choice of POVMs, $\bmE$ and $\bmF$, suitable choices would be important to obtain a large coincidence of outcomes. From a theoretical point of view, the optimal error bound is obtained by optimizing the choices of POVMs. However,~\cref{cor:entanglement_distillation} holds for any POVMs, as long as Alice's choice satisfies the uniqueness condition on the irreducible subspace. This flexibility is one of the practical strengths of our approach.


\subsubsection{Demonstration: send-and-measure protocol}  \label{SSS:demonstration_ent_distill}


\begin{figure*}[t]
    {\centering
    \includegraphics[width=1.0\textwidth]{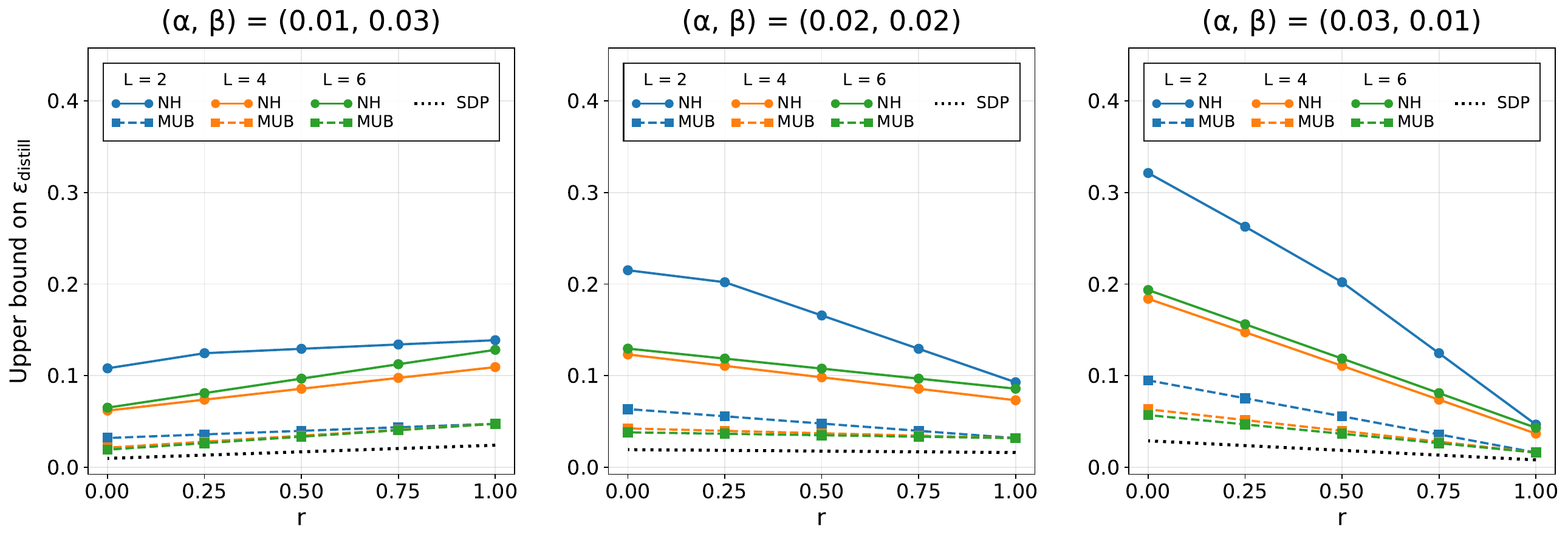}
    }
    \caption{Error bounds for $\varepsilon_{\mathrm{distill}}$ on distilling an MES from the noisy state $\cN^{\sA'\to\sB} (\Phi^{\sA\sA'})$, where the noise $\cN^{\sA'\to\sB}$ is a probabilistic mixture of depolarizing and dephasing channels with noise parameters $\alpha$ and $\beta$, respectively. The mixing probabilities of the former and latter channels are $1-r$ and $r$, respectively. We set $d=5$. Dashed and solid lines correspond to MUB-based measurements and neighboring-level Hadamard (NH)-based measurements, respectively. In both cases, $L$ denotes the number of measurement bases, optimized over the choice of $L$ MUB- or NH-based bases. As a benchmark, we also show, with the black dotted lines, the error $\varepsilon_{\mathrm{distill}}$ achieved by the optimal local operation on Bob's system, which is computed by a semidefinite program (SDP). At $L = 4$, the NH-based bound remains within a factor of $4.6$--$6.4$ of this optimal error and within $2.3$--$2.9$ of the corresponding MUB-based bound. See~\cref{sec:NN-hadamard simulation} for details of the setup.}
    \label{fig:numerical NNH bases simulation}
\end{figure*}


We demonstrate the performance of our approach in a simple entanglement-distribution setting, where Alice prepares an MES $\ket{\Phi}^{\sA\sA'}$ and sends $\sA'$  to Bob through an unknown noisy channel $\cN^{\sA'\to\sB}$.
We set $\xi^{\sA\sB} = \cN^{\sA'\to\sB} (\Phi^{\sA\sA'})$ and apply \cref{cor:entanglement_distillation}.

When the measurements are rank-1 projective, i.e., measurements in orthonormal bases, the protocol admits a simpler send-and-measure implementation. Since Alice's only action on $\sA$ is the basis measurement, she can perform it before the transmission; this is equivalent to directly sending a uniformly random complex-conjugate basis state instead of preparing an MES.
Bob then measures the received state and tells his outcome to Alice. This simplified send-and-measure protocol can notably reduce the experimental overhead. Although Bob's measurements are arbitrary, using the complex conjugates of Alice's bases is natural in the low-noise regime, as this choice is optimal in the noiseless case.

In this setting, $W = \bI$ as the measurements are rank-1 PVMs, implying that $\|W\|_\infty = 1$ and $\cV_W = \cH^\sA$. The MES is an eigenstate of $\Omega$ with eigenvalue $1$, and $\|\Omega\|_\infty = 1$. One then needs to check whether $\nu(\Omega) \neq 0$ and how large the unique irreducible subspace $\cK_\star$ is.
We consider concrete measurement bases and analyze the corresponding send-and-measure LOCC protocol.
The detailed calculations are given in~\cref{sec:calculation of bounds in a specific case}.

We start with a pair of MUB-based measurements, such as measurements in the computational and Fourier bases. When the former and the latter bases are chosen with probability $q$ and $1-q$, respectively, one can show that $\nu(\Omega) = \min\{q, 1-q\}$. Unless $q=0$ or $1$, $\nu(\Omega) \neq 0$, which implies that the irreducible subspace $\cK_\star$ of $\cV_W$ is unique. 
Moreover, one can check that $\cK_\star = \cH^\sA$. Thus, Alice and Bob can certify by the send-and-measure LOCC protocol that $\log d$ ebits can be LOCC distilled with infidelity
\begin{align}
    \varepsilon_{\mathrm{distill}}
    &\leq  \min_{0< q <1}\f{1-\big(q p_\mathrm{comp}^2 + (1-q) p_\mathrm{Fourier}^2\big)}{\min\{q, 1-q\}}  \\
    \label{Eq:Ent_Check_Fourier}
    &=  2-\big(p_\mathrm{comp}^2 + p_\mathrm{Fourier}^2\big),
\end{align}
where $p_\mathrm{comp}$ and $p_\mathrm{Fourier}$ are the coincidence probabilities in the corresponding bases. In this choice of bases, setting $q=1/2$ always yields the tightest bound, independently of the values of $p_\mathrm{comp}$ and $p_\mathrm{Fourier}$.

MUB-based measurements may, however, be experimentally demanding due to the difficulty of implementing the Fourier transform. As an alternative, we also consider a qudit encoded in energy levels, as is typical in various quantum platforms. In this setting, the energy basis is usually accessible, and rotations between neighboring energy levels may also be feasible. We therefore fix the first basis as the energy eigenbasis $\{\ket{1},\ket{2},\dots,\ket{d}\}$. The second basis is obtained by sequentially applying Hadamard rotations first to all neighboring odd-even pairs of the energy levels and then to all neighboring even-odd pairs. We choose the former and the latter bases with probabilities $q$ and $1-q$, respectively.

In this choice, we can show that
\begin{equation}
    \nu(\Omega) = \f{1 - \sqrt{1- 4 q (1-q) \sin^2(\f{\pi}{d})}}{2},
\end{equation}
which implies that the irreducible subspace $\cK_\star$ is unique if $q \neq 0, 1$. It can also be shown that $\cK_\star = \cH^{\sA}$. Hence, the send-and-measure LOCC protocol enables Alice and Bob to certify that $\log d$ ebits are LOCC distillable with infidelity
\begin{align}
    \varepsilon_{\mathrm{distill}}   
    &\leq
    \inf_{0<q<1} 2\f{1-\big(q p_\mathrm{comp}^2 +  (1-q) p_\mathrm{Hadamard}^2\big)}{1 - \sqrt{1- 4 q (1-q) \sin^2(\f{\pi}{d})}}.\label{Eq:Ent_Check_Hadamard}
\end{align}
Unlike the Fourier case, the optimal choice of $q$ can depend on the coincidence probabilities $p_{\mathrm{comp}}$ and $p_{\mathrm{Hadamard}}$. To obtain the tightest bound, we first estimate each probability within error $\kappa$ by repeating $\cO(\kappa^{-2}\log\delta^{-1})$ send-and-measure protocols. We then minimize the bound with each probability replaced by its estimate minus $\kappa$ (or zero if negative), which yields a valid bound with probability at least $1-\delta$.

We compare~\cref{Eq:Ent_Check_Fourier,Eq:Ent_Check_Hadamard} in a special case. When all coincidence probabilities happen to be the same value $p$, these bounds lead to $2(1 - p^2)$ and $2(1-p^2)/(1 - \cos(\pi/d))$, respectively. The factor $1/(1 - \cos(\pi/d))$ is the price of using the simpler basis generated by Hadamard rotations between neighboring energy levels. Nevertheless, when $d$ is not too large, the protocol still provides a meaningful quantitative LOCC distillation guarantee.

There is considerable flexibility, even if we focus on rank-1 PVMs. While the above examples use only two measurement bases, one can also use more. In~\cref{fig:numerical NNH bases simulation}, we numerically study such cases for $d=5$, assuming that the noisy channel $\cN^{\sA'\to\sB}$ is a probabilistic mixture of depolarizing and dephasing channels. For each $L$, the $L$ bases are chosen from the available MUB- or neighboring-level Hadamard (NH)-based ones so as to minimize the error bound. The results show that the MUB-based measurements yield consistently tighter bounds, remaining within a small factor of the actual error $\varepsilon_{\mathrm{distill}}$.
The NH bases nevertheless provide meaningful guarantees. 

We emphasize that the error bound does not necessarily improve even if the number $L$ of bases increases. Among the NH-based bounds in our numerical study, the best is obtained at $L=4$, and the error bound becomes worse when $L=6$. 
This non-monotonicity may reflect two effects of sampling the measurement bases uniformly: adding bases with lower coincidence probabilities can increase the numerator of the bound, while the resulting decrease in the probability of selecting each existing basis can narrow the spectral gap $\nu(\Omega)$.

In addition,~\cref{sec:finding basis pairs better than MUBs} numerically explores basis pairs that may yield a tighter bound on $\varepsilon_{\mathrm{distill}}$ than MUBs, focusing on a single-qubit case and the noisy state obtained from the MES under the biased noise.


\subsubsection{Certifying QEC decoding error by CQ decoding errors} 
\label{SSS:QEC}

Our entanglement-certification result also provides a way of checking the performance of QECCs.
Let $\cN^{\sC \to \sB}$ be a noisy channel and let $\cE^{\sA' \to \sC}$ be an encoding channel. While the decoding error $\varepsilon_{\mathrm{QEC}}$ can be evaluated in several ways, one of the standard ways~\cite{watrous2018TheoryQI} is the entanglement-fidelity criterion, defined using an MES $\ket{\Phi}^{\sA\sA'}$ between the input $\sA'$ and a reference $\sA$ as 
\begin{align}
    \varepsilon_{\mathrm{QEC}} \coloneqq \min_{\cD} \varepsilon_{\mathrm{QEC}}(\cD),
\end{align}
where the minimization is over all decoders $\cD^{\sB \to \sA'}$,
\begin{equation}
    \varepsilon_{\mathrm{QEC}}(\cD)
    \coloneqq 1- \rF\big(\cD^{\sB \to \sA'} \circ \cN^{\sC\to\sB} \circ \cE^{\sA'\to\sC}(\Phi^{\sA\sA'}), \Phi^{\sA\sA'}\big),
\end{equation}
and $\sA \cong \sA'$.
This criterion can be related to other measures of decoding error, such as one-shot formulations based on the diamond distance, with a slight modification of the encoding map~\cite{kretschmann2004tema, watrous2018TheoryQI, khatri2024principlemodern, nakata2025constructing}.

Precisely evaluating decoding errors is generally difficult. For stabilizer codes under Pauli noise, one can efficiently estimate the logical error rate via stabilizer-based simulation. Such an estimate, however, relies on the assumed noise model and becomes approximate when the physical noise is not fully characterized. This issue is more significant for general QECCs, including random-coding-type QEC~\cite{lloyd1997capacity, shor2002quantum, devetak2005private, nakata2021oneshot, Gullans2020QuantumCodingLowDepthRandom, Darmawan2024LowdepthrandomCliffordcircuits, Liu2026approximateqec1Dlogdepthcircuits} and QECCs in many-body systems~\cite{Brandao2019Quantumerrorcorrectingcodeseigenstates, Bao2019Eigenstatethermalizationhypothesisaqec, Qasim2025Approximateqeceigenstatethermalizationchaos, Rampp2024haydenpreskillchaoticintegrable, Nakata2024haydenpreskillhamiltoniansystem, nakata2025constructing}, for which explicit and efficient decoders are not necessarily available. Our result bypasses this difficulty by relating the QEC decoding error to classical--quantum (CQ) decoding errors, which are directly measurable without implementing a QEC decoder.

The procedure is essentially the same as the send-and-measure protocol considered above, but we call it a prepare-and-measure protocol.
Let $\bmB = \{B_l\}_{l=1}^{L}$ be a family of orthonormal bases in $\sA'$, and let $\bmM = \{M_l \}_{l=1}^L$ be a family of POVMs on $\sB$, where
\begin{equation}
    B_l = \{\ket{b_{j|l}}^{\sA'}\}_{j=1}^{d_{\sA'}}, \ \ \ 
    M_l = \{ M_{j|l}^\sB\}_{j=1}^{d_{\sA'}}.
\end{equation}
In each round, one chooses $l$ with probability $q_l$, then chooses $j$ uniformly at random, and encodes $\ket{b_{j|l}}$ by $\cE$. After the encoded state becomes noisy, the state is measured by $M_l$. 
This corresponds to the situation where $W = \bI$ (and thus $\cV_W = \cH^\sA$), and  
\begin{equation}
    \Omega^{\sA\sA'} 
    = \sum_{l=1}^{L} q_l \sum_{j=1}^{d_{\sA'}} \ketbra{\bar{b}_{j|l}}{\bar{b}_{j|l}}^\sA \otimes \ketbra{b_{j|l}}{b_{j|l}}^{\sA'}.
\end{equation}
Note that the irreducibility of $\cH^\sA$ with respect to a measurement in the complex-conjugate basis $\bar{\bmB}$ is equivalent to that of $\cH^{\sA'}$ with respect to $\bmB$.

This procedure is naturally interpreted as testing the CQ code $\{\cE^{\sA' \to \sC}(\ketbra{b_{j|l}}{b_{j|l}}^{\sA'})\}_{j=1}^{d_{\sA'}}$, defined in the basis $B_l$. When this CQ code is decoded by $M_l$, the average decoding error is given by
\begin{equation}
    \varepsilon_l 
    \coloneqq 1 - \f{1}{d_{\sA'}} \sum_{j=1}^{d_{\sA'}} \tr\Big[M_{j|l} \cN \circ \cE\big(\ketbra{b_{j|l}}{b_{j|l}}\big)\Big].
\end{equation}
The important quantity in our protocol is then the average of the squared \emph{success} probabilities of the CQ decoding,
\begin{equation}
    C_{\mathrm{CQ}}(\bmB, \bmM) 
    \coloneqq \sum_{l=1}^L q_l (1-\varepsilon_l)^2,
\end{equation}
which provides a quantitative guarantee on the QEC performance of the encoding $\cE$ against the noise $\cN$. 
By experimentally estimating this quantity, one can obtain the decoding error bound. An explicit protocol is given in~\cref{Protocol:QEC_Check}.

\begin{figure}[t]
\begin{protocolbox}{Protocol:QEC_Check}{Checking QEC decoding error by CQ decoding errors}
    \begin{algorithmic}[1]
        \Statex \textbf{Given:} 
        \begin{compactenum}
            \item A noisy channel $\cN^{\sC \to \sB}$,
            \item an encoder $\cE^{\sA' \to \sC}$,
            \item a family of orthonormal bases $\bmB = \{B_l\}_{l=1}^L$, where 
            $B_l = \{ \ket{b_{j|l}}^{\sA'}\}_{j=1}^{d_{\sA'}}$ is an orthonormal basis in $\sA'$,
            \item a family of POVMs $\bmM = \{M_l\}_{l=1}^L$, where 
            $M_l \coloneqq \{M_{j|l}^{\sB}\}_{j=1}^{d_{\sA'}}$ is a POVM on $\sB$,
            \item a probability distribution $\bmq = \{q_l\}_{l=1}^L$,
            \item an accuracy parameter $\kappa > 0$ and a failure probability $\delta \in (0,1)$.
        \end{compactenum}
        \Statex \textbf{Condition:} The Hilbert space $\cH^{\sA'}$ is irreducible with respect to the basis measurements in $\bmB$.
        \Statex \textbf{Output:} An estimate $\widehat{C}_{\mathrm{CQ}}$ of 
        $C_{\mathrm{CQ}}(\bmB, \bmM)$.
        \Statex

        \State Set $m = \lceil \f{1}{2\kappa^2}\log\f{2}{\delta} \rceil$.
        \For{$t=1,\ldots,m$}
            \StateP{Sample $l \in \{1,\ldots,L\}$ with probability $q_l$.}
            \For{$r=1,2$}
                \StateP{Sample $j$ uniformly at random from $\{1,\dots, d_{\sA'}\}$ and prepare $\cN \circ \cE(\ketbra{b_{j|l}}{b_{j|l}})$.}
                \StateP{Measure $\cN \circ \cE(\ketbra{b_{j|l}}{b_{j|l}})$ by the POVM $M_{l}=\{M_{j|l}^{\sB}\}_{j=1}^{d_{\sA'}}$ to obtain an outcome $i$.}
                \StateP{Set $x_r = 1$ if $j = i$ and $0$ otherwise.}
            \EndFor
            \StateP{Set $c_t = x_1 x_2$.}
        \EndFor
        \State Output $\widehat{C}_{\mathrm{CQ}} = \f{1}{m}\sum_{t=1}^m c_t$.
    \end{algorithmic}
\end{protocolbox}
\end{figure}

\begin{corollary}[Checking QEC performance by CQ codes]
\label{cor:QEC_check}
Let $\kappa>0$ and $\delta \in (0,1)$, and let $\widehat{C}_{\mathrm{CQ}}$ be the outcome of~\cref{Protocol:QEC_Check}, which uses $\cO(\kappa^{-2} \log \delta^{-1})$ runs on the inputs $\cN^{\sC\to\sB}$, $\cE^{\sA'\to\sC}$, $\bmB$, $\bmM$, and $\bmq$. Then, the decoding error $\varepsilon_{\mathrm{QEC}}$ satisfies, with probability at least $1-\delta$,
\begin{align}
\label{Eq:QEC_Dec_Error_CQ_Dec_Error_1}
    \varepsilon_{\mathrm{QEC}} 
    \leq \f{1 - \widehat{C}_{\mathrm{CQ}} + \kappa}{\nu(\Omega)}. 
\end{align}
\end{corollary}

\Cref{cor:QEC_check} directly follows from \Cref{cor:1} with $\xi^{\sA\sB} = \cN^{\sC\to\sB}\circ\cE^{\sA'\to\sC}(\Phi^{\sA\sA'})$.
It shows that our prepare-and-measure protocol enables us to certify an upper bound on the decoding error without implementing a full decoder: one only needs to prepare basis states and measure the corresponding noisy encoded outputs with the chosen POVMs, using a number of rounds independent of the system size. Note that the decoding error of the Petz map obeys the same bound in Eq.~\eqref{Eq:QEC_Dec_Error_CQ_Dec_Error_1}.

The bound in~\cref{Eq:QEC_Dec_Error_CQ_Dec_Error_1} depends on the choice of decoding POVMs $\bmM$ of the CQ codes, and optimization over them makes the bound tightest. However, this optimization is not a part of the protocol, and any choice of POVMs gives a valid certificate; unlike process-fidelity bounds for a given operation~\cite{Hofmann2005Complementaryclassicalfidelities, Mayer2018Quantumprocessfidelitybounds}, the POVMs for different bases need not arise from a common decoder. A natural choice when the encoding is unitary and the noise is weak is the POVM obtained by applying the inverse of the encoding unitary and then measuring in the corresponding basis.

\cref{cor:QEC_check} can be viewed as a reduction from the decoding problem of a QECC to the decoding problems of the associated CQ codes. Such a reduction has been studied in the literature, leading to non-trivial decoders for CSS codes, quantum polar codes, and random codes~\cite{Renes2012efficientpolar, wilde2013PolarCodesClassicalQuantum, Renes2014PolarCodesPrivateQuantum, Renes2015effcientpolarnopresherdent, renes2016uncertainrelationsandAQEC, dupuis2021polarQCclifford, nakata2025constructing}. However, almost all of these studies are based on two MUBs, reflecting the standard intuition that quantum information is protected when classical information in two complementary bases is protected. 

Our result shows that mutual unbiasedness is not necessary in this approach. What is required instead is that the chosen family of bases makes the whole logical Hilbert space irreducible. As long as this requirement is satisfied, the same strategy works even for general families of bases, possibly more than two and not necessarily mutually unbiased. This matters when MUBs are hard to implement, as for the energy-level encoding in~\cref{SSS:demonstration_ent_distill}, where rotations between neighboring levels are more accessible than the Fourier transform. This provides a broader foundation for developing reduction-based approaches to decoding and certifying QECCs beyond the conventional MUB-based framework.

\subsection{Quantum measurement and trade-off relations}
\label{sec:Application to measurement process}

Our approach, which links entanglement and distinguishability via irreducibility, also contributes to quantum measurement theory. 
To illustrate this, we specialize the setting of the previous section to the case where the state $\xi$ is generated by applying a measurement process $\cM^{{\sA'}\to\sB\sM}$ (i.e., a channel induced by a quantum instrument) to a pure state $\ket{\sigma}^{\sA\sA'}$.
Specifically, consider the state given by
\begin{align}
    \Theta^{\sA\sB\sM}
    &\coloneqq \cM^{{\sA'}\to\sB\sM}(\ketbra{\sigma}{\sigma}^{\sA\sA'}) \\
    &= \sum_m p_m  \theta_m^{\sA\sB} \otimes \ketbra{m}{m}^\sM,
\end{align}
where $\{\ket{m}^\sM\}_m$ is an orthonormal basis of the register $\sM$, which stores the measurement outcome $m$ occurring with probability $p_m$, and $\theta_m^{\sA\sB}$ denotes the state on $\sA\sB$ conditioned on the outcome $m$.
In addition, we define a measurement channel $\Lambda^{{\sA'}\to\sM}$ associated with the measurement process $\cM^{{\sA'}\to\sB\sM}$ by tracing out $\sB$ from the output: $\Lambda^{{\sA'}\to\sM} \coloneqq \tr_\sB \circ \cM^{{\sA'}\to\sB\sM}$.

Two approaches to trade-off relations in quantum measurement have been studied: the information gain--irreversibility approach and the noise--disturbance approach based on observables.
Along these lines, in~\cref{sec:information gain and irreversibility}, we introduce one-shot versions of these quantities and discuss their operational meaning.
In~\cref{sec:trade-off relation between the information gain and the observable disturbance}, we use our entanglement-distinguishability results to derive a unified trade-off relation, thereby bridging the two approaches.


\subsubsection{Information gain--irreversibility and noise--disturbance approaches}
\label{sec:information gain and irreversibility}

\emph{Information gain} quantifies the amount of information obtainable through a given measurement.
To our knowledge, this notion was first proposed in Ref.~\cite{Groenewold1971InformationGain}, later studied in operational frameworks such as measurement compression and measurement simulation~\cite{Winter2001Compressionofquantummeasurement, Winter2004ExtrinsicIntrinsic, Wilde2012imulatingquantummeasurements, Berta2014identifyinginformationgain}.
In these studies, information gain is quantified by entropy reduction or mutual information.
A convenient way to capture this is to purify the input of a measurement channel $\Lambda^{\sA'\to\sM}$ as $\ket{\sigma}^{\sA\sA'}$ with a reference $\sA$, and measure how much correlation the outcome $\sM$ retains with $\sA$.

Building on these works, we quantify information gain using mutual information. Since we address the \emph{one-shot} setting, we particularly use the sandwiched quantum R\'{e}nyi mutual information of order $\alpha = \f{1}{2}$ (its formal definition is given in Eq.~\eqref{eq:def of 1/2 mutual information}), which, upon smoothing, reduces to the standard mutual information in the i.i.d. asymptotic limit.

\begin{definition}[One-shot information gain]
\label{def:one-shot information gain}
For a state $\ket{\sigma}^{\sA\sA'}$ and a measurement channel $\Lambda^{\sA'\to\sM}$, let $\Theta^{\sA\sM} = \Lambda^{\sA'\to\sM}(\ketbra{\sigma}{\sigma}^{\sA\sA'})$. The one-shot information gain $\mathrm{IG}(\sigma, \Lambda)$ is defined as 
\begin{align}
    \mathrm{IG}(\sigma, \Lambda) \coloneqq \widetilde{I}_{\f{1}{2}}(\sA; \sM)_\Theta,
\end{align}
where $\widetilde{I}_{\f{1}{2}}(\sA; \sM)_\Theta = -\log\max_\tau \rF(\Theta^{\sA\sM}, \sigma^\sA \otimes \tau^\sM)$, and the maximization is over all states $\tau^\sM$.
It satisfies $0 \leq \mathrm{IG}(\sigma, \Lambda) \leq -\log{\tr[(\sigma^{\sA'})^2]} \leq \log{d_{\sA'}}$. 
\end{definition}

The quantity $\mathrm{IG}(\sigma,\Lambda)$ has a direct information-theoretic interpretation. The measurement on $\sA'$ induces an ensemble decomposition $\sigma^\sA=\sum_m p_m\theta_m^\sA$, while the register $\sM$ records the corresponding label $m$. Thus, $\mathrm{IG}(\sigma,\Lambda)$ quantifies, in a R\'{e}nyi sense, how much information the measurement outcome provides about the conditional state of $\sA$.
If the post-measurement state $\Theta^{\sA\sM}$ is close to a product state, the measurement outcome in $\sM$ provides essentially no information about $\sA$. Accordingly, $\mathrm{IG}(\sigma,\Lambda)\approx0$.
In contrast, if $\Theta^{\sA\sM}$ takes the correlated form $\sum_m p_m \ketbra{\psi_m}{\psi_m}^\sA \otimes \ketbra{m}{m}^\sM$ for some states $\{\ket{\psi_m}^\sA\}_m$, then $\mathrm{IG}(\sigma, \Lambda)$ equals its maximal value $-\log{\tr[(\sigma^{\sA'})^2]}$. This is the R\'{e}nyi-2 uncertainty of the state to be measured, and hence, in this case, the measurement outcome fully resolves this uncertainty by specifying a pure conditional state of $\sA$.

The upper bound $-\log\tr[(\sigma^{\sA'})^2]$ more generally shows that the information revealed by the measurement is limited not by the measurement alone but also by the initial correlation in $\ket{\sigma}^{\sA\sA'}$ that characterizes the uncertainty in $\sA$. If $\ket{\sigma}^{\sA\sA'}$ is close to a product state, this uncertainty is small, and no measurement can reveal much information about $\sA$. By contrast, if $\ket{\sigma}^{\sA\sA'}$ is close to an MES, the uncertainty is nearly maximal and can be fully resolved by a suitable measurement.

The following proposition shows that any rank-1 projective measurement achieves the upper bound on $\mathrm{IG}(\sigma, \Lambda)$; its proof is given in~\cref{sec:upperbound on the info gain}.
\begin{proposition}
\label{prop:upperbound infogain}
For an arbitrary orthonormal basis $E = \{\ket{e_m}^{\sA'}\}_m$, let $\Lambda_E^{\sA'\to\sM}$ be the rank-1 projective measurement defined by
\begin{align}
    \label{eq:rank-1 proj measure}
    \Lambda_E^{\sA'\to\sM}(\cdot) 
    = \sum_m \tr_{\sA'}\big[\ketbra{e_m}{e_m}^{\sA'}(\cdot)\big] \ketbra{m}{m}^\sM.
\end{align}
Then, for a state $\ket{\sigma}^{\sA\sA'}$, it achieves the maximal value of the information gain, $\mathrm{IG}(\sigma, \Lambda_E) = -\log{\tr[(\sigma^{\sA'})^2]}$.
\end{proposition}

We also introduce the one-shot \emph{irreversibility} of a measurement process, which quantifies the recoverability of a quantum system after measurements and has been studied in its trade-off with information gain~\cite{Maccone2007Entropicinformation, Buscemi2008GlobalInformationBalance}.

\begin{definition}[One-shot irreversibility]
\label{def:irreversibility}
For a state $\ket{\sigma}^{\sA\sA'}$ and the measurement process $\cM^{{\sA'}\to\sB\sM}$, the one-shot irreversibility is defined as
\begin{align}
\label{eq:definition of q-disturb}
    \mathrm{Irv}(\sigma, \cM) 
    \coloneqq -\log \max_\cP \rF(\cP^{\sB\sM\to\hat{\sB}}(\Theta^{\sA\sB\sM}), \ketbra{\sigma}{\sigma}^{\sA\hat{\sB}}),
\end{align}
where $\Theta^{\sA\sB\sM} = \cM^{{\sA'}\to\sB\sM}(\ketbra{\sigma}{\sigma}^{\sA\sA'})$, and $\hat{\sB} \cong \sA'$.
The maximization is taken over all channels from $\sB\sM$ to $\hat{\sB}$.
It satisfies $0 \leq \mathrm{Irv}(\sigma, \cM) \leq -2\log\|\sigma^{\sA'}\|_\infty \leq 2\log{d_{\sA'}}$.
\end{definition}

This definition clearly captures how well the initial state $\ket{\sigma}^{\sA\sA'}$ can be recovered from the post-measurement state $\Theta^{\sA\sB\sM}$ using only local operations on $\sB\sM$.
We show in~\cref{sec:upperbound on q-distub} that the upper bound on $\mathrm{Irv}(\sigma, \cM)$ is given by $-2\log\|\sigma^{\sA'}\|_\infty$ by considering a measurement that maximally disturbs the system.

As expected, information gain and irreversibility satisfy a trade-off relation: a large information gain forces a large irreversibility, so that gaining more information renders the initial state harder to recover. This is a one-shot counterpart of the results in Refs.~\cite{Maccone2007Entropicinformation, Buscemi2008GlobalInformationBalance}. The proof of this proposition is given in~\cref{sec:proof tradeoff onfo gain and q-disturb}.

\begin{proposition}
\label{thm:Trade-off relation between the information gain and irreversibility}
    For a state $\ket{\sigma}^{\sA\sA'}$, the one-shot information gain and irreversibility of a measurement process $\cM^{{\sA'}\to\sB\sM}$ satisfy the relation: $\mathrm{IG}(\sigma, \Lambda) \leq \mathrm{Irv}(\sigma, \cM)$, where $\Lambda^{\sA'\to\sM} = \tr_\sB \circ \cM^{\sA'\to\sB\sM}$.
\end{proposition}

As a remark, a nontrivial reverse inequality to this proposition, namely, $\mathrm{Irv}(\sigma, \cM) \leq c \mathrm{IG}(\sigma, \Lambda)$ for some constant $c$, does not hold. This is because one can construct a situation in which the information gain is small while the irreversibility is large, for example, by discarding the system $\sA'$ into the environment.


Apart from the approach to characterizing quantum measurements by the information gain and irreversibility, there has been another approach based on the disturbance and noise on observables.
It dates back to discussions by Heisenberg~\cite{Heisenberg1927Uberden, WheelerZurek1983QuantumTheoryMeasurement} and has since been developed alongside studies of uncertainty relations~\cite{Ozawa2003universallyvalidreformulation, ozawa2004uncertaintyprinciplequantuminstruments, Ozawa2005Universaluncertainty, OZAWA2004ncertaintynoisedisturbancegeneralizedmeasure, Buscemi2014NoiseandDisturbance, Renes2017uncertaintyAnoperationalapproach}.

In Ref.~\cite{Buscemi2014NoiseandDisturbance}, disturbance and noise are formulated within an operational framework by how much information about the observable is lost from the post-measurement system.
Specifically, consider an observable $O = \sum_{j=1}^{D} o_j P_j$ with distinct eigenvalues $o_j$, mutually orthogonal projectors $P_j$ summing to the identity, and degeneracies $d_j \coloneqq \tr[P_j]$ satisfying $\sum_{j=1}^D d_j = d$.
By preparing the input of the measurement in the state $P_j/d_j$ with probability $d_j/d$, we can characterize the measurement by the difficulty of recovering the eigenvalue label $j$ from the measurement outputs (classical outcomes and quantum post-measurement state).
For a nondegenerate observable, this reduces to a uniform prior over the $d$ eigenstates.

Based on this interpretation, we define the disturbance and noise for an observable by a measurement process $\cM^{{\sA'}\to\sB\sM}$ and a measurement channel $\Lambda^{\sA' \to \sM} = \tr_\sB\circ\cM^{\sA'\to\sB\sM}$, respectively. Both formulations are based on state discrimination and quantified by the optimal guessing probability (\cref{eq:optimal guessing probability main}).

\begin{definition}[One-shot disturbance for an observable]
\label{def:observable disturbance}
The one-shot disturbance of an observable $O^{\sA'}$ under a measurement process $\cM^{{\sA'}\to\sB\sM}$ is defined by 
\begin{align}
\label{eq:definition of c-disturbance}
    \mathrm{Dst}(O, \cM)
    &\coloneqq -\log{\rP_{\rm opt}(\scE_{O, \cM})},
\end{align}
where $\scE_{O, \cM}$ is the ensemble
\begin{align}
    \scE_{O, \cM} = \big\{d_j/d_{\sA'}; \ \cM^{{\sA'}\to\sB\sM}\big(P_j^{\sA'}/d_j\big)\big\}_{j=1}^{D}.
\end{align}
It satisfies $0 \leq \mathrm{Dst}(O, \cM) \leq \log{(d_{\sA'}/\max_j d_j)}$.
\end{definition}

This quantity quantifies the recoverability of the eigenvalue label $j$ from the post-measurement system $\sB\sM$. 
It is small when the label $j$ can be recovered accurately, and large when the recovery is difficult.

The same quantity can be rewritten in entropic terms, via the conditional min-entropy (see~\cref{eq:def of conditional min entropy} for the definition), as $\mathrm{Dst}(O, \cM) = \widetilde{H}_\infty(\sX|\sB\sM)_\Xi$, where
\begin{align}
\label{eq:def of Xi}
    \Xi^{\sX\sB\sM} 
    = \f{1}{d_{\sA'}} \sum_{j=1}^{D} \ketbra{j}{j}^\sX \otimes \cM^{\sA'\to\sB\sM}\big(P_j^{\sA'}\big).
\end{align}
This provides a one-shot counterpart to the disturbance for observables in Ref.~\cite{Buscemi2014NoiseandDisturbance}, which is defined via the conditional Shannon entropy.

By contrast, the noise (also known as error) for observables measures how accurately a given measurement resolves the observable $O$, focusing on the intuition that a more accurate measurement allows the label $j$ to be better inferred from the outcome in the register $\sM$.

\begin{definition}[One-shot noise for an observable]
The one-shot noise of an observable $O^{\sA'}$ under a measurement channel $\Lambda^{\sA'\to\sM} = \tr_\sB\circ\cM^{\sA'\to\sB\sM}$ is defined by 
\begin{align}
    \mathrm{Noi}(O, \Lambda) \coloneqq -\log{\rP_{\rm opt}}(\scE_{O, \Lambda}),
\end{align}
where $\scE_{O, \Lambda} = \{d_j/d_{\sA'}; \, \Lambda^{\sA'\to\sM}(P_j^{\sA'}/d_j)\}_{j=1}^{D}$.
It satisfies $0 \leq \mathrm{Noi}(O, \Lambda) \leq \log{(d_{\sA'}/\max_j d_j)}$.
\end{definition}

This reflects how well the label $j$ can be recovered from the classical outcome in $\sM$ alone. The noise is small when $j$ can be recovered accurately, and large otherwise.

The noise is clearly larger than the disturbance: 
\begin{align}
\label{eq:noise disturbance relation}
    \mathrm{Noi}(O, \Lambda) \geq \mathrm{Dst}(O, \cM), 
\end{align}
as the former ignores access to the system $\sB$.
As with the disturbance, the noise has the entropic expression: $\mathrm{Noi}(O, \Lambda) = \widetilde{H}_\infty(\sX|\sM)_\Xi$, where $\Xi^{\sX\sM} = \tr_\sB[\Xi^{\sX\sB\sM}]$ and $\Xi^{\sX\sB\sM}$ is given in Eq.~\eqref{eq:def of Xi}.

The disturbance and noise obey, as expected, a trade-off relation across different observables. Let $\{O_l^{\sA'}\}_l$ be a family of observables labeled by $l$, each with the spectral decomposition $O_l^{\sA'} = \sum_{j=1}^{D_l} o_{j|l} P_{j|l}^{\sA'}$ and degeneracies $d_{j|l} \coloneqq \tr[P_{j|l}^{\sA'}]$. For any two of them, the following holds; we prove \Cref{prop:one-shot noise disturbance} in \Cref{sec:proof of disturbance noise trade offs}.

\begin{proposition}
\label{prop:one-shot noise disturbance}
For a measurement process $\cM^{\sA'\to\sB\sM}$ with $\Lambda^{\sA'\to\sM} = \tr_\sB\circ\cM^{\sA'\to\sB\sM}$ and any two observables $O_1^{\sA'}$ and $O_2^{\sA'}$, we have
\begin{align}
\label{eq:one-shot noise disturbance}
    2^{-\mathrm{Dst}(O_1, \cM)} + 2^{-\mathrm{Noi}(O_2, \Lambda)}
    \leq 1 + \sqrt{c(O_1, O_2)},
\end{align}
where $c(O_1, O_2) \coloneqq \max_{j,k}\big\|P_{j|1}^{\sA'}P_{k|2}^{\sA'}\big\|_\infty^2$.
\end{proposition}

Using the arithmetic-geometric mean inequality, we obtain a simpler form of Eq.~\eqref{eq:one-shot noise disturbance}:
\begin{align}
\label{eq:one-shot noise disturbance log form}
    \mathrm{Dst}(O_1, \cM) + \mathrm{Noi}(O_2, \Lambda)
    \geq 2\log{\bigg(\f{2}{1+\sqrt{c(O_1, O_2)}}\bigg)}.
\end{align}
This is a one-shot counterpart of the noise--disturbance trade-off in Ref.~\cite{Buscemi2014NoiseandDisturbance}, formulated there by the Shannon-entropic quantities.

The bound in \Cref{prop:one-shot noise disturbance} is in fact optimal: for the $1$-qubit case with Pauli operators $O_1 = Z$ and $O_2 = X$, a projective measurement that lies midway between the eigenbases of $Z$ and $X$ gives $2^{-\mathrm{Dst}(O_1, \cM)} = 2^{-\mathrm{Noi}(O_2, \Lambda)} = \cos^2(\pi/8)$, attaining equality in Eqs.~\eqref{eq:one-shot noise disturbance} and~\eqref{eq:one-shot noise disturbance log form}.


\subsubsection{Bridging the two approaches: a trade-off relation among the measurement quantities}
\label{sec:trade-off relation between the information gain and the observable disturbance}

So far, we have described two approaches to characterizing quantum measurements. One is the information-gain--irreversibility approach, and the other is the noise--disturbance approach based on observables. Although these two have long developed along separate lines, \cref{thm:1} relates disturbance to irreversibility, thereby establishing their connection.

To this end, we again consider a family of observables $\{O_l^{\sA'}\}_{l=1}^L$, where $O_l^{\sA'} = \sum_{j=1}^{D_l} o_{j|l} P_{j|l}^{\sA'}$ is chosen with probability $q_l > 0$.
We then specialize the family of POVMs $\bmE$ to one in which each $E_l^\sA$ is the complex-conjugate PVM $\{\bar{P}_{j|l}^\sA\}_{j=1}^{D_l}$.
Correspondingly, the operator $\Omega^{\sA\hat{\sB}}$ with $\hat{\sB} \cong \sA'$ becomes  
\begin{align}
\label{inteq:90}
    \Omega^{\sA\hat{\sB}}
    &= \sum_{l=1}^L q_l \Omega_l^{\sA\hat{\sB}} \\
    &= \sum_{l=1}^L q_l \sum_{j=1}^{D_l} \bar{P}_{j|l}^\sA \otimes P_{j|l}^{\hat{\sB}},
\end{align}
which satisfies $\|\Omega\|_\infty = 1$.

Since $\Omega^{\sA\hat{\sB}}\ket{\Phi}^{\sA\hat{\sB}} = \ket{\Phi}^{\sA\hat{\sB}}$, we fix the initial state to be the MES $\ket{\Phi}^{\sA\sA'}$.
Measuring $\sA$ by $\{\bar{P}_{j|l}^\sA\}_j$ yields outcome $j$ with probability $d_{j|l}/d_{\sA'}$ and induces the ensemble: 
\begin{align}
    \scE_{O_l, \cM} = \big\{d_{j|l}/d_{\sA'};\ \cM^{\sA'\to\sB\sM}\big(P_{j|l}^{\sA'}/d_{j|l}\big)\big\}_{j=1}^{D_l},
\end{align}
on $\sB\sM$.

In this situation, the following theorem holds. See~\cref{sec:tradeoff infogain and classical disturb} for a proof.

\begin{theorem}
\label{thm:Trade-off relation between the information gain and the observable disturbance}
For a measurement process $\cM^{{\sA'}\to\sB\sM}$ on the MES $\ket{\Phi}^{\sA\sA'}$, the following trade-off relation holds:
\begin{align}
\label{eq:tradeoff infogain and c-disturbance equation}
    \nu(\Omega)\big(1-2^{-\mathrm{Irv}(\Phi, \cM)}\big) 
    \leq 1 - \big\langle 2^{-2\mathrm{Dst}(O, \cM)} \big\rangle,
\end{align}
where $\big\langle 2^{-2\mathrm{Dst}(O, \cM)}\big\rangle \coloneqq \sum_{l=1}^L q_l 2^{-2\mathrm{Dst}(O_l, \cM)}$, and $\Omega^{\sA\hat{\sB}}$ is given in Eq.~\eqref{inteq:90}.
\end{theorem}

When the choice of the observables has the property $\nu(\Omega) \neq 0$,~\cref{thm:Trade-off relation between the information gain and the observable disturbance} implies that the irreversibility and the disturbance for an observable are constrained to each other. If the disturbance is small for all the observables, the measurement process is nearly reversible. Equivalently, a large irreversibility forces at least some observables to become irrecoverable.

By combining~\cref{thm:Trade-off relation between the information gain and irreversibility,thm:Trade-off relation between the information gain and the observable disturbance,eq:noise disturbance relation}, we obtain a trade-off relation that unifies all the measures introduced so far into a comprehensive form:
\begin{multline}
\label{eq:tradeoff all}
    \nu(\Omega)\big(1-2^{-\mathrm{IG}(\Phi, \Lambda)}\big)  
    \leq \nu(\Omega)\big(1-2^{-\mathrm{Irv}(\Phi, \cM)}\big)\\
    \leq 1 - \big\langle 2^{-2\mathrm{Dst}(O, \cM)} \big\rangle
    \leq 1 - \big\langle 2^{-2\mathrm{Noi}(O, \Lambda)} \big\rangle,
\end{multline}
where $\big\lag 2^{-2\mathrm{Noi}(O, \Lambda)} \big\rag \coloneqq \sum_{l=1}^L q_l 2^{-2\mathrm{Noi}(O_l, \Lambda)}$. 
When they have small values, they can be rewritten in a more intuitive form, such as
\begin{multline}
\nu(\Omega)\mathrm{IG}(\Phi, \Lambda) 
    \leq \nu(\Omega)\mathrm{Irv}(\Phi, \cM) \\
    \lesssim 2\sum_{l=1}^L q_l \mathrm{Dst}(O_l, \cM) \leq 2\sum_{l=1}^L q_l \mathrm{Noi}(O_l, \Lambda).
\end{multline}
This chain shows that gaining substantial information necessarily incurs large irreversibility, which in turn forces large disturbance and noise on observables.
This result bridges the two historically distinct approaches to measurement theory into a unified trade-off relation.

Our result also provides a novel type of \emph{uncertainty relation} between observables. We consider the case with $L=2$ and $q_1 = q_2 = 1/2$. Using the relation between the information gain and the disturbance for an observable in~\cref{eq:tradeoff all}, we obtain
\begin{multline}
\label{eq:info gain and cl distu L=2}
    2^{-2\mathrm{Dst}(O_1, \cM)} + 2^{-2\mathrm{Dst}(O_2, \cM)} \\
    \leq 2\left(1 - \nu\left(\frac{\Omega_1 + \Omega_2}{2}\right)\left(1-2^{-\mathrm{IG}(\Phi, \Lambda)}\right)\right).
\end{multline}
This can be viewed as an uncertainty relation between $\mathrm{Dst}(O_1, \cM)$ and $\mathrm{Dst}(O_2, \cM)$,
analogous to \Cref{prop:one-shot noise disturbance}, but both quantities are now disturbances.
Such a pair does not admit a bound determined by the observables alone---taking $\cM$ to be the identity channel makes both disturbances vanish.
The information gain on the right-hand side constrains it instead, yielding a novel uncertainty relation made possible by \Cref{thm:1}.

This relation becomes more concrete for a nondegenerate observable $O_1^{\sA'}$ and a measurement process that perfectly distinguishes its eigenbasis $\{\ket{e_{m|1}}\}_{m=1}^{d_{\sA'}}$, namely,
\begin{align}
    &\cM^{{\sA'}\to\sB\sM}(\cdot) \notag\\
    &\hspace{1pc}= \sum_m \bra{e_{m|1}}(\cdot)\ket{e_{m|1}}\ketbra{e_{m|1}}{e_{m|1}}^\sB \otimes \ketbra{m}{m}^\sM,
\end{align}
where $\sB \cong \sA'$.
In this case, $\mathrm{Dst}(O_1, \cM) = 0$ and $\mathrm{IG}(\Phi, \Lambda) = \log d_{\sA'}$ due to~\cref{prop:upperbound infogain}. Then,~\cref{eq:info gain and cl distu L=2} reduces to
\begin{multline}
\label{eq:info gain and single cl distu L=2}
    \mathrm{Dst}(O_2, \cM) \\
    \geq -\f{1}{2}\log{\left(1-2\left(1-\f{1}{d_{\sA'}}\right)\nu\left(\f{\Omega_1 + \Omega_2}{2}\right)\right)}.
\end{multline}
This demonstrates that the given measurement process, which extracts maximal information about $O_1$, necessarily disturbs $O_2$, by an amount that grows with $\nu((\Omega_1 + \Omega_2)/2)$.
This fact is particularly pronounced when the eigenbases of nondegenerate observables $O_1^{\sA'}$ and $O_2^{\sA'}$ are MUBs, for which $\nu\big((\Omega_1 + \Omega_2)/2\big) = 1/2$. Hence, $\mathrm{Dst}(O_2, \cM) \geq \f{1}{2}\log{d_{\sA'}}$, implying that this measurement process, optimal for $O_1^{\sA'}$, cannot make $\mathrm{Dst}(O_2, \cM)$ small and thus extracts little information about $O_2^{\sA'}$.

In~\cref{sec:measurement trade-off simulation}, we numerically study the trade-off between information gain and disturbance, examining how measurement sharpness and observable choices affect the tightness of the bounds.


\section{Discussion}
\label{sec:discussions}

We provided quantitative bounds relating the errors in one-sided LO transformation to the MES and state discrimination for a family of POVMs, revealing a general principle behind the correspondence between quantum correlations and classical distinguishability of measurement labels. At the heart of this connection is the irreducibility induced by the POVMs, a much weaker requirement than the complementarity on which previous quantitative results have largely relied. Moreover, we characterized the conditions under which both errors vanish simultaneously.

Based on this connection, we constructed certification protocols in which distillable entanglement and QEC decoding errors are bounded by classical input-output statistics, without direct implementation of state tomography or decoders. Numerical evaluations with experimentally accessible bases confirmed that the resulting bounds remain informative beyond MUB-based measurements.
We also formulated one-shot versions of information gain, irreversibility, and disturbance and noise for observables, and provided their trade-off relations. This bridges the information gain--irreversibility and observable noise--disturbance approaches in quantum measurement theory. 
Our work hence sheds light on a common framework underlying entanglement distillation, QEC, and quantum measurement trade-offs.

This study opens several directions for future research. One promising direction is to further develop QEC. Our extension beyond MUBs offers novel insights, particularly toward generalizing MUB-based code constructions, such as CSS codes, to non-MUB settings. It may also contribute to quantum polar codes~\cite{Renes2012efficientpolar, wilde2013PolarCodesClassicalQuantum, Renes2014PolarCodesPrivateQuantum, Renes2015effcientpolarnopresherdent, dupuis2021polarQCclifford} and recently developed non-orthogonal codes~\cite{bhalerao2025improvingcommunicationrates}.
It is also interesting to ask whether a decoder can be constructed directly from the POVMs used for state discrimination, as is possible in the MUB setting~\cite{renes2016uncertainrelationsandAQEC, nakata2025constructing}.

From another perspective, the relation between entanglement and state distinguishability is ubiquitous in quantum information theory, and has the potential to unify several seemingly distinct tasks in the field. This viewpoint may provide connections between, for example, channel discrimination, magic state distillation, and other fundamental information-processing tasks.

\acknowledgments
The authors thank Zhaoyi Li, Seth Lloyd, Yosuke Mitsuhashi, Koki Ono, and Kaito Watanabe for helpful discussions.
This work was supported by JST SPRING Grant Number JPMJSP2108, JST CREST Grant Number JPMJCR23I3, JST PRESTO Grant Number JPMJPR2456, JST PRESTO Grant Number JPMJPR24FA, Japan, MEXT Quantum Leap Flagship Program (MEXT QLEAP) Grant No. JPMXS0120319794, JSPS KAKENHI Grant Number JP24K16975, JP25K00924, JP26H02015, JP26K00621, JST NEXUS Grant Number JPMJNX26C2, JST Moonshot R\&D Grant Number JPMJMS256E, RIKEN iTHEMS, RIKEN Pioneering Project `Mathematical foundation of quantum information' (PI Yasuyuki Kawahigashi).

\bibliographystyle{prx}
\bibliography{ref}


\clearpage
\onecolumngrid
\setcounter{page}{1}

\setcounter{section}{0}
\setcounter{equation}{0}

\renewcommand\thesection{S\arabic{section}}
\renewcommand\thesubsection{\thesection.\arabic{subsection}}
\renewcommand\theequation{S\arabic{equation}}
\renewcommand\theHequation{S\arabic{equation}}

\crefalias{section}{suppsection}
\crefalias{subsection}{suppsubsection}

\makeatletter
\renewcommand\p@subsection{}

\def\@sectioncntformat#1{Supplementary Note~S\arabic{section}:\quad}

\renewcommand\@hangfrom@section[3]{\@hangfrom{#1#2}#3}
\renewcommand\@hangfroms@section[2]{#1#2}

\let\original@addcontentsline\addcontentsline
\renewcommand\addcontentsline[3]{%
  \def\sit@a{#1}\def\sit@toc{toc}%
  \ifx\sit@a\sit@toc
    \original@addcontentsline{stoc}{#2}{#3}%
  \else
    \original@addcontentsline{#1}{#2}{#3}%
  \fi}
\makeatother

\begingroup
\begin{center}
{\large\bfseries Supplementary Information for\\[3pt]
``When Classical Correlations Certify Entanglement Recovery''\par}
\vspace{12pt}
{Takeru~Utsumi,$^{1}$ Yota~Tachibana,$^{2}$ Yoshifumi~Nakata,$^{3,4}$ Takaya~Matsuura,$^{5}$ \\ Ryuji~Takagi,$^{1}$ Francesco~Buscemi,$^{2}$ and Masato~Koashi$^{6}$\par}
\vspace{6pt}
{\small\itshape
$^{1}$Graduate School of Arts and Sciences, The University of Tokyo, 3-8-1 Komaba, Meguro-ku, Tokyo 153-8902, Japan\\
$^{2}$Department of Mathematical Informatics, Graduate School of Informatics, Nagoya University, Furo-cho, Chikusa-ku, Nagoya, Aichi 464-8601, Japan\\
$^{3}$Department of Computer Science, School of Computing, Institute of Science Tokyo, 4259 Nagatsuta-cho, Midori-ku, Yokohama, Kanagawa 226-8501, Japan\\
$^{4}$Yukawa Institute for Theoretical Physics, Kyoto University, Kitashirakawa Oiwake-cho, Sakyo-ku, Kyoto 606-8502, Japan\\
$^{5}$RIKEN Center for Quantum Computing (RQC), 2-1 Hirosawa, Wako, Saitama 351-0198, Japan\\
$^{6}$Photon Science Center, Graduate School of Engineering, The University of Tokyo, 7-3-1 Hongo, Bunkyo-ku, Tokyo 113-8656, Japan\par}
\end{center}
\vspace{10pt}
\endgroup

\supptableofcontents

\begin{bibunit}[prx]

\section{Notation}
\label{sec:notation in appendix}

We here summarize the notation used throughout this paper.
Physical systems are denoted by sans-serif letters, e.g., $\sA$, $\sB$, and $\sC$. 
We denote by $d_\cH$ the dimension of a Hilbert space $\cH$, and, for simplicity, write $d_\sA = \dim \cH^\sA$, where $\cH^\sA$ is a Hilbert space associated with system $\sA$. We write $\cH \cong \cK$ when Hilbert spaces $\cH$ and $\cK$ are isomorphic; in particular, we use $\sA \cong \sB$ to denote that $\cH^\sA$ and $\cH^\sB$ are isomorphic.
Operators on $\cH^\sA$ and maps from $\cH^\sA$ to $\cH^\sB$ are denoted by $T^\sA$ and $\cT^{\sA \to \sB}$, respectively; when the input and output systems coincide, we simply write $\cT^\sA$.
The identity operator and identity map are denoted by $\bI^\sA$ and $\cI^\sA$, respectively.
We omit system labels, identity operators, and identity maps whenever they are clear from the context.

A quantum state is described by a positive semidefinite (PSD) operator with unit trace, and a quantum channel is a completely positive (CP) and trace-preserving (TP) map. 
While we denote a pure state by $\ket{\varphi}$, the corresponding density operator is sometimes described as $\varphi = \ketbra{\varphi}{\varphi}$.
For a bipartite state $\rho^{\sA\sB}$, the reduced state on $\sA$ is written as $\rho^\sA = \tr_\sB \rho^{\sA\sB}$, where $\tr_\sB$ is the partial trace over the subsystem $\sB$; the full trace is denoted by $\tr$.

We denote the complex conjugate and the transpose in a given basis by $\bar{\cdot}$ and $\cdot^\top$, respectively, and denote the Hermitian conjugate by $\cdot^\dag$. 
The support of an operator $A$, i.e., the orthogonal complement of its kernel, is denoted by $\supp(A)$. 
For a PSD operator $A$, negative powers $A^{-a}$ for $a > 0$ are taken on $\supp(A)$.

We denote by $\ket{\Phi}$ a maximally entangled state (MES) defined in an orthonormal computational basis. For instance, the MES between $\sA$ and $\sA'$ is $\ket{\Phi}^{\sA\sA'} = \f{1}{\sqrt{d_\sA}}\sum_{i=1}^{d_\sA}\ket{i}^\sA\ket{i}^{\sA'}$, where $\{\ket{i}\}_i$ is the computational basis. 
Note that $\sqrt{d_\sA}\ket{\Phi}^{\sA\sA'} = \sum_i \ket{i}^\sA\ket{i}^{\sA'} = \sum_i\ket{e_i}^\sA\ket{\bar{e}_i}^{\sA'}$ holds for any orthonormal basis $\{\ket{e_i}^\sA\}_i$. Here, the complex conjugate is taken in the computational basis; if $\ket{e_i}^\sA = \sum_j c_{ij} \ket{j}^\sA$, then $\ket{\bar{e}_i}^\sA = \sum_j \bar{c}_{ij} \ket{j}^\sA$.

For an operator $A$, the Schatten-$p$ norm is defined by $\|A\|_p = (\tr[|A|^p])^{1/p}$ for $p \in [1, \infty)$, where $|A| = \sqrt{A^\dag A}$. We particularly use the trace norm, the Hilbert--Schmidt norm, and the operator norm, corresponding to $p=1, 2$, and $p \to \infty$, respectively. For a vector $v$, the Euclidean norm is $\|v\| =\sqrt{v^\dag v}$.
The fidelity between states $\rho$ and $\sigma$ is defined as $\rF(\rho, \sigma) \coloneqq \big(\tr\big[\sqrt{\sqrt{\sigma}\rho\sqrt{\sigma}}\big]\big)^2$.
It is monotonic under any quantum channel $\cT$, i.e., $\rF\big(\cT(\rho), \cT(\sigma)\big) \geq \rF(\rho, \sigma)$.

For $\alpha \in (0, 1) \cup (1, \infty)$, the sandwiched quantum R\'{e}nyi-$\alpha$ divergence of a state $\rho$ and a PSD operator $\sigma$ is defined as 
\begin{align}
    \widetilde{D}_\alpha(\rho \| \sigma) \coloneqq \f{1}{\alpha - 1} \log \Big(\tr\Big[\big(\sigma^{\f{1-\alpha}{2\alpha}}\rho\sigma^{\f{1-\alpha}{2\alpha}}\big)^\alpha\Big]\Big),
\end{align}
if $\alpha \in (0, 1)$ and $\tr(\rho\sigma) \neq 0$, or if $\alpha \in (1, \infty)$ and $\supp[\rho] \subseteq \supp[\sigma]$; otherwise, $\widetilde{D}_\alpha(\rho \| \sigma) \coloneqq \infty$~\cite{muller2013quantum}.
Of common interest are the cases $\alpha = 1/2$ and $\alpha \to \infty$.
We particularly use the R\'{e}nyi-$\f{1}{2}$ mutual information~\cite{Hayashi2016Correlationmutualinformation}: 
\begin{align}
\label{eq:def of 1/2 mutual information}
    \widetilde{I}_{\f{1}{2}}(\sA; \sB)_\rho 
    &\coloneqq \min_\sigma \widetilde{D}_{\f{1}{2}}(\rho^{\sA\sB} \| \rho^\sA \otimes \sigma^\sB) \\
    &= -\log\max_\sigma \rF(\rho^{\sA\sB}, \rho^\sA \otimes \sigma^\sB),
\end{align}
and the conditional R\'{e}nyi-$\infty$ entropy, which is also known as the conditional min-entropy~\cite{Renner05securityQKD}: 
\begin{align}
\label{eq:def of conditional min entropy}
    \widetilde{H}_\infty(\sA|\sB)_\rho 
    &\coloneqq -\min_\sigma \widetilde{D}_\infty(\rho^{\sA\sB} \| \bI^\sA \otimes \sigma^\sB).
\end{align}
In both quantities, the optimizations are taken over all states on $\sB$.


\section{Proof of Theorem~\ref{thm:1}}
\label{sec:proof thm1 new}

Our proof of Theorem~\ref{thm:1} uses the pretty good recovery map~\cite{Berta2014Entanglementassistedguessingcomplementary} and the pretty good measurement (PGM)~\cite{hausladen1994prettygood} associated with a certain ensemble.
We introduce the pretty good recovery map and its relation to the well-known Petz map, and then recall the PGM and prove a useful lemma relating them.

For a given state $\rho^{\sA\sB}$, the \emph{pretty good recovery map}~\cite{Berta2014Entanglementassistedguessingcomplementary} is defined as
\begin{align}
\label{eq:state based petz map}
    \cP_\rho^{\sB\to\sA}(\cdot) 
    = \tr_\sB\big[(\rho^{\sA\sB})^{\top_\sA}(\rho^\sB)^{-1/2}(\cdot)(\rho^\sB)^{-1/2}\big].
\end{align}
The partial transpose $\top_\sA$ is taken with respect to the computational basis of $\sA$.
A linear map $\cT^{\sB\to\sA}$ is CP if and only if its Choi operator $\cT^{\sB\to\sA}(\ketbra{I}{I}^{\sB\sB'})$ is positive semidefinite, where $\ket{I}^{\sB\sB'} = \sqrt{d_\sB}\ket{\Phi}^{\sB\sB'} = \sum_{i=1}^{d_\sB}\ket{i}^\sB\ket{i}^{\sB'}$ is the unnormalized MES in the computational basis $\{\ket{i}\}$.
Then, noting that $A^\sB\ket{I}^{\sB\sB'} = (A^{\sB'})^\top\ket{I}^{\sB\sB'}$ holds for any operator $A^\sB$, we see that the pretty good recovery map is CPTP on the support of $\rho^\sB$, because 
\begin{align}
    \cP_\rho^{\sB\to\sA}(\ketbra{I}{I}^{\sB\sB'}) 
    &= \big((\rho^{\sB'})^{-1/2}\rho^{\sA\sB'}(\rho^{\sB'})^{-1/2}\big)^\top \\
    &\geq 0,
\end{align}
and $\tr[\cP_\rho^{\sB\to\sA}(\cdot)] = \tr[\Pi_{\supp(\rho)}^\sB(\cdot)]$, where $\Pi_{\supp(\rho)}^\sB$ is the projector onto $\supp(\rho^\sB)$.

The pretty good recovery map can be regarded as the state-based version of the standard Petz map~\cite{petz1986sufficient, petz1988sufficiency}:
\begin{align}
    \cP_{\sigma, \cT}(\cdot) 
    = \sigma^{1/2}\cT^\dag\big(\cT(\sigma)^{-1/2}(\cdot)\cT(\sigma)^{-1/2}\big)\sigma^{1/2},
\end{align}
where $\cT^\dag$ denotes the adjoint map of $\cT$, defined by $\tr[A\cT(B)]=\tr[\cT^\dag(A)B]$ for any operators $A$ and $B$.
Indeed, define $\sigma^\sA \coloneqq (\rho^\sA)^\top$, and let $\sigma^{\sA'}$ be $\sigma^\sA$ with $\sA$ relabeled as $\sA'$.
Consider a decomposition of the state $\rho^{\sA\sB} = \cT^{\sA'\to\sB}(\ketbra{\psi}{\psi}^{\sA\sA'})$ with a channel $\cT^{\sA'\to\sB}$ and a pure state $\ket{\psi}^{\sA\sA'} = \sqrt{\sigma^{\sA'}}\ket{I}^{\sA\sA'}$, whose marginal on $\sA$ is $\rho^\sA$.
Uhlmann's theorem~\cite{uhlmann1976transition} ensures that such a decomposition always exists.
Noting that $(\ketbra{I}{I}^{\sA\sA'})^{\top_\sA} = F^{\sA\sA'}$, where $F^{\sA\sA'}$ is the swap operator, the right-hand side of Eq.~\eqref{eq:state based petz map} can be rephrased as 
\begin{align}
    &\tr_\sB\big[(\cT^{\sA'\to\sB}(\ketbra{\psi}{\psi}^{\sA\sA'}))^{\top_\sA} (\cT^{\sA'\to\sB}(\sigma^{\sA'}))^{-1/2}(\cdot)(\cT^{\sA'\to\sB}(\sigma^{\sA'}))^{-1/2}\big] \notag \\
    &= \tr_{\sA'}\big[\big((\sigma^{\sA'})^{1/2}\ketbra{I}{I}^{\sA\sA'}(\sigma^{\sA'})^{1/2}\big)^{\top_\sA} (\cT^{\sA'\to\sB})^\dag\big((\cT^{\sA'\to\sB}(\sigma^{\sA'}))^{-1/2}(\cdot)(\cT^{\sA'\to\sB}(\sigma^{\sA'}))^{-1/2}\big)\big] \\
    &= \tr_{\sA'}\big[F^{\sA\sA'}(\sigma^{\sA'})^{1/2}(\cT^{\sA'\to\sB})^\dag\big((\cT^{\sA'\to\sB}(\sigma^{\sA'}))^{-1/2} (\cdot)(\cT^{\sA'\to\sB}(\sigma^{\sA'}))^{-1/2}\big)(\sigma^{\sA'})^{1/2}\big] \\
    &= (\sigma^\sA)^{1/2}(\cT^{\sA\to\sB})^\dag\big((\cT^{\sA\to\sB}(\sigma^\sA))^{-1/2} (\cdot)(\cT^{\sA\to\sB}(\sigma^\sA))^{-1/2}\big)(\sigma^\sA)^{1/2},
\end{align}
where, in the last equality, we identify $\sA'$ with $\sA$ in the computational basis.
Hence, the standard form of the Petz map is obtained from the pretty good recovery map.
We simply refer to the pretty good recovery map given by Eq.~\eqref{eq:state based petz map} as the Petz map unless this causes confusion.

The PGM~\cite{hausladen1994prettygood} is a POVM measurement defined from an ensemble of states. Concretely, for an ensemble of states $\scE = \{p_j; \ \rho_j^\sB\}_{j=1}^J$, the POVM elements of the PGM $\Gamma_\scE^\sB$ are given by $\big\{p_j(\rho^\sB)^{-1/2}\rho_j^\sB(\rho^\sB)^{-1/2}\big\}_{j=1}^J$, where $\rho^\sB = \sum_{j=1}^J p_j \rho_j^\sB$.
To be precise, an additional POVM element is included which projects onto the kernel of $\rho^\sB$.

The PGM is known to satisfy the Barnum--Knill inequality~\cite{barnum2002reversing}: for any ensemble of states $\scE$, the guessing probability by the PGM, $\rP(\Gamma_\scE|\scE)$, is at least the square of the guessing probability by the optimal POVM $\rP_{\rm opt}(\scE)$; that is,  
\begin{equation}
\label{eq:property of PGM}
\big(\rP_{\rm opt}(\scE)\big)^2 \leq \rP(\Gamma_\scE|\scE).
\end{equation}
An improved inequality of Eq.~\eqref{eq:property of PGM} has been proposed in Ref.~\cite{renes2017better}, which can yield a tighter bound, but we use Eq.~\eqref{eq:property of PGM} in this work for simplicity.

For the Petz map and the PGM, we provide the following lemma.
\begin{lemma}
\label{lem:3}
Let $\sA \cong \hat{\sB}$. For any state $\rho^{\sA\sB}$ and POVM $E = \{E_j^{\sA}\}_j$, the Petz map $\cP_\rho^{\sB\to\hat{\sB}}$ and the PGM $\Gamma_{\scE_{\rho, E}}^\sB$ satisfy
\begin{align}
\label{eq:equality of petz and pgm}
    \rP(\Gamma_{\scE_{\rho, E}}|\scE_{\rho, E})
    = \tr\Big[\sum_j (E_j^{\sA} \otimes \bar{E}_j^{\hat{\sB}}) \cP_\rho^{\sB \to \hat{\sB}}(\rho^{\sA\sB})\Big],
\end{align}
where the ensemble $\scE_{\rho, E} = \{p_j; \ \rho_j^\sB\}_j$ is given by $p_j = \tr[E_j^\sA \rho^{\sA\sB}]$ and $\rho_j^\sB = \tr_\sA[E_j^\sA\rho^{\sA\sB}]/p_j$.
\end{lemma}

Lemma~\ref{lem:3} implies that guessing the outcome $j$ of an arbitrary POVM $E^\sA$ on system $\sA$ from system $\sB$ via the PGM is equivalent to performing the conjugate POVM $\bar{E}^{\hat{\sB}}$ on the output of the Petz map $\cP_\rho^{\sB\to\hat{\sB}}$.

\begin{proof}[Proof of Lemma~\ref{lem:3}]
For any state $\rho^{\sA\sB}$, any POVM $\{E_j^\sA\}_j$, and any channel $\cP^{\sB\to\hat{\sB}}$, the following holds:
\begin{align}
    \tr\Big[\sum_j (E_j^\sA \otimes \bar{E}_j^{\hat{\sB}})\cP^{\sB\to\hat{\sB}}(\rho^{\sA\sB})\Big]
    &= \sum_j \tr\big[\big(E_j^\sA \otimes (\cP^{\sB\to\hat{\sB}})^\dag(\bar{E}_j^{\hat{\sB}})\big)\rho^{\sA\sB}\big] \\
    \label{eq:map and povm general relation}
    &= \sum_j p_j \tr\big[(\cP^{\sB\to\hat{\sB}})^\dag(\bar{E}_j^{\hat{\sB}})\rho_j^\sB\big],
\end{align}
where we note that $\big\{(\cP^{\sB \to \hat{\sB}})^\dag(\bar{E}_j^{\hat{\sB}})\big\}_j$ forms a POVM since the adjoint map of a channel is CP and unital. 
For the Petz map $\cP_\rho^{\sB\to\hat{\sB}}$ defined in Eq.~\eqref{eq:state based petz map}, its adjoint is given by
\begin{align}
    (\cP_\rho^{\sB\to\hat{\sB}})^\dag(\cdot) = (\rho^\sB)^{-1/2}\tr_{\hat{\sB}}\big[\big((\cdot)^{\top_{\hat{\sB}}} \otimes \bI^\sB\big)\rho^{\hat{\sB}\sB}\big](\rho^\sB)^{-1/2}.
\end{align}
Hence, $(\cP_\rho^{\sB \to \hat{\sB}})^\dag(\bar{E}_j^{\hat{\sB}}) = p_j (\rho^\sB)^{-1/2}\rho_j^\sB(\rho^\sB)^{-1/2}$.
This coincides with the element of the PGM for the ensemble $\scE_{\rho, E}$, namely, $\Gamma_{\scE_{\rho, E}} = \{p_j\rho^{-1/2}\rho_j\rho^{-1/2}\}_j$. As a result, we obtain 
\begin{align}
    \tr\Big[\sum_j (E_j^\sA \otimes \bar{E}_j^{\hat{\sB}})\cP_\rho^{\sB\to\hat{\sB}}(\rho^{\sA\sB})\Big] 
    &= \sum_jp_j\tr\big[p_j(\rho^\sB)^{-1/2}\rho_j^\sB(\rho^\sB)^{-1/2}\rho_j^\sB\big] \\
    &= \rP(\Gamma_{\scE_{\rho, E}}|\scE_{\rho, E}).
\end{align}

\end{proof}

We now prove Theorem~\ref{thm:1} in a slightly generalized form, which remains nontrivial when the largest eigenvalue of $\Omega^{\sA\hat{\sB}}$ is degenerate.
Let $L$ be an arbitrary positive integer. For each $l = 1, 2, \ldots, L$, let $E_l^\sA = \{E_{j|l}^\sA\}_j$ be a POVM, and $\Omega_l^{\sA\hat{\sB}} = \sum_j E_{j|l}^\sA \otimes \bar{E}_{j|l}^{\hat{\sB}}$, where $\sA \cong \hat{\sB}$. For a probability distribution $\{q_l\}_{l=1}^L$ satisfying $q_l > 0$ for all $l$, let $\Omega^{\sA\hat{\sB}} = \sum_{l=1}^L q_l \Omega_l^{\sA\hat{\sB}}$, and let $\Pi_{\cV_\Omega}^{\sA\hat{\sB}}$ be the projector onto its principal eigenspace $\cV_\Omega$. We define $\tilde{\nu}(\Omega)$ as the difference between the two largest \emph{distinct} eigenvalues of $\Omega^{\sA\hat{\sB}}$, where $\tilde{\nu}(\Omega) \coloneqq \|\Omega\|_\infty$ if all eigenvalues of $\Omega^{\sA\hat{\sB}}$ have the same value. Unlike the spectral gap $\nu(\Omega)$, i.e., the difference between the largest and the second-largest eigenvalues of $\Omega^{\sA\hat{\sB}}$ counted with multiplicity, $\tilde{\nu}(\Omega)$ is always positive, and it reduces to $\nu(\Omega)$ when $\dim \cV_\Omega = 1$.

\begin{theorem}
\label{thm:general thm1}
For any state $\xi^{\sA\sB}$, the Petz map $\cP_\xi^{\sB\to\hat{\sB}}$ satisfies
\begin{align}
\label{inteq:general thm1}
    \min_{\tau: \, \supp(\tau) \subseteq \cV_\Omega}\big(1 - \rF\big(\cP_\xi^{\sB\to\hat{\sB}}(\xi^{\sA\sB}), \tau^{\sA\hat{\sB}}\big)\big)
    \leq \f{\|\Omega\|_\infty - \big\lag \rP_{\rm opt}^2(\scE_{\xi, E})\big\rag}{\tilde{\nu}(\Omega)},
\end{align}
where $\big\lag \rP_{\rm opt}^2(\scE_{\xi, E}) \big\rag \coloneqq \sum_{l=1}^L q_l \big(\rP_{\rm opt}(\scE_{\xi, E_l})\big)^2$, and the minimization is taken over all states $\tau^{\sA\hat{\sB}}$ supported on $\cV_\Omega$.
\end{theorem}

\Cref{thm:general thm1} contains \cref{thm:1}; if $\dim \cV_\Omega = 1$, the only state supported on $\cV_\Omega$ is $\ket{\phi}$ in \cref{thm:1}, and $\tilde{\nu}(\Omega) = \nu(\Omega) > 0$, which makes Eq.~\eqref{inteq:general thm1} equivalent to Eq.~\eqref{eq:equation of main thm1} with $\cP = \cP_\xi$.
If $\dim \cV_\Omega \geq 2$, then $\nu(\Omega) = 0$, and Eq.~\eqref{eq:equation of main thm1} becomes trivial, whereas Eq.~\eqref{inteq:general thm1} may still provide a nontrivial bound on the infidelity to the closest state supported on $\cV_\Omega$ owing to $\tilde{\nu}(\Omega) > 0$.

\begin{proof}[Proof of \Cref{thm:general thm1}]

We first show that, for any state $\sigma^{\sA\hat{\sB}}$,
\begin{align}
\label{eq:infidelity to subspace}
    \max_{\tau: \, \supp(\tau) \subseteq \cV_\Omega}\rF(\sigma^{\sA\hat{\sB}}, \tau^{\sA\hat{\sB}}) 
    = \tr\big[\Pi_{\cV_\Omega}^{\sA\hat{\sB}}\sigma^{\sA\hat{\sB}}\big].
\end{align}
Let $t \coloneqq \tr[\Pi_{\cV_\Omega}^{\sA\hat{\sB}}\sigma^{\sA\hat{\sB}}]$. For any state $\tau$ supported on $\cV_\Omega$, the measurement $\{\Pi_{\cV_\Omega}, \bI - \Pi_{\cV_\Omega}\}$ yields the outcome corresponding to $\Pi_{\cV_\Omega}$ with probability $t$ for $\sigma$ and with certainty for $\tau$. The monotonicity of the fidelity under this measurement then implies $\rF(\sigma, \tau) \leq t$.
We show that $\rF(\sigma, \tau) = t$ holds for some state $\tau$ supported on $\cV_\Omega$. 
If $t = 0$, this holds for any such $\tau$. If $t \neq 0$, the state $\tau = \Pi_{\cV_\Omega}\sigma\Pi_{\cV_\Omega}/t$, whose support is contained in $\cV_\Omega$, achieves this equality. 
To see this, for a purification $\ket{\sigma}^{\sA\hat{\sB}\sR}$ of $\sigma^{\sA\hat{\sB}}$, the vector $(\Pi_{\cV_\Omega}^{\sA\hat{\sB}} \otimes \bI^\sR)\ket{\sigma}^{\sA\hat{\sB}\sR}/\sqrt{t}$ is a purification of $\tau^{\sA\hat{\sB}}$, and Uhlmann's theorem~\cite{uhlmann1976transition} yields
\begin{align}
    \rF(\sigma^{\sA\hat{\sB}}, \tau^{\sA\hat{\sB}})
    &\geq \f{1}{t}\big|\bra{\sigma}^{\sA\hat{\sB}\sR}(\Pi_{\cV_\Omega}^{\sA\hat{\sB}} \otimes \bI^\sR)\ket{\sigma}^{\sA\hat{\sB}\sR}\big|^2 \\
    &= t.
\end{align}
Together with $\rF(\sigma, \tau) \leq t$, this implies Eq.~\eqref{eq:infidelity to subspace}.

We next use the following technique, which is well known in the context of quantum state verification~\cite{Pallister2018OptimalVerificationEntangledStates, Zhu2019optimalverificationfidelityestimation, Li2019efficientverification, Yu2022StatisticalMethodsStateVerification, Li2026QuantumStateVerification}.
Every eigenvalue of $\Omega^{\sA\hat{\sB}}$ other than $\|\Omega\|_\infty$ is at most $\|\Omega\|_\infty - \tilde{\nu}(\Omega)$ by definition, and hence the spectral decomposition of $\Omega^{\sA\hat{\sB}}$ yields
\begin{align}
\label{eq:certification}
    \tilde{\nu}(\Omega)\Pi_{\cV_\Omega} \geq \Omega - \big(\|\Omega\|_\infty - \tilde{\nu}(\Omega)\big)\bI.
\end{align}
For a state $\xi^{\sA\sB}$, by multiplying both sides of Eq.~\eqref{eq:certification} by $\cP_\xi^{\sB\to\hat{\sB}}(\xi^{\sA\sB})$ and taking the trace, we obtain
\begin{align}
    \tilde{\nu}(\Omega)\big(1 - \tr\big[\Pi_{\cV_\Omega}\cP_\xi^{\sB\to\hat{\sB}}(\xi^{\sA\sB})\big]\big)
    \label{inteq:1}
    \leq \|\Omega\|_\infty - \tr\big[\Omega^{\sA\hat{\sB}}\cP_\xi^{\sB\to\hat{\sB}}(\xi^{\sA\sB})\big].
\end{align}

The second term on the right-hand side of Eq.~\eqref{inteq:1} can be computed by using Lemma~\ref{lem:3} as follows:
\begin{align}
    \tr\big[\Omega^{\sA\hat{\sB}} \cP_\xi^{\sB\to\hat{\sB}}(\xi^{\sA\sB})\big] 
    &= \sum_{l=1}^L q_l\tr\Big[\sum_j(E_{j|l}^\sA \otimes \bar{E}_{j|l}^{\hat{\sB}})\cP_\xi^{\sB\to\hat{\sB}}(\xi^{\sA\sB})\Big] \\
    \label{inteq:43}
    &= \sum_{l=1}^L q_l \rP\big(\Gamma_{\scE_{\xi, E_l}}|\scE_{\xi, E_l}\big).
\end{align}
Note that for each $l$, $\{E_{j|l}^\sA\}_j$ is a POVM, and the ensemble $\scE_{\xi, E_l}$ is given by $\scE_{\xi, E_l} = \big\{p_{j|l}; \ \xi_{j|l}^\sB = \tr_\sA\big[E_{j|l}^\sA\xi^{\sA\sB}\big]/p_{j|l}\big\}_j$, where $p_{j|l} = \tr[E_{j|l}^\sA\xi^{\sA\sB}]$.
By substituting Eq.~\eqref{inteq:43} into Eq.~\eqref{inteq:1}, we obtain  
\begin{align}
\label{inteq:2}
    \tilde{\nu}(\Omega)\big(1 - \tr\big[\Pi_{\cV_\Omega}\cP_\xi^{\sB\to\hat{\sB}}(\xi^{\sA\sB})\big]\big)
    \leq \|\Omega\|_\infty - \sum_{l=1}^L q_l \rP\big(\Gamma_{\scE_{\xi, E_l}}|\scE_{\xi, E_l}\big).
\end{align}

Applying the Barnum--Knill inequality in Eq.~\eqref{eq:property of PGM} to bound the right-hand side of Eq.~\eqref{inteq:2}, we obtain
\begin{align}
\label{inteq:22}
    \tilde{\nu}(\Omega)\big(1 - \tr\big[\Pi_{\cV_\Omega}\cP_\xi^{\sB\to\hat{\sB}}(\xi^{\sA\sB})\big]\big)
    \leq \|\Omega\|_\infty - \big\lag \rP_{\rm opt}^2(\scE_{\xi, E})\big\rag,
\end{align}
where $\big\lag \rP_{\rm opt}^2(\scE_{\xi, E})\big\rag = \sum_{l=1}^L q_l \big(\rP_{\rm opt}(\scE_{\xi, E_l})\big)^2$. 
Noting that $\tilde{\nu}(\Omega) > 0$, we combine Eq.~\eqref{inteq:22} with Eq.~\eqref{eq:infidelity to subspace} for $\sigma = \cP_\xi^{\sB\to\hat{\sB}}(\xi^{\sA\sB})$, and obtain Eq.~\eqref{inteq:general thm1}.

\end{proof}

We briefly discuss the implications of Eq.~\eqref{inteq:22} for the range of $\rP_{\rm opt}(\scE_{\xi, E_l})$.
We note that for any state $\xi^{\sA\sB}$, $\rP_{\rm opt}(\scE_{\xi, E_l}) \leq \max_j \|E_{j|l}\|_\infty$ for all $l$, which follows from a straightforward calculation:
\begin{align}
    \rP_{\rm opt}(\scE_{\xi, E_l})
    &= \max_{\{F_j\}_j:\, \mathrm{POVM}} \sum_j p_{j|l}\tr[F_j^\sB \xi_{j|l}^\sB] \notag\\
    &= \max_{\{F_j\}_j:\, \mathrm{POVM}} \tr\Big[\sum_j(E_{j|l}^\sA \otimes F_j^\sB)\xi^{\sA\sB}\Big] \notag\\
    &\leq \max_j \|E_{j|l}\|_\infty
    \max_{\{F_j\}_j:\, \mathrm{POVM}}
    \tr\Big[\sum_j (\bI^\sA \otimes F_j^\sB)\xi^{\sA\sB}\Big] \notag\\
    \label{eq:upper bound guess prob}
    &= \max_j \|E_{j|l}\|_\infty.
\end{align}
Hence when $\rP_{\rm opt}(\scE_{\xi, E_l}) = 1$ is attainable in principle, $\max_j \|E_{j|l}\|_\infty = 1$ must hold.
On the other hand, the non-negativity of the left-hand side of Eq.~\eqref{inteq:22} implies that $\sum_{l=1}^L q_l \big(\rP_{\rm opt}(\scE_{\xi, E_l})\big)^2 \leq \|\Omega\|_\infty$.
Thus, if $\sum_{l=1}^L q_l \max_j \|E_{j|l}\|_\infty^2 > \|\Omega\|_\infty$, there exists no state $\xi^{\sA\sB}$ for which all $\rP_{\rm opt}(\scE_{\xi, E_l})$ simultaneously achieve $\max_j \|E_{j|l}\|_\infty$.


\section{Proof of Proposition~\ref{prop:converse relation}}
\label{sec:derivation of converse relation}

We prove Proposition~\ref{prop:converse relation}, which states that a reliable one-sided LO transformation bounds the guessing probabilities.
Let $E_l = \{E_{j|l}^\sA\}_j$ be a POVM, and $\Omega^{\sA\hat{\sB}} = \sum_{l=1}^L q_l \Omega_l^{\sA\hat{\sB}}$, where $\Omega_l^{\sA\hat{\sB}} = \sum_j E_{j|l}^\sA \otimes \bar{E}_{j|l}^{\hat{\sB}}$, and let $\{q_l\}_{l=1}^L$ be a probability distribution with $q_l > 0$.
As with \cref{thm:1}, we prove Proposition~\ref{prop:converse relation} in a slightly generalized form, where the target is extended from pure states in $\cV_\Omega$ to all states supported on $\cV_\Omega$.

\begin{proposition}
\label{prop:general converse relation}
For any state $\xi^{\sA\sB}$, we have
\begin{align}
\label{inteq:general converse relation}
    \big\lag \rP_{\rm opt}(\scE_{\xi, E})\big\rag
    \geq \|\Omega\|_\infty\Big(1 - \min_{\cP, \tau}\big(1 - \rF\big(\cP^{\sB\to\hat{\sB}}(\xi^{\sA\sB}), \tau^{\sA\hat{\sB}}\big)\big)\Big),
\end{align}
where the minimization is taken over all channels $\cP^{\sB\to\hat{\sB}}$ and all states $\tau^{\sA\hat{\sB}}$ such that $\supp\big(\tau^{\sA\hat{\sB}}\big) \subseteq \cV_\Omega$.
\end{proposition}

\Cref{prop:converse relation} follows from \cref{prop:general converse relation}, since restricting $\tau^{\sA\hat{\sB}}$ to pure states $\ket{\phi}^{\sA\hat{\sB}} \in \cV_\Omega$ does not decrease the minimum in Eq.~\eqref{inteq:general converse relation}; the two coincide when $\dim \cV_\Omega = 1$.

\begin{proof}[Proof of \Cref{prop:general converse relation}]

By a similar calculation to that used to derive Eq.~\eqref{eq:map and povm general relation}, we have 
\begin{align}
    \label{inteq:77}
    \tr\big[\Omega^{\sA\hat{\sB}}\cP^{\sB\to\hat{\sB}}(\xi^{\sA\sB})\big]
    &= \tr\Big[\sum_{j, l} q_l (E_{j|l}^\sA\otimes\bar{E}_{j|l}^{\hat{\sB}})\cP^{\sB\to\hat{\sB}}(\xi^{\sA\sB})\Big] \\
    &= \sum_{j, l} q_l p_{j|l} \tr\big[(\cP^{\sB\to\hat{\sB}})^\dag(\bar{E}_{j|l}^{\hat{\sB}})\xi_{j|l}^\sB\big] \\
    &\leq \sum_l q_l \max_{\{F_j\}_j: \, {\rm POVM}} \sum_j p_{j|l} \tr\big[F_j^\sB\xi_{j|l}^\sB\big] \\
    \label{inteq:62}
    &= \big\lag \rP_{\rm opt}(\scE_{\xi, E})\big\rag,
\end{align}
where $p_{j|l} = \tr[E_{j|l}^\sA\xi^{\sA\sB}]$, and $\xi_{j|l}^\sB = \tr_\sA[E_{j|l}^\sA\xi^{\sA\sB}]/p_{j|l}$.
We took the maximization over all POVMs $\{F_j^\sB\}_j$, noting that $\big\{(\cP^{\sB\to\hat{\sB}})^\dag(\bar{E}_{j|l}^{\hat{\sB}})\big\}_j$ forms a POVM.

Let $\Omega^{\sA\hat{\sB}} = \sum_i \lambda_i(\Omega) \Pi_{\cV_{\lambda_i}}^{\sA\hat{\sB}}$ be the spectral decomposition of $\Omega^{\sA\hat{\sB}}$, where $\Pi_{\cV_{\lambda_i}}^{\sA\hat{\sB}}$ is the projector onto the eigenspace $\cV_{\lambda_i}$ corresponding to $\lambda_i(\Omega)$. In particular, $\cV_\Omega = \cV_{\lambda_1}$ is the principal eigenspace of $\Omega^{\sA\hat{\sB}}$.
The left-hand side of Eq.~\eqref{inteq:77} can be evaluated as
\begin{align}
    \tr\big[\Omega^{\sA\hat{\sB}}\cP^{\sB\to\hat{\sB}}(\xi^{\sA\sB})\big] 
    &= \sum_i\lambda_i(\Omega) \tr\big[\Pi_{\cV_{\lambda_i}}^{\sA\hat{\sB}}\cP^{\sB\to\hat{\sB}}(\xi^{\sA\sB})\big] \\
    &= \|\Omega\|_\infty \tr\big[\Pi_{\cV_\Omega}^{\sA\hat{\sB}}\cP^{\sB\to\hat{\sB}}(\xi^{\sA\sB})\big]  
    + \sum_{i \geq2 }\lambda_i(\Omega) \tr\big[\Pi_{\cV_{\lambda_i}}^{\sA\hat{\sB}}\cP^{\sB\to\hat{\sB}}(\xi^{\sA\sB})\big] \\
    \label{inteq:63}
    &\geq \|\Omega\|_\infty \tr\big[\Pi_{\cV_\Omega}^{\sA\hat{\sB}}\cP^{\sB\to\hat{\sB}}(\xi^{\sA\sB})\big].
\end{align}
In the inequality, we simply drop the second term, which is non-negative as $\Omega^{\sA\hat{\sB}}$ is positive semidefinite.
By Eq.~\eqref{eq:infidelity to subspace} for $\sigma = \cP^{\sB\to\hat{\sB}}(\xi^{\sA\sB})$, the right-hand side of Eq.~\eqref{inteq:63} equals $\|\Omega\|_\infty\max_{\tau: \, \supp(\tau) \subseteq \cV_\Omega}\rF\big(\cP^{\sB\to\hat{\sB}}(\xi^{\sA\sB}), \tau^{\sA\hat{\sB}}\big)$. Combining Eqs.~\eqref{inteq:62} and~\eqref{inteq:63}, we obtain
\begin{align}
\label{inteq:65}
    \big\lag \rP_{\rm opt}(\scE_{\xi, E})\big\rag
    \geq \|\Omega\|_\infty \max_{\tau: \, \supp(\tau) \subseteq \cV_\Omega}\rF\big(\cP^{\sB\to\hat{\sB}}(\xi^{\sA\sB}), \tau^{\sA\hat{\sB}}\big),
\end{align}
which holds for any channel $\cP^{\sB\to\hat{\sB}}$.
Since the left-hand side is independent of $\cP^{\sB\to\hat{\sB}}$, maximizing the right-hand side over $\cP^{\sB\to\hat{\sB}}$ yields Eq.~\eqref{inteq:general converse relation}.

\end{proof}

\section{Proof of Theorem~\ref{prop:qualitative equivalent condition}}
\label{sec:proof of irreducibility and unique state spectrum equivalence}

We prove Theorem~\ref{prop:qualitative equivalent condition} by analyzing the spectral structure of the operator $\Omega$ defined in Eq.~\eqref{eq:def of average of omega}. To this end, we first introduce several technical tools.

We use a technique known as \emph{vectorization}~\cite{Horn1991TopicsinMatrixAnalysis, gilchrist2011vectorizationquantumsuse, Miszczak2011SINGULARVALUEDECOMPOSITION}.
Although a few notational conventions can be found in the literature, we define it as a linear map from $\bC^{d \times d}$ to $\bC^{d} \otimes \bC^{d}$, given by 
\begin{align}
\label{eq:definition of vectorization}
    \rvec: \ketbra{\psi}{\phi} \mapsto \ket{\psi}\ket{\bar{\phi}},
\end{align}
where the complex conjugate is taken in the computational basis. The vectorization in this way is a bijection between matrices and vectors; thus, its inverse $\rvec^{-1}$ is well-defined. The vectorization preserves the Hilbert--Schmidt norm: $\|A - B\|_2 = \|\rvec(A) - \rvec(B)\|$, and satisfies $\rvec(ABC) = (A \otimes C^\top)\rvec(B)$ for any operators $A$, $B$, and $C$.

Let $\cL(\cH)$ denote the set of linear operators on a Hilbert space $\cH$.
Via the vectorization isomorphism, spectral properties of the operator $\Omega$ are inherited by a CP trace-nonincreasing map $\cE^\Omega$ on $\cL(\cH)$, defined as
\begin{align}
\label{eq:completely dephasing}
    \cE^\Omega(\cdot) \coloneqq \sum_{j, l} q_l E_{j|l}(\cdot)E_{j|l},
\end{align}
where $\sum_{l=1}^L q_l = 1$ and $q_l > 0$.
For the spectral analysis that follows, for a linear map $\cF$ on $\cL(\cH)$ and $\lambda \in \bC$, we define
\begin{align}
    \cS_\lambda(\cF) \coloneqq \{T \in \cL(\cH) \mid \cF(T) = \lambda T\},
\end{align}
which is a linear subspace of $\cL(\cH)$.
We call $\lambda$ an eigenvalue of $\cF$ if $\cS_\lambda(\cF) \neq \{0\}$, in which case $\cS_\lambda(\cF)$ is the corresponding eigenspace. We denote the spectrum, the set of all eigenvalues, of $\cF$ by $\mathrm{spec}(\cF)$.
Note that the Hermiticity of the POVM elements $E_{j|l}$ makes the map $\cE^\Omega$ self-adjoint, and thus all its eigenvalues are real.

Next, we briefly examine the structure of an invariant subspace $\cK \subseteq \cH$ under the action of $\cE^\Omega$. See also \Cref{def:reducibility} for invariance and irreducibility with respect to POVMs.
Since $E_{j|l}$ are Hermitian, if $\cR \subseteq \cK$ is invariant under all $E_{j|l}$, its orthogonal complement $\cR^\perp$ is also invariant. Indeed, for any $\ket{u} \in \cR$ and $\ket{v} \in \cR^\perp$, the invariance of $\cR$ implies $\bra{v}E_{j|l}\ket{u} = 0$, which yields $\bra{u}E_{j|l}\ket{v} = \overline{\bra{v}E_{j|l}\ket{u}} = 0$ due to Hermiticity.
Because any invariant subspace can be recursively decomposed into mutually orthogonal invariant subspaces, $\cK$ admits an irreducible decomposition: $\cK = \bigoplus_a \cR_a$, where each $\cR_a$ is invariant under all $E_{j|l}$ and does not allow for further nontrivial invariant orthogonal decomposition.
While the decomposition is not unique in general, this does not affect the subsequent discussion, and we take one such decomposition.

The entire space $\cH$ is invariant under the action of $\{E_l\}_l$, and let $\cH = \bigoplus_a \cR_a$ be an arbitrary orthogonal decomposition of the Hilbert space $\cH$ into $\bmE$-irreducible subspaces. We then denote by $\Pi_a$ the projector onto the irreducible subspace $\cR_a$. The invariance of $\cR_a$ and the Hermiticity of $E_{j|l}$ give $[\Pi_a, E_{j|l}] = 0$ for all $j$, $l$, and $a$.
By decomposing an operator $T \in \cL(\cH)$ into the orthogonal $(a, b)$-blocks as $T = \sum_{a, b} \Pi_a T \Pi_b$, we see that $\cE^\Omega$ preserves this block structure: $\cE^\Omega(\Pi_a T \Pi_b) = \Pi_a \cE^\Omega(T) \Pi_b$.
For each $(a, b)$-block, we define a map $\cE_{ab}^\Omega: \Pi_a\cL(\cH)\Pi_b \to \Pi_a\cL(\cH)\Pi_b$ by
\begin{align}
    \cE_{ab}^\Omega(\cdot)
    = \sum_{j, l} q_l E_{j|l}\big|_{\cR_a}(\cdot)E_{j|l}\big|_{\cR_b},
\end{align}
where $E_{j|l}\big|_{\cR_a}$ is the restriction of $E_{j|l}$ to $\cR_a$.
Under this block decomposition, the action of $\cE^\Omega$ is expressed as
\begin{align}
    \cE^\Omega(T) = \sum_{a,b} \cE_{ab}^\Omega(T_{ab}),
\end{align}
where $T_{ab} = \Pi_a T \Pi_b$. Noting that the $(a,b)$-blocks are orthogonal to each other, the eigenspace $\cS_\lambda(\cE^\Omega)$ and the spectrum $\mathrm{spec}(\cE^\Omega)$ of $\cE^\Omega$ accordingly decompose as
\begin{align}
\label{eq:eigenspace block decomposition}
    \cS_\lambda(\cE^\Omega) = \bigoplus_{a, b} \cS_\lambda(\cE_{ab}^\Omega), \ \ \ \
    \mathrm{spec}(\cE^\Omega) = \bigcup_{a, b}\mathrm{spec}(\cE_{ab}^\Omega),
\end{align}
where $(a,b)$ runs over all pairs, so that the direct sum is over operator subspaces and includes not only the diagonal blocks but also the off-diagonal ones.

Based on the above discussion, we obtain the following theorem regarding the largest eigenvalue $\|\Omega\|_\infty$ and the spectral gap $\nu(\Omega)$ of the operator $\Omega$.
Recall that $\cV_W$ denotes the principal eigenspace of the operator $W = \sum_{j, l} q_l E_{j|l}^2$ corresponding to $\|W\|_\infty$.

\begin{theorem}
\label{thm:spectral structure of irreducible subspace}
Let $\lambda_1^{(a)}$ be the largest eigenvalue of $\cE_{aa}^\Omega$, and let $\lambda_1 = \|\Omega\|_\infty$ and $\nu = \nu(\Omega)$.
The following equivalences hold:
\begin{itemize}
    \item $\lambda_1 = \|W\|_\infty$ if and only if $\cV_W$ contains at least one nonzero $\bmE$-invariant subspace.
    \item $\nu \neq 0$ if and only if 
    \begin{enumerate}
        \item[(1)] there is a unique index $a_\star $ such that $\lambda_1^{(a_\star )} = \lambda_1$, and $\lambda_1^{(a)} < \lambda_1$ for all $a \neq a_\star$, and
        \item[(2)] $\lambda_1 \not\in \mathrm{spec}(\cE_{ab}^\Omega)$ for all $a \neq b$.
    \end{enumerate}
\end{itemize}
\end{theorem}

Since $\nu$ is independent of the decomposition of $\cH$, conditions (1) and (2) hold for one irreducible decomposition if and only if they hold for all.

As will be shown rigorously later in this section, Theorem~\ref{prop:qualitative equivalent condition} follows immediately from Theorem~\ref{thm:spectral structure of irreducible subspace} since the latter is more general.
From the first item of Theorem~\ref{thm:spectral structure of irreducible subspace}, when $\lambda_1 = \|W\|_\infty$, the largest $\bmE$-invariant subspace contained in $\cV_W$ is $\bigoplus_{a\in A}\cR_a$, where $A = \{a : \cR_a \subseteq \cV_W\}$ and each $\cR_a$ is $\bmE$-invariant.
By further imposing $\nu \neq 0$, condition~(1) ensures that $A$ consists of a single element $a_\star$, so that $\cR_{a_\star}$ is the unique $\bmE$-irreducible subspace $\cK_\star$ in $\cV_W$.
Conversely, if $\cV_W$ contains exactly one $\bmE$-irreducible subspace $\cK_\star$, the first item gives $\lambda_1 = \|W\|_\infty$. The uniqueness of $\cK_\star$ further implies conditions~(1) and~(2), and hence $\nu \neq 0$ by the second item. As a consequence, we obtain the equivalence:
\begin{align}
    \|\Omega\|_\infty = \|W\|_\infty \ \text{and} \ \nu(\Omega) \neq 0 
    \iff \cV_W \ \text{contains exactly one} \ \bmE\text{-irreducible subspace.}
\end{align}

\begin{proof}[Proof of Theorem~\ref{thm:spectral structure of irreducible subspace}]

We first show the equivalence between (i) $\lambda_1 = \|W\|_\infty$ and (ii) $\cV_W$ contains at least one nonzero $\bmE$-invariant subspace.
For notational simplicity, we write $\cS_W = \cS_{\|W\|_\infty}(\cE^\Omega)$, where $\cS_{\|W\|_\infty}(\cE^\Omega) = \{T \mid \cE^\Omega(T) = \|W\|_\infty T\}$.

\medskip
\noindent\textbf{(i) $\Rightarrow$ (ii).} 
The condition $\lambda_1 = \|W\|_\infty$ means that there is at least one vector $\ket{T} \neq 0$ such that $\Omega\ket{T} = \|W\|_\infty\ket{T}$. By applying $\rvec^{-1}$ to both sides, we obtain $\cE^\Omega(T) = \|W\|_\infty T$, where $T = \rvec^{-1}(\ket{T})$. 
This implies that $\cS_W$ contains at least one nonzero element $T$.

We show that any $T \in \cS_W$ can be decomposed as a linear combination of positive semidefinite (PSD) operators which also belong to $\cS_W$. 
Since $\cE^\Omega$ preserves Hermiticity, if $T \in \cS_W$, then $T^\dag \in \cS_W$. When we write $T = T_\rR + i T_\rI$, the Hermitian operators 
$T_\rR = (T + T^\dag)/2$ and $T_\rI = (T - T^\dag)/(2i)$ are both in $\cS_W$. Moreover, the Hermitian operator $T_\rR$ can be uniquely decomposed as $T_\rR = T_{\rR+} - T_{\rR-}$ with PSDs $T_{\rR\pm}$ satisfying $T_{\rR+} \perp T_{\rR-}$.
Since $T_\rR \in \cS_W$ and $\cE^\Omega$ is a CP trace-nonincreasing linear map, we see that 
\begin{align}
    \|W\|_\infty\|T_\rR\|_1 
    &= \|\cE^\Omega(T_\rR)\|_1 \\
    &= \|\cE^\Omega(T_{\rR+}) - \cE^\Omega(T_{\rR-})\|_1 \\
    &\leq \|\cE^\Omega(T_{\rR+})\|_1 + \|\cE^\Omega(T_{\rR-})\|_1 \\
    &\leq \|W\|_\infty\|T_{\rR+}\|_1 + \|W\|_\infty\|T_{\rR-}\|_1 \\
    &= \|W\|_\infty\|T_\rR\|_1,
\end{align}
where the second inequality follows from the fact that, for any $A \geq 0$, $\|\cE^\Omega(A)\|_1 = \tr[\cE^\Omega(A)] = \tr[W A] \leq \|W\|_\infty \tr[A] = \|W\|_\infty \|A\|_1$.
The first and last lines are identical, so all intermediate inequalities must hold with equality, which implies $\cE^\Omega(T_{\rR+}) \perp \cE^\Omega(T_{\rR-})$.
While $\|W\|_\infty T_\rR = \cE^\Omega(T_\rR) = \cE^\Omega(T_{\rR+}) - \cE^\Omega(T_{\rR-})$ by $T_\rR \in \cS_W$, the uniqueness of the decomposition of $T_\rR$ ensures that $\cE^\Omega(T_{\rR+}) = \|W\|_\infty T_{\rR+}$ and $\cE^\Omega(T_{\rR-}) = \|W\|_\infty T_{\rR-}$. Hence, $T_{\rR\pm} \in \cS_W$. The same argument applies to $T_\rI$, and thus $T_{\rI\pm} \in \cS_W$, with $T_\rI = T_{\rI+} - T_{\rI-}$.
As a result, any $T \in \cS_W$ can be decomposed into PSDs $T_{\rR\pm}, T_{\rI\pm} \in \cS_W$ as $T = (T_{\rR+} - T_{\rR-}) + i(T_{\rI+} - T_{\rI-})$.

We next show that for any PSD $T \in \cS_W$, its range $\rrange(T)$ is $\bmE$-invariant and contained in $\cV_W$. As seen above, if $\cS_W$ contains a nonzero element, it also contains a nonzero PSD element; thus, it suffices to consider PSD elements.

For any PSDs $A$ and $B$, $\rrange(A + B) = \rrange(A) + \rrange(B)$, where $\cU + \cV = \{\ket{u} + \ket{v} : \ket{u} \in \cU, \ket{v} \in \cV\}$ denotes the sum of subspaces.
Noting that $E_{j|l} T E_{j|l} \geq 0$ and $q_l > 0$, applying this to $\|W\|_\infty T = \cE^\Omega(T) = \sum_{j, l} q_l E_{j|l} T E_{j|l}$ yields that
\begin{align}
    \rrange(T) 
    &= \sum_{j, l} \rrange(E_{j|l} T E_{j|l}) \\
    &\supseteq \rrange(E_{j|l} T E_{j|l}) \\
    &= \rrange\big(E_{j|l} \sqrt{T}\big) \\
    &= E_{j|l} \, \rrange(T),
\end{align}
where we used $\rrange\big(\sqrt{T}\big) = \rrange(T)$.
This implies the invariance of $\rrange(T)$ under each $E_{j|l}$, i.e., $\rrange(T)$ is $\bmE$-invariant.

Let a PSD operator $T$ have eigenvalue decomposition $T = \sum_s t_s \ketbra{t_s}{t_s}$ with $t_s > 0$. Since $\tr[\cE^\Omega(T)] = \tr[WT]$, by taking the trace on both sides of $\cE^\Omega(T) = \|W\|_\infty T$, we obtain  
\begin{align}
\label{inteq:60}
    \sum_s t_s 
    \bra{t_s}\big(\|W\|_\infty \bI - W\big)\ket{t_s} = 0.
\end{align}
Since $W \leq \|W\|_\infty \bI$ by the definition of the operator norm, each term in the sum must be zero, and in fact, $(\|W\|_\infty\bI - W)\ket{t_s} = 0$ for all $s$.
As $\{\ket{t_s}\}_s$ forms a basis of $\rrange(T)$, this implies that $W = \|W\|_\infty \bI$ holds on $\rrange(T)$, i.e., $\rrange(T) \subseteq \cV_W$.

To summarize, the condition $\lambda_1 = \|W\|_\infty$ implies the existence of a nonzero PSD $T \in \cS_W$, whose range is $\bmE$-invariant and contained in $\cV_W$. Hence, $\cV_W$ contains at least one nonzero $\bmE$-invariant subspace.

\medskip
\noindent\textbf{(i) $\Leftarrow$ (ii).}
Let $\cR$ be a nonzero $\bmE$-invariant subspace contained in $\cV_W$, and let $\Pi_{\cR} \neq 0$ be the projector onto $\cR$. From the $\bmE$-invariance of $\cR$, all POVM elements $E_{j|l}$ satisfy $\Pi_{\cR}E_{j|l}\Pi_{\cR}\ket{v} = E_{j|l}\Pi_{\cR}\ket{v}$ for any vector $\ket{v} \in \cH$. 
Noting that $E_{j|l}$ is Hermitian, we obtain 
$\Pi_{\cR}E_{j|l}\Pi_{\cR} = E_{j|l}\Pi_{\cR} = \Pi_{\cR}E_{j|l}$. 
Then, we see that $\Pi_{\cR}$ satisfies
\begin{align}
    \cE^\Omega(\Pi_{\cR}) 
    &= \sum_{j, l} q_l E_{j|l}\Pi_{\cR}E_{j|l} \\
    &= \sum_{j, l} q_l \Pi_{\cR} E_{j|l}^2
       \Pi_{\cR} \\
    &= \Pi_{\cR} \Big(\sum_{j, l} q_l E_{j|l}^2\Big)
       \Pi_{\cR} \\
    &= \Pi_{\cR}\, W\, \Pi_{\cR} \\
    &= \|W\|_\infty\, \Pi_{\cR},
\end{align}
where in the last line we used $W\big|_{\cR} = \|W\|_\infty\, \bI_{\cR}$, which follows from $\cR \subseteq \cV_W$. This shows that $\|W\|_\infty$ is an eigenvalue of $\cE^\Omega$.

On the other hand, using the operator inequality $(E_{j|l} \otimes \bI - \bI \otimes \bar{E}_{j|l})^2 \geq 0$, we have $E_{j|l} \otimes \bar{E}_{j|l} \leq \f{1}{2}(E_{j|l}^2 \otimes \bI + \bI \otimes \bar{E}_{j|l}^2)$. Multiplying by $q_l$ and summing over $j$ and $l$ yields 
\begin{align}
\label{eq:omega w relation}
    \Omega \leq \f{1}{2}(W \otimes \bI + \bI \otimes \bar{W}),
\end{align}
which gives $\lambda_1 = \|\Omega\|_\infty \leq \|W\|_\infty$, where we used $\|\bar{W}\|_\infty = \|W\|_\infty$. Since $\|W\|_\infty \in \mathrm{spec}(\cE^\Omega)$ as shown above, and since $\mathrm{spec}(\cE^\Omega)$ coincides with the spectrum of $\Omega$, we conclude $\lambda_1 = \|W\|_\infty$.

\medskip

Next, we establish the equivalence between (iii) $\nu \neq 0$ and (iv) the following two conditions: (1) the uniqueness of $a_\star$ such that $\lambda_1^{(a_\star)} = \lambda_1$ and $\lambda_1^{(a)} < \lambda_1$ for $a \neq a_\star$, and (2) the absence of $\lambda_1$ from $\mathrm{spec}(\cE_{ab}^\Omega)$ for $a \neq b$, by using the block structure given in Eq.~\eqref{eq:eigenspace block decomposition}.
Note that $\nu \neq 0$ is equivalent to $\dim \cS_{\lambda_1}(\cE^\Omega) = 1$.

\medskip

\noindent\textbf{(iii) $\Rightarrow$ (iv).}
Since $\cE^\Omega$ preserves Hermiticity, $T \in \cS_{\lambda_1}(\cE^\Omega)$ implies $T^\dag \in \cS_{\lambda_1}(\cE^\Omega)$. As $\dim \cS_{\lambda_1}(\cE^\Omega) = 1$, we can take $\cS_{\lambda_1}(\cE^\Omega)$ to be spanned by a nonzero Hermitian operator $T = T^\dag$.

We decompose $T$ as $T = \sum_{a,b} T_{ab}$ with $T_{ab} = \Pi_a T \Pi_b$.
If we suppose $T_{aa} = 0$ for all $a$, there must exist some off-diagonal block $T_{ab} \neq 0$ with $a \neq b$. Hermiticity of $T$ gives $T_{ba} = T_{ab}^\dag$, and Eq.~\eqref{eq:eigenspace block decomposition} yields 
$T_{ab} \in \cS_{\lambda_1}(\cE_{ab}^\Omega)$ and $T_{ba} \in \cS_{\lambda_1}(\cE_{ba}^\Omega)$. We then define operators in $\cS_{\lambda_1}(\cE^\Omega)$ as 
\begin{align}
    T_1 = T_{ab} + T_{ba}, \ \ \ \
    T_2 = i(T_{ab} - T_{ba}).
\end{align}
They are Hermitian and linearly independent because any linear relation $c_1 T_1 + c_2 T_2 = 0$ yields $c_1 + i c_2 = 0$ on the $(a, b)$-block and $c_1 - i c_2 = 0$ on the $(b, a)$-block, forcing $c_1 = c_2 = 0$. 
This implies $\dim \cS_{\lambda_1}(\cE^\Omega) \geq 2$, which contradicts $\nu \neq 0$.
Hence, $T_{aa} \neq 0$ for some $a$, which means $\lambda_1 \in \mathrm{spec}(\cE_{aa}^\Omega)$, and thus $\lambda_1 = \lambda_1^{(a)}$ for this $a$.

We confirm the uniqueness of such an index $a$. If there are two distinct indices $a_1 \neq a_2$ satisfying $\lambda_1^{(a_i)} = \lambda_1$, we can pick nonzero $T_i \in \cS_{\lambda_1}(\cE_{a_i a_i}^\Omega)$ for each $i = 1, 2$. Since these operators have disjoint supports due to the orthogonality of the $(a_1, a_1)$- and $(a_2, a_2)$-blocks, they are necessarily linearly independent in $\cS_{\lambda_1}(\cE^\Omega)$. This contradicts the assumption that $\dim \cS_{\lambda_1}(\cE^\Omega) = 1$.
We denote this unique index by $a_\star$. As is evident from the spectral decomposition in Eq.~\eqref{eq:eigenspace block decomposition}, this uniqueness guarantees that $\lambda_1^{(a)} < \lambda_1$ for all $a \neq a_\star$.

Similarly, if $\lambda_1 \in \mathrm{spec}(\cE_{ab}^\Omega)$ for some $a \neq b$, a nonzero $T_{ab} \in \cS_{\lambda_1}(\cE_{ab}^\Omega)$ exists, which is linearly independent of $T_{a_\star a_\star } \in \cS_{\lambda_1}(\cE_{a_\star a_\star }^\Omega)$, again contradicting $\dim \cS_{\lambda_1}(\cE^\Omega) = 1$.

\medskip

\noindent\textbf{(iii) $\Leftarrow$ (iv).}
From conditions (1) and (2), we have $\cS_{\lambda_1}(\cE_{ab}^\Omega) = \{0\}$ for all pairs $(a, b)$ other than $(a_\star, a_\star)$. Then, Eq.~\eqref{eq:eigenspace block decomposition} reduces to $\cS_{\lambda_1}(\cE^\Omega) = \cS_{\lambda_1}(\cE_{a_\star a_\star}^\Omega)$.

To show that the spectral gap satisfies $\nu \neq 0$, i.e., that the largest eigenvalue $\lambda_1$ is non-degenerate, we invoke the quantum Perron--Frobenius theorem for irreducible positive maps~\cite{evans1978SpectralPropertiesPositiveMaps, Albeverio1978Frobeniustheorypositivemaps}. 
Here, an eigenvalue is \emph{simple} if it is a root of the characteristic polynomial with multiplicity one, which implies that its eigenspace is one-dimensional.
For details on the Perron--Frobenius theorem for positive maps, see Ref.~\cite[Section~6]{wolf2012guidedtour}.

\begin{lemma}[Perron--Frobenius theorem for irreducible positive maps~\cite{evans1978SpectralPropertiesPositiveMaps, Albeverio1978Frobeniustheorypositivemaps}]
\label{prevthm:perron frobenius}
Let $\cF: \cL(\cH)\to\cL(\cH)$ be an irreducible positive map; that is, the only projectors $\Pi$ satisfying $\cF(\Pi\cL(\cH)\Pi) \subseteq \Pi \cL(\cH) \Pi$ are $\Pi \in \{0, \, \bI\}$. Let $r = \sup\{|\lambda| \mid \lambda \in \mathrm{spec}(\cF)\}$ be its spectral radius.
Then $r$ is a simple eigenvalue of $\cF$, and there is a strictly positive definite operator $T > 0$, unique up to scalar multiplication, such that $\cF(T) = rT$.
\end{lemma}

It is straightforward to verify that the CP map $\cE_{a_\star a_\star}^\Omega$ restricted to $\cL(\cR_{a_\star})$ is irreducible in the sense of Lemma~\ref{prevthm:perron frobenius}.
Suppose there is a projector $\Pi \not\in \{0, \, \bI_{\cR_{a_\star}}\}$ such that $\cE_{a_\star a_\star}^\Omega(\Pi T \Pi) \in \Pi\cL(\cR_{a_\star})\Pi$ holds for any $T \in \cL(\cR_{a_\star})$. Taking $T = \Pi$ and multiplying by $(\bI_{\cR_{a_\star}}-\Pi)$ yields 
\begin{align}
    0
    &=(\bI_{\cR_{a_\star}}-\Pi)\cE_{a_\star a_\star}^\Omega(\Pi)(\bI_{\cR_{a_\star}}\!-\Pi) \\
    &= \sum_{j, l} q_l (\bI_{\cR_{a_\star}}-\Pi)E_{j|l}\big|_{\cR_{a_\star}} \! \Pi E_{j|l}\big|_{\cR_{a_\star}}\!(\bI_{\cR_{a_\star}}\!-\Pi).
\end{align}
As each term in the sum is PSD, $(\bI_{\cR_{a_\star}}-\Pi)E_{j|l}|_{\cR_{a_\star}}\Pi = 0$ for all $j$ and $l$, which means that the subspace projected by $\Pi$ is invariant under all $E_{j|l}$. This contradicts the fact that $\cR_{a_\star}$ is irreducible. Hence, $\cE_{a_\star a_\star}^\Omega$ is an irreducible map when $\cR_{a_\star}$ is irreducible.

Applying Lemma~\ref{prevthm:perron frobenius} to $\cE_{a_\star a_\star}^\Omega$, we see that its largest eigenvalue $\lambda_1^{(a_\star)}$ is simple, with a corresponding strictly positive definite operator $T^{(a_\star)} > 0$. Here, noting that $\cE_{a_\star a_\star}^\Omega$ is self-adjoint with a non-negative spectrum, $\lambda_1^{(a_\star)}$ equals the spectral radius. As a result, $\dim \cS_{\lambda_1}(\cE_{a_\star a_\star}^\Omega) = 1$. We thus conclude $\dim \cS_{\lambda_1}(\cE^\Omega) = 1$, which is equivalent to $\nu \neq 0$.

\end{proof}

We now prove Theorem~\ref{prop:qualitative equivalent condition}, which is restated below, by clarifying the structure of the unique target state under both conditions $\lambda_1 = \|W\|_\infty$ and $\nu \neq 0$.
If we assume only $\nu \neq 0$---without imposing $\lambda_1 = \|W\|_\infty$---the preceding analysis at least guarantees a unique block index $a_\star$ satisfying $\lambda_1^{(a_\star)} = \lambda_1$.
This index gives a positive definite operator $T^{(a_\star)} \in \cL(\mathcal{R}_{a_\star})$, unique up to a scalar multiple, satisfying $\cE_{a_\star a_\star}^\Omega(T^{(a_\star)}) = \lambda_1 T^{(a_\star)}$.
We write $T_\star$ for the zero extension of $T^{(a_\star)}$ to the full Hilbert space $\cH$ by setting $T_\star$ to zero on $\cR_{a_\star}^\perp$.
By vectorization and normalization, we obtain a state $\ket{\phi} = \rvec(T_\star)/\|T_\star\|_2$, which serves as the target state in Theorem~\ref{thm:1}.
This state can be identified by solving the eigenvalue equation $\cE^\Omega(T) = \lambda_1 T$; however, its explicit form is generally difficult to obtain analytically.
By further imposing $\lambda_1 = \|W\|_\infty$, the state $\ket{\phi}$ is fixed to be the uniform superposition with $\cR_{a_\star} = \cK_\star$, which is the unique $\bmE$-irreducible subspace in $\cV_W$.

\begin{restated}{theorem}{prop:qualitative equivalent condition}
For a family of POVMs $\bmE$ and a probability distribution $\bmq$, the following are equivalent:
\begin{itemize}
\item[(I)] The principal eigenspace $\cV_W$ of $W$ contains exactly one $\bmE$-irreducible subspace.
\item[(II)] $\|\Omega\|_\infty = \|W\|_\infty$ and $\nu(\Omega) \neq 0$.
\end{itemize}
Under these conditions, the principal eigenspace $\cV_\Omega$ of $\Omega^{\sA\hat{\sB}}$ is one-dimensional and is spanned by
\begin{align}
\tag{\ref*{eq:subspace mes on k star}}
\label{eq:subspace mes on k star SI}
    \ket{\Phi_\star}^{\sA\hat{\sB}} = \f{1}{\sqrt{d_{\cK_\star}}}\sum_j \ket{u_j}^\sA\ket{\bar{u}_j}^{\hat{\sB}},
\end{align}
where $\cK_{\star}$ is the unique $\bmE$-irreducible subspace in $\cV_W$, $\{\ket{u_j}^\sA\}_j$ is any orthonormal basis in $\cK_\star$, and $d_{\cK_\star}$ is the dimension of $\cK_\star$.
\end{restated}

\begin{proof}[Proof of Theorem~\ref{prop:qualitative equivalent condition}]

We first show (II) $\Rightarrow$ (I) together with the MES form.
The proof of Theorem~\ref{thm:spectral structure of irreducible subspace} shows that, under $\lambda_1 = \|W\|_\infty$ and $\nu \neq 0$, the space $\cS_W$ is one-dimensional and spanned by $T_\star \geq 0$ with $\rrange(T_\star) = \cR_{a_\star}$. By the range argument in the proof of (i) $\Rightarrow$ (ii) in \Cref{thm:spectral structure of irreducible subspace}, $\cR_{a_\star} = \rrange(T_\star)$ is $\bmE$-invariant and contained in $\cV_W$. Since $\cR_{a_\star}$ is an irreducible block, it is an $\bmE$-irreducible subspace of $\cV_W$. Its uniqueness follows since any $\bmE$-irreducible subspace $\cK \subseteq \cV_W$ gives $\Pi_\cK \in \cS_W$ by the computation in the proof of (i) $\Leftarrow$ (ii), and $\dim \cS_W = 1$ forces $\Pi_\cK \propto T_\star$, i.e., $\cK = \cR_{a_\star}$. We denote this subspace by $\cK_\star$, which establishes (I).
Since $T_\star \propto \Pi_{\cK_\star}$ as shown above, the state $\ket{\Phi_\star} = \rvec(\Pi_{\cK_\star})/\sqrt{d_{\cK_\star}}$ is the unique solution of $\Omega \ket{\Phi_\star} = \|W\|_\infty\ket{\Phi_\star}$, i.e., the MES on $\cK_\star \otimes \bar{\cK}_\star$ in Eq.~\eqref{eq:subspace mes on k star SI}.

Conversely, we show (I) $\Rightarrow$ (II). Assume that $\cV_W$ contains exactly one $\bmE$-irreducible subspace $\cK_\star$. Since $\cK_\star$ is a nonzero $\bmE$-invariant subspace contained in $\cV_W$, Theorem~\ref{thm:spectral structure of irreducible subspace} yields $\lambda_1 = \|W\|_\infty$. To establish $\nu \neq 0$, take a block decomposition in which $\cK_\star$ coincides with a single block $\cR_{a_\star}$. 
For any nonzero PSD $T \in \cS_W$, $\rrange(T)$ is a nonzero $\bmE$-invariant subspace of $\cV_W$ (as shown in the proof of (i) $\Rightarrow$ (ii)). Its irreducible decomposition consists of $\bmE$-irreducible subspaces of $\cV_W$, all of which equal $\cK_\star$ by uniqueness; hence $\rrange(T) = \cK_\star$. Since every element of $\cS_W$ decomposes into PSD elements of $\cS_W$, we conclude that $\cS_W \subseteq \Pi_{a_\star} \cL(\cH) \Pi_{a_\star}$, i.e., no eigenoperator with eigenvalue $\|W\|_\infty$ exists outside the $(a_\star, a_\star)$-block.
This translates to conditions (1) and (2) of Theorem~\ref{thm:spectral structure of irreducible subspace}, hence $\nu \neq 0$.

\end{proof}

We briefly look at an instructive example in which the subspace $\cK_\star$ in Theorem~\ref{prop:qualitative equivalent condition} depends on the choice of the probability distribution $\bmq$, even for a fixed family of POVMs $\bmE$. While $\bmE$-irreducibility depends only on the POVM elements, the principal eigenspace $\cV_W$ of $W$ depends on $\bmq$. Let $\cH \cong \bC^3$, and let $\cR_1 = \rspan\{\ket{1}\}$ and $\cR_2 = \rspan\{\ket{2}, \ket{3}\}$. Consider the two POVMs
\begin{align}
    E_1 = \Big\{\ketbra{1}{1}, \ \f{1}{2}\Big(\bI_{\cR_2} + \f{1}{\sqrt{2}}X_{\cR_2}\Big), \ \f{1}{2}\Big(\bI_{\cR_2} - \f{1}{\sqrt{2}}X_{\cR_2}\Big)\Big\}, \ \ \
    E_2 = \Big\{\f{1}{2}\ketbra{1}{1} + \ketbra{2}{2}, \ \f{1}{2}\ketbra{1}{1} + \ketbra{3}{3}\Big\},
\end{align}
where $\bI_{\cR_2}$ is the identity operator on $\cR_2$ and $X_{\cR_2} = \ketbra{2}{3} + \ketbra{3}{2}$. All the POVM elements are block-diagonal with respect to the decomposition $\cH = \cR_1 \oplus \cR_2$. The one-dimensional subspace $\cR_1$ is trivially $\bmE$-irreducible, and $\cR_2$ is $\bmE$-irreducible as the restrictions of the POVM elements to $\cR_2$ have no common eigenvector. A direct calculation yields
\begin{align}
    W = \Big(q_1 + \f{1}{2}q_2\Big)\ketbra{1}{1} + \Big(\f{3}{4}q_1 + q_2\Big)\bI_{\cR_2},
\end{align}
where $q_1 + q_2 = 1$.
Hence, $\cK_\star = \cV_W = \cR_1$ for $q_1 > 2/3$, in which case $\ket{\Phi_\star}$ is a product state, whereas $\cK_\star = \cV_W = \cR_2$ for $q_1 < 2/3$, where $\ket{\Phi_\star}$ carries one ebit. At $q_1 = 2/3$, $\cV_W = \cH$ contains both $\cR_1$ and $\cR_2$, and thus $\nu(\Omega) = 0$. This example shows that, for a fixed family $\bmE$, the choice of $\bmq$ can change $\cK_\star$ and thus the number of ebits $\log d_{\cK_\star}$ carried by the target state $\ket{\Phi_\star}$. Such a dependence on $\bmq$ does not arise for PVMs, or more generally whenever $W \propto \bI$, since $\cV_W = \cH$ for every $\bmq$ in these cases.

While Theorem~\ref{prop:qualitative equivalent condition} can be proven as shown above, there is an alternative approach based on the commutant $\mathrm{Comm}(\bmE)$ of the POVM elements $\{E_{j|l}\}_{j, l}$ restricted to the subspace $\cK \coloneqq \sum_{T \in \cS_W} \rrange(T)$.
As shown in the proof of (i) $\Rightarrow$ (ii), $\cS_W$ is spanned by its PSD elements, each of which has $\bmE$-invariant range contained in $\cV_W$. Hence $\cK$, which is also the sum of these ranges, is $\bmE$-invariant and $\cK \subseteq \cV_W$, and it is nonzero precisely when $\lambda_1 = \|W\|_\infty$. In contrast to the unique $\bmE$-irreducible subspace $\cK_\star$ in the main text, $\cK$ is defined without presupposing conditions (I) or (II).
Then, we define 
\begin{align}
    \mathrm{Comm}(\bmE) \coloneqq \big\{\hat{T} \in \cL(\cK) \mid \big[\hat{T}, Q_{j|l}\big] = 0, \ \forall j, l \big\},
\end{align}
where $Q_{j|l} \coloneqq E_{j|l}|_{\cK}$ and $[A, B] \coloneqq AB - BA$ is the commutator of operators $A$ and $B$.
This approach establishes an isomorphism between the space $\cS_W$ and $\mathrm{Comm}(\bmE)$ via the restriction map $|_{\cK}$.
Moreover, $\dim \mathrm{Comm}(\bmE) = 1$ if and only if $\cK$ is $\bmE$-irreducible.
The remainder of this section discusses the isomorphism between $\cS_W$ and $\mathrm{Comm}(\bmE)$, and also the $\bmE$-irreducibility of $\cK$.

As is clear from the discussion so far, for any $T \in \cS_W$, we have $\rrange(T) \subseteq \cK$. Its support is also contained in $\cK$, since $T \in \cS_W$ implies $T^\dag \in \cS_W$. Hence, we can write $T = \Pi_{\cK} T \Pi_{\cK}$.
Noting that $[\Pi_{\cK}, E_{j|l}] = 0$ by the $\bmE$-invariance of $\cK$, from $\|W\|_\infty T = \cE^\Omega(T)$ we see that
\begin{align}
    \|W\|_\infty \Pi_{\cK}T\Pi_{\cK} 
    = \sum_{j, l} q_l (\Pi_{\cK}E_{j|l}\Pi_{\cK}) (\Pi_{\cK}T\Pi_{\cK}) (\Pi_{\cK}E_{j|l}\Pi_{\cK}).
\end{align}
This implies that, upon restriction to the subspace $\cK$, $\|W\|_\infty\hat{T} = \sum_{j, l} q_l Q_{j|l} \hat{T} Q_{j|l}$ holds, where $\hat{T} = T|_{\cK}$.
Then, we obtain
\begin{align}
    \sum_{j, l} q_l \big\|\big[\hat{T}, Q_{j|l}\big]\big\|_2^2
    &= \tr\big[W\big|_{\cK}\big(\hat{T}^\dag\hat{T} + \hat{T}\hat{T}^\dag\big)\big]
    - 2\tr\Big[\hat{T}^\dag\sum_{j, l} q_l Q_{j|l}\hat{T}Q_{j|l}\Big] \\
    &= 2\|W\|_\infty\Big(\big\|\hat{T}\big\|_2^2 
       - \big\|\hat{T}\big\|_2^2\Big) \\
    &= 0,
\end{align}
where we used $\sum_{j, l} q_l Q_{j|l}^2 = W \big|_{\cK} = \|W\|_\infty \bI_{\cK}$, and the eigenvalue equation: $\sum_{j, l} q_l Q_{j|l}\hat{T}Q_{j|l} = \|W\|_\infty \hat{T}$. 
Thus, for any $T \in \cS_W$, its restriction $\hat{T} = T|_{\cK}$ satisfies $\big[\hat{T}, Q_{j|l}\big] = 0$ for all $l$ and $j$.

Conversely, any $\hat{T} \in \mathrm{Comm}(\bmE)$ extends to an operator $T \in \cL(\cH)$ by zero extension, i.e., by setting $T$ to zero on $\cK^\perp$. Indeed, $[\hat{T}, Q_{j|l}] = 0$ implies $Q_{j|l}\hat{T}Q_{j|l} = \hat{T}Q_{j|l}^2$, so $\sum_{j, l} q_l Q_{j|l}\hat{T}Q_{j|l} = \hat{T}\sum_{j, l} q_l Q_{j|l}^2 = \|W\|_\infty\hat{T}$, which yields $\cE^\Omega(T) = \|W\|_\infty T$, i.e., $T \in \cS_W$.
Since the restriction and the zero extension are linear and mutually inverse (recall $T = \Pi_\cK T \Pi_\cK$ for $T \in \cS_W$), we obtain the isomorphism $\cS_W \cong \mathrm{Comm}(\bmE)$.

To see that $\dim \mathrm{Comm}(\bmE) = 1$ implies the $\bmE$-irreducibility of $\cK$, we first observe that, for any subspace $\cR \subseteq \cK$, the invariance condition $E_{j|l} \cR \subseteq \cR$ is equivalent to $[E_{j|l}, \Pi_{\cR}] = 0$, where $\Pi_{\cR}$ is the projector onto $\cR$. This follows from the Hermiticity of the POVM elements: the condition $E_{j|l} \cR \subseteq \cR$ implies that $\Pi_{\cR} E_{j|l} \Pi_{\cR} = E_{j|l} \Pi_{\cR}$. Taking the Hermitian conjugate of both sides, we obtain $\Pi_{\cR} E_{j|l} \Pi_{\cR} = \Pi_{\cR} E_{j|l}$.
Combining these two equations yields $E_{j|l} \Pi_{\cR} = \Pi_{\cR} E_{j|l}$, i.e., $[E_{j|l}, \Pi_{\cR}] = 0$.

Since the identity $\bI_{\cK}$ clearly satisfies $[\bI_{\cK}, Q_{j|l}] = 0$ for all $j, l$, $\dim \mathrm{Comm}(\bmE) = 1$ implies that $\bI_{\cK}$ is the unique element in $\mathrm{Comm}(\bmE)$ up to scalar multiples. 
This means that $\cK$ is $\bmE$-irreducible as no nontrivial projector can commute with all $Q_{j|l}$.
Conversely, if $\cK$ is $\bmE$-irreducible, then any $\hat{T} \in \mathrm{Comm}(\bmE)$ must be a scalar multiple of the identity $\bI_{\cK}$. Indeed, if there is a non-scalar operator $\hat{T} \in \mathrm{Comm}(\bmE)$, then its Hermitian parts, $\hat{T}_\rR = (\hat{T} + \hat{T}^\dag)/2$ and $\hat{T}_\rI = (\hat{T} - \hat{T}^\dag)/(2i)$, also belong to $\mathrm{Comm}(\bmE)$. At least one of these must be a non-scalar Hermitian operator. The spectral projectors of such an operator then commute with all $Q_{j|l}$, thereby defining nontrivial $\bmE$-invariant subspaces. This contradicts the $\bmE$-irreducibility of $\cK$; thus, we must have $\dim \mathrm{Comm}(\bmE) = 1$.

The equivalence: $\lambda_1 = \|W\|_\infty$ and $\nu \neq 0 \Leftrightarrow \cV_W$ contains exactly one $\bmE$-irreducible subspace, follows straightforwardly from the argument up to this point. The conditions $\lambda_1 = \|W\|_\infty$ and $\nu \neq 0$ are clearly equivalent to $\dim \cS_W = 1$. Due to the isomorphism $\cS_W \cong \mathrm{Comm}(\bmE)$, this is further equivalent to $\dim \mathrm{Comm}(\bmE) = 1$, which, as we have shown, holds if and only if $\cK$ is $\bmE$-irreducible. 
In that case, $\cK$ is the unique $\bmE$-irreducible subspace contained in $\cV_W$: for any $\bmE$-irreducible $\cR \subseteq \cV_W$, we have $[\Pi_\cR, E_{j|l}] = 0$ as observed, so $\cE^\Omega(\Pi_\cR) = \Pi_\cR W \Pi_\cR = \|W\|_\infty \Pi_\cR$, i.e., $\Pi_\cR \in \cS_W$. Hence $\cR \subseteq \cK$, and the $\bmE$-irreducibility of $\cK$ forces $\cR = \cK$; we obtain $\cK = \cK_\star$, recovering Theorem~\ref{prop:qualitative equivalent condition}.


\section{Constructing POVMs from a target state}
\label{sec:construction povm from given state}

The following proposition shows that, for a target state of the form $\ket{\phi}^{\sA\hat{\sB}} = \sum_{i=1}^{d_\sA} r_i \ket{u_i}^\sA\ket{\bar{u}_i}^{\hat{\sB}}$ with $r_i>0$, one can construct a POVM for which $\ket{\phi}^{\sA\hat{\sB}}$ is the unique eigenvector with the largest eigenvalue of $\Omega^{\sA\hat{\sB}}$.
Considering only this form is sufficiently general, as any bipartite pure state on $\sA\hat{\sB}$ can be obtained from some state of this form via a local unitary on $\hat{\sB}$, which preserves the Schmidt coefficients.

\begin{proposition}
\label{prop:construction povm from given state}
    Let $\ket{\phi}^{\sA\hat{\sB}}=\sum_{i=1}^{d_\sA} r_i\ket{u_i}^\sA\ket{\bar{u}_i}^{\hat{\sB}}$ with $r_i>0$ and $\sum_i r_i^2 = 1$, in an orthonormal basis $\{\ket{u_i}\}_i$. For every $\lambda \in (0,1)$ there exists a POVM $\{E_j\}_j$ such that $\ket{\phi}^{\sA\hat{\sB}}$ is the unique eigenvector corresponding to the largest eigenvalue $\lambda$ of $\Omega^{\sA\hat{\sB}} = \sum_j E_j^\sA \otimes \bar{E}_j^{\hat{\sB}}$.
\end{proposition}

The condition $r_i>0$ is without loss of generality. For a target state of Schmidt rank $s<d_\sA$, apply the construction on its support $\cV=\rspan\{\ket{u_i}: r_i>0\}_i$ and complete the POVM on $\cV^\perp$ by adding $M > 1/\lambda$ copies of the element $\f{1}{M}\bI_{\cV^\perp}$, where $\bI_{\cV^\perp}$ is the identity operator on $\cV^\perp$. Since each $E_j$ is supported on $\cV$ or $\cV^\perp$, $\Omega$ is block-diagonal across $\cV\otimes\bar{\cV}$ and $\cV^\perp\otimes\bar{\cV}^\perp$; the cross sectors vanish. 
The operator $\Omega$ on the latter block equals $\sum_{m=1}^M \f{1}{M}\bI_{\cV^\perp}\otimes\f{1}{M}\bI_{\bar{\cV}^\perp} = \f{1}{M}\bI_{\cV^\perp\otimes\bar{\cV}^\perp}$, of spectral radius $1/M < \lambda$. The state $\ket{\phi}$ hence remains the unique top eigenvector of $\Omega$.

We now prove Proposition~\ref{prop:construction povm from given state}. We omit system superscripts and subscripts for brevity.

\begin{proof}[Proof of Proposition~\ref{prop:construction povm from given state}]

Since $0 < \lambda < 1$, we can choose a sufficiently small $\epsilon > 0$ such that, for all $i$, 
\begin{align}
    \label{inteq:82}
    0 < 1 - 2\epsilon (d-1)r_i, \  \ \ \ \ \ \ 
    0 < \lambda - 2\epsilon^2\big(1+(d-2)r_i^2\big) < \big(1 - 2\epsilon(d-1)r_i\big)^2.
\end{align}
This condition on $\epsilon$ ensures that we can choose the coefficients of the POVM elements $D_{i,n}$ constructed from $\{\ket{u_i}\}_i$.
With this $\epsilon$ fixed, for any $i$ and $k$ such that $1 \leq i < k \leq d$, we set $\ket{v^{\pm}_{ik}}=\sqrt{r_i}\ket{u_i} \pm \sqrt{r_k}\ket{u_k}$ and $F^{\pm}_{ik} = \epsilon\ketbra{v^{\pm}_{ik}}{v^{\pm}_{ik}}$. 
For each $i$, we choose positive coefficients $\{w_{i,n}\}_{n=1}^{N}$ satisfying
\begin{align}
    \label{inteq:80}
    &\sum_{n=1}^N w_{i,n}=1-2\epsilon(d-1)r_i, \\
    \label{inteq:81}
    &\sum_{n=1}^N w_{i,n}^2=\lambda-2\epsilon^2\big(1+(d-2)r_i^2\big).
\end{align}
Due to our choice of $\epsilon$, the sum in Eq.~\eqref{inteq:80} is strictly positive, and the sum of squares in Eq.~\eqref{inteq:81} is positive and strictly less than the square of the sum (i.e., $0 < \sum_n w_{i,n}^2 < (\sum_n w_{i,n})^2$). These ensure that such coefficients $w_{i,n} \in (0,1)$ exist for a sufficiently large integer $N$. Indeed, under Eq.~\eqref{inteq:80}, the sum of squares $\sum_n w_{i,n}^2$ varies over the whole interval $\big[(\sum_n w_{i,n})^2/N, (\sum_n w_{i,n})^2\big)$ as the positive coefficients are varied, and its lower end tends to zero as $N$ grows.
We set $D_{i,n}=w_{i,n}\ketbra{u_i}{u_i}$, and define $\{E_j\}_j$ as the collection of all $F^{\pm}_{ik}$ and $D_{i,n}$.
Summing $F^{+}_{ik}+F^{-}_{ik} = 2\epsilon\big(r_i\ketbra{u_i}{u_i}+r_k\ketbra{u_k}{u_k}\big)$ over all pairs $1 \leq i < k \leq d$, each operator $\ketbra{u_i}{u_i}$ appears $d-1$ times. Adding $D_{i,n}$ and applying Eq.~\eqref{inteq:80}, we obtain
\begin{align}
    \sum_j E_j 
    &= \sum_{i=1}^d \Big(2\epsilon(d-1)r_i + \sum_{n=1}^N w_{i,n}\Big) \ketbra{u_i}{u_i} \\ 
    &= \sum_{i=1}^d \ketbra{u_i}{u_i} \\
    &= \bI,
\end{align}
and thus $\{E_j\}_j$ is a POVM.

We first check that $\lambda$ is an eigenvalue of $\Omega$ with eigenvector $\ket{\phi}$.
Let $T = \rvec^{-1}(\ket{\phi}) = \sum_i r_i\ketbra{u_i}{u_i}$. Then $\Omega\ket{\phi} = \lambda\ket{\phi}$ is equivalent to $\cE^{\Omega}(T) = \lambda T$ with $\cE^{\Omega}(\cdot)=\sum_j E_j(\cdot)E_j$. 
Noting that $\bra{v^{\pm}_{ik}}T\ket{v^{\pm}_{ik}} = r_i^2 + r_k^2$, we have
\begin{align}
    \cE^{\Omega}(T) 
    &= \sum_{1 \leq i < k \leq d} 2\epsilon^2\big(r_i^2+r_k^2\big)\big(r_i\ketbra{u_i}{u_i}+r_k\ketbra{u_k}{u_k}\big) + \sum_{i=1}^d r_i\sum_{n=1}^N w_{i,n}^2\ketbra{u_i}{u_i} \\
    &= \sum_{i=1}^d r_i\Big(2\epsilon^2\sum_{k: k \neq i}\big(r_i^2+r_k^2\big)+\sum_{n=1}^N w_{i,n}^2\Big)\ketbra{u_i}{u_i} \\
    &= \sum_i r_i\Big(2\epsilon^2\big(1+(d-2)r_i^2\big)+\sum_n w_{i,n}^2\Big)\ketbra{u_i}{u_i} \\
    &= \lambda T,
\end{align}
where the second line collects the terms proportional to $\ketbra{u_i}{u_i}$, to which each pair containing $i$ contributes, and the third and fourth lines use $\sum_k r_k^2 = 1$ and Eq.~\eqref{inteq:81}, respectively.

We next show that $\cE^{\Omega}$ is irreducible, which fixes $\lambda$ as the largest eigenvalue and $T$ as the unique corresponding eigenoperator.
For any nonzero vector $\ket{\psi}=\sum_i\psi_i\ket{u_i}$, the set $\{E_j\ket{\psi}\}_j$ spans the entire space $\cH$. Indeed, if $\psi_i \neq 0$, then $D_{i,n} \ket{\psi}$ generates $\ket{u_i}$. If $\psi_k = 0$, choose any $i$ with $\psi_i \neq 0$; then the element $F^{+}$ for the pair of $i$ and $k$ maps $\ket{\psi}$ to a vector proportional to $\sqrt{r_i}\ket{u_i}+\sqrt{r_k}\ket{u_k}$, from which $\ket{u_k}$ is obtained. All basis vectors are generated.
Thus, $\cE^{\Omega}$ maps every nonzero positive semidefinite operator to a strictly positive definite one.
This implies irreducibility of $\cE^\Omega$; consider an invariant subspace associated with a projector $\Pi$, which by definition satisfies $\cE^{\Omega}(\Pi\cL(\cH)\Pi) \subseteq \Pi\cL(\cH)\Pi$. 
This condition requires the range of $\cE^{\Omega}(\Pi)$ to be confined within $\Pi\cH$.
However, for any nonzero $\Pi$, the strict positivity of $\cE^{\Omega}$ dictates that $\cE^{\Omega}(\Pi)$ is positive-definite, meaning its range must be the entire space $\cH$.
These two requirements are compatible only if $\Pi\cH = \cH$, i.e., $\Pi=\bI$.
Consequently, $\cE^{\Omega}$ is irreducible in the sense of Lemma~\ref{prevthm:perron frobenius}.
By this lemma, its spectral radius is a simple eigenvalue ($\nu(\Omega) \neq 0$) with a strictly positive operator. 
Since $\cE^{\Omega}$ is self-adjoint, eigenoperators with distinct eigenvalues are orthogonal in the Hilbert--Schmidt inner product, whereas $\tr[AB] > 0$ for any strictly positive $A$ and $B$. 
As $\cE^{\Omega}(T)=\lambda T$ with $T=\sum_i r_i\ketbra{u_i}{u_i} > 0$, the spectral radius is therefore $\lambda$ and the corresponding operator is $T$.
Hence $\ket{\phi} = \rvec(T)$ is the unique top eigenvector.

\end{proof}

This construction extends to an arbitrary number $L$ of POVMs. 
For each $l = 1, \dots, L$, let $\{E_{j|l}\}_j$ be a POVM obtained by the construction above, with its own $\epsilon_l$ and coefficients $\{w_{i,n|l}\}_n$.
For $\Omega_l=\sum_j E_{j|l}\otimes\bar{E}_{j|l}$, the map $\cE^{\Omega_l}(\cdot)=\sum_j E_{j|l}(\cdot)E_{j|l}$ satisfies $\cE^{\Omega_l}(T)=\lambda T$ and is strictly positive.
Then, for any probability distribution $\{q_l\}_{l=1}^L$, the averaged operator $\Omega=\sum_l q_l\Omega_l$ obeys $\cE^{\Omega}(T)=\sum_l q_l\cE^{\Omega_l}(T)=\lambda T$, and $\cE^{\Omega}$ is again strictly positive, hence irreducible. Thus, $\ket{\phi}$ is the unique top eigenvector of $\Omega$ with eigenvalue $\lambda$.


\section{Quantitative bounds on the spectral gap}
\label{sec:equivalent conditions positive spectral gap}

The spectral gap $\nu(\Omega)$ of $\Omega$ in Eq.~\eqref{eq:def of average of omega} appears in the denominator of our error bounds, and thus directly affects their tightness. While \Cref{sec:proof of irreducibility and unique state spectrum equivalence} qualitatively characterizes when $\nu(\Omega)$ vanishes, its value is also of interest. We here further investigate the spectral gap quantitatively.
In \Cref{sec:analytical derivation of spectral gap bounds}, we derive upper and lower bounds on the spectral gap and interpret the quantities appearing in them. In \Cref{sec:numerical evaluation of spectral gap bounds}, we evaluate the tightness of the bounds numerically using a parameterized family of bases.

\subsection{Analytical derivation of upper and lower bounds}
\label{sec:analytical derivation of spectral gap bounds}

For $\Omega=\sum_{l=1}^L q_l\Omega_l$ with $\Omega_l = \sum_j E_{j|l} \otimes \bar{E}_{j|l}$, we here assume that $\|\Omega\|_\infty = 1$ (and thus $\|W\|_\infty = 1$), and fix one corresponding eigenvector, denoted by $\ket{\phi}$. 
For each $l = 1, \ldots, L$, let
\begin{align}
    \widtil{\Omega}_l \coloneqq \Omega_l - \ketbra{\phi}{\phi}. 
\end{align}
Note that $\widtil{\Omega}_l\widtil{\Omega}_m = \Omega_l\Omega_m - \ketbra{\phi}{\phi}$ holds, because $0 \leq \Omega_l \leq \bI$ and $\bra{\phi}\Omega\ket{\phi} = 1$ imply $\Omega_l\ket{\phi} = \ket{\phi}$ for all $l$.
We then provide the following bounds on the spectral gap $\nu(\Omega) = 1-\lambda_2(\Omega)$ of $\Omega$, where $\lambda_2(\Omega)$ is the second-largest eigenvalue of $\Omega$ counted with multiplicity. 
If $\|\Omega\|_\infty$ is degenerate, $\nu(\Omega) = 0$, but the bounds remain valid as they also become zero.

\begin{proposition}
\label{prop:little general lower bound spectral gap: global}
Suppose that $\|\Omega\|_\infty = 1$.
Then, the spectral gap $\nu(\Omega)$ of $\Omega = \sum_{l=1}^L q_l\Omega_l$ satisfies 
\begin{align}
    \max\big\{(1-\mu)^2\min_l q_l, \, 1 - S_q\big\} 
\label{eq:lower bound nu general}
    \leq \nu(\Omega) 
    \leq \min\Big\{1 - \max_l q_l\gamma_l, \, \min_\pi \ell_q^\pi\big(1 - \mu_\pi^2\big)\Big\}.  
\end{align}
Here, for each permutation $\pi$ of $\{1, \ldots, L\}$, $\ell_q^\pi \coloneqq \sum_{m=1}^L m q_{\pi(m)}$, $\mu_\pi \coloneqq \big\|\widtil{\Omega}_{\pi(L)} \cdots \widtil{\Omega}_{\pi(2)}\widtil{\Omega}_{\pi(1)}\big\|_\infty$, and $\mu \coloneqq \min_\pi \mu_\pi$.
With $\gamma_l \coloneqq \|\widtil{\Omega}_l\|_\infty$ and $c_{lm} \coloneqq \big\|\widtil{\Omega}_l^{1/2}\widtil{\Omega}_m^{1/2}\big\|_\infty$,
\begin{align}
\label{eq:def of S}
    S_q 
    \coloneqq \min\bigg\{\max_l\Big(q_l\gamma_l + \sum_{m: m \neq l} \sqrt{q_lq_m}c_{lm}\Big), \sqrt{\sum_l q_l^2\gamma_l^2 + \sum_{l, m: l \neq m} q_lq_m c_{lm}^2}\bigg\}.
\end{align}
\end{proposition}

Before proceeding to the proof of Proposition~\ref{prop:little general lower bound spectral gap: global}, let us apply it to the case of $L$ projective measurements onto bases $E_l = \{\ket{e_{j|l}}\}_{j=1}^d$ and clarify the interpretation of the quantities appearing in the bounds.
In this case, $\Omega_l$ reduces to $\sum_{j=1}^d \ketbra{e_{j|l}}{e_{j|l}} \otimes \ketbra{\bar{e}_{j|l}}{\bar{e}_{j|l}}$,
which we denote by $\Pi_l$.
Since the MES $\ket{\Phi}$ is clearly one of the eigenstates corresponding to the eigenvalue $1$ of $\Omega = \sum_{l=1}^L q_l \Pi_l$, the quantity $\mu$ can be rewritten as
\begin{align}
    \mu
    &= \min_{\pi \in \mathrm{perm.}}\big\|\widtil{\Pi}_{\pi(L)}\cdots\widtil{\Pi}_{\pi(1)}\big\|_\infty \\
    &= \min_{\pi \in \mathrm{perm.}}\big\|\Pi_{\pi(L)}\cdots\Pi_{\pi(1)} - \ketbra{\Phi}{\Phi}\big\|_\infty \\
    &= \min_{\pi \in \mathrm{perm.}}\max_{\substack{\ket{T}: \ \ket{T} \neq 0 \\ \|\ket{T}\| = 1, \langle\Phi\ket{T} = 0}} \big\|\Pi_{\pi(L)}\cdots\Pi_{\pi(1)}\ket{T}\big\| \\
    \label{inteq:53}
    &= \min_{\pi \in \mathrm{perm.}}\max_{\substack{T: \ T \neq 0 \\ \|T\|_2 = 1, \tr[T] = 0}}\big\|\Delta_{\pi(L)}\circ\cdots\circ\Delta_{\pi(1)}(T)\big\|_2,
\end{align}
where $\Delta_l(\cdot) = \sum_j \bra{e_{j|l}}(\cdot)\ket{e_{j|l}}\ketbra{e_{j|l}}{e_{j|l}}$ is the completely dephasing channel in the basis. We here used the invariance of the Hilbert--Schmidt norm under vectorization defined in Eq.~\eqref{eq:definition of vectorization}.

The quantity $\mu$ is interpreted as a measure of the (ir)reducibility of the Hilbert space $\cH$.
To see this, recall that the Hilbert--Schmidt norm satisfies $\|\Delta_l(A)\|_2 \leq \|A\|_2$ for any matrix $A$ and the equality holds if and only if $A$ is diagonal in the basis. 
Thus, $\mu = 1$ implies the existence of a nonzero traceless matrix $T$ that is diagonal in all $L$ bases $\{E_l\}_{l=1}^L$. The Hilbert space $\cH$ is then decomposed as a direct sum by the eigenspaces of $T$, meaning that $\cH$ is reducible with respect to these bases. On the other hand, if $\cH$ is irreducible with respect to $L$ bases $\{E_l\}_{l=1}^L$, there is no nonzero traceless matrix that is simultaneously diagonal in all the bases, implying $\mu < 1$.

Regarding the quantities $\gamma_l$ and $c_{lm}$, in the case of basis measurements, we have $\gamma_l = \big\|\widtil{\Pi}_l\big\|_\infty = 1$ for all $l$, and $c_{lm} = \big\|\widtil{\Pi}_l\widtil{\Pi}_m\big\|_\infty = \cos \varphi_{lm}^{(2)}$, where $\varphi_{lm}^{(2)}$ is the second smallest \emph{canonical angle} (also known as principal angle)~\cite{wong1967differential, Jordan1875canonicalangle} between the subspaces $\rspan\{\ket{e_{j|l}}\ket{\bar{e}_{j|l}}\}_{j=1}^d$ and $\rspan\{\ket{e_{j|m}}\ket{\bar{e}_{j|m}}\}_{j=1}^d$.
This quantity thus quantifies the pairwise closeness between the subspaces defined by $\Pi_l$ and $\Pi_m$.

For basis measurements, let us compare the two lower bounds in Eq.~\eqref{eq:lower bound nu general}. 
For simplicity, we take $\{q_l\}_l$ to be uniform, $q_l = 1/L$, so that $S_q = S/L$ with $S \coloneqq \min\big\{\max_l(\gamma_l + \sum_{m: m\neq l}c_{lm}), \sqrt{\sum_l\gamma_l^2 + \sum_{l, m: l\neq m}c_{lm}^2}\big\}$, and the two lower bounds read $(1-\mu)^2/L$ and $1-S/L$.

The first lower bound captures global properties of all $L$ bases via $\mu$. In contrast, the second bound, $1-S/L$, depends only on the pairwise canonical angles $\big\{\varphi_{lm}^{(2)}\big\}_{l,m}$ and thus does not fully reflect collective structure.
For instance, if there exists a basis $E_l$ that decomposes $\cH$ with every other basis $E_m$, then $c_{lm}=1$ for all $m\neq l$, making the first argument in Eq.~\eqref{eq:def of S} trivial, regardless of the collective structure of $\{E_l\}_{l=1}^L$.
Similarly, if every pair $\{E_l, E_m\}$ nontrivially decomposes $\cH$, then $\sum_{l, m: l \neq m}c_{lm}^2 = L(L-1)$, and the second term in Eq.~\eqref{eq:def of S} likewise becomes trivial. Even when $\cH$ is irreducible with respect to the entire family of bases---so that Theorem~\ref{prop:qualitative equivalent condition} ensures $\nu(\Omega) > 0$---the second lower bound may yield only the trivial result $\nu(\Omega) \geq 0$.

Nevertheless, as seen in the following examples, the second bound can outperform the first in certain cases.
When $\{E_l\}_{l=1}^L$ are MUBs, we have $\varphi_{lm}^{(2)} = \pi/2$. Then, the bound is the tightest possible, yielding $1 - 1/L$ and hence $\nu(\Omega) = 1 - 1/L$, which is consistent with the result in Ref.~{\cite[Lemma~3]{Zhu2019optimalverificationfidelityestimation}}.
This is not the case for the first bound in Eq.~\eqref{eq:lower bound nu general}, as it gives $1/L$ with $\mu = 0$ even for MUBs.
In another case $L = 2$ with general bases, the second bound gives $\f{1}{2}(1 - \cos\varphi_{12}^{(2)}) \leq \nu(\Omega)$. In fact, this lower bound is tight, i.e., $\nu(\Omega) = \f{1}{2}(1 - \cos\varphi_{12}^{(2)})$~\cite{Galanta2008Subspacesangles, Zhu2013Anglesbetweensubspaces, KAUR2023lowerboundsontheminimumsingular}.

We now prove Proposition~\ref{prop:little general lower bound spectral gap: global}.
An exact reformulation of the spectral gap (Eq.~\eqref{inteq:57}) is given in this proof.

\begin{proof}[Proof of Proposition~\ref{prop:little general lower bound spectral gap: global}]

Fix an eigenvector $\ket{\phi}$ corresponding to the largest eigenvalue $1$ of the Hermitian matrix $\Omega$, and let $\cV_1^\perp \coloneqq \rspan\{\ket{\phi}\}^\perp$.
Then, $\lambda_2(\Omega)$ is given by the largest eigenvalue of $\Omega - \ketbra{\phi}{\phi}$; namely, 
\begin{align}
    \lambda_2(\Omega) 
    &= \|\Omega - \ketbra{\phi}{\phi}\|_\infty \\
    \label{inteq:37}
    &= \Big\|\sum_{l=1}^{L} q_l\widtil{\Omega}_l\Big\|_\infty \\
    \label{inteq:67}
    &= \max_{\ket{w} \in \cV_1^\perp: \|\ket{w}\| = 1}\bra{w}\sum_{l=1}^{L} q_l\widtil{\Omega}_l\ket{w}, 
\end{align}
where we defined $\widtil{\Omega}_l$ by $\widtil{\Omega}_l \coloneqq \Omega_l - \ketbra{\phi}{\phi}$.

We first derive the upper bounds on $\nu(\Omega)$. 
The first upper bound in Eq.~\eqref{eq:lower bound nu general} is obtained by building on the approach in Ref.~\cite{Zhu2019optimalverificationfidelityestimation}.
Let $\gamma_l$ be the largest eigenvalue of $\widtil{\Omega}_l$. By definition, $\widtil{\Omega}_l \geq 0$ for each $l$, which implies $\big\|\sum_l q_l\widtil{\Omega}_l \big\|_\infty \geq \max_l q_l\gamma_l$.
From this inequality and Eq.~\eqref{inteq:37}, we obtain $\lambda_2(\Omega) \geq \max_l q_l\gamma_l$, and hence
\begin{align}
    \nu(\Omega) \leq 1 - \max_l q_l\gamma_l.
\end{align}
Note that when each $\Omega_l$ is a projector with rank at least $2$, we have $\gamma_l = 1$, recovering the bound $\nu(\Omega) \leq 1 - \max_l q_l$ in Ref.~\cite{Zhu2019optimalverificationfidelityestimation}.

We next derive the second upper bound in Eq.~\eqref{eq:lower bound nu general}. 
To do so, we fix an arbitrary permutation $\pi$. Relabeling the indices as $(\widtil{\Omega}_{\pi(l)}, q_{\pi(l)}) \to (\widtil{\Omega}_l, q_l)$, we may assume $\pi = \rid$, so that $\mu_\pi = \|\widtil{\Omega}_L\cdots\widtil{\Omega}_2\widtil{\Omega}_1\|_\infty$ and $\ell_q^\pi = \sum_{l=1}^L l q_l$.
We denote $\widtil{\Omega}_l\cdots\widtil{\Omega}_2\widtil{\Omega}_1$ by $\Upsilon_l$ and define an operator $D_l \coloneqq (\bI-\widtil{\Omega}_l)\Upsilon_{l-1}$, where we set $\Upsilon_0 = \bI$; in particular, $\mu_\pi = \|\Upsilon_L\|_\infty$.
Since the operator $D_l$ is rewritten as $D_l = \Upsilon_{l-1} - \Upsilon_{l}$, we have 
\begin{align}
\label{inteq:51}
    \Upsilon_l + \sum_{m=1}^l D_m = \bI.
\end{align}
If $\mu_\pi = 0$, the bound $\nu(\Omega) \leq \ell_q^\pi(1-\mu_\pi^2)$ holds trivially, since $\lambda_2(\Omega) \geq 0$ and $\sum_l q_l = 1$ imply $\nu(\Omega) \leq 1 \leq \ell_q^\pi$; we thus assume $\mu_\pi > 0$ in the rest of this derivation.

Let $\ket{u}$ be the normalized right singular vector of $\Upsilon_L$ corresponding to $\mu_\pi$. Since $\widtil{\Omega}_1\ket{\phi} = 0$ implies $\Upsilon_L\ket{\phi} = 0$ and $\ket{u}$ is orthogonal to the kernel of $\Upsilon_L$, we have $\ket{u} \in \cV_1^\perp$. From Eq.~\eqref{inteq:67}, the spectral gap satisfies
\begin{align}
    \nu(\Omega) 
    &\leq \sum_{l=1}^L q_l\bra{u}(\bI - \widtil{\Omega}_l)\ket{u} \\
    \label{inteq:72}
    &= \sum_{l=1}^L q_l\big\|(\bI - \widtil{\Omega}_l)^{1/2} \ket{u}\big\|^2.
\end{align}
To evaluate the right-hand side, we use Eq.~\eqref{inteq:51} and obtain
\begin{align}
    \big\|(\bI - \widtil{\Omega}_l)^{1/2}\ket{u}\big\| 
    &= \Big\|(\bI - \widtil{\Omega}_l)^{1/2}\big(\Upsilon_{l-1}+\sum_{m=1}^{l-1}D_m\big)\ket{u}\Big\| \\
    \label{inteq:71}
    &\leq \big\|(\bI - \widtil{\Omega}_l)^{1/2}\Upsilon_{l-1}\ket{u}\big\|
    + \sum_{m=1}^{l-1}\big\|(\bI - \widtil{\Omega}_l)^{1/2}D_m\ket{u}\big\|,
\end{align}
where the last line follows from the triangle inequality.
Since $0 \leq \bI - \widtil{\Omega}_l \leq \bI$, for any vector $\ket{y}$, we have $\|(\bI - \widtil{\Omega}_l)^{1/2}\ket{y}\| \leq \|\ket{y}\|$. Applying this to the second term in Eq.~\eqref{inteq:71} gives 
\begin{align}
    \big\|(\bI - \widtil{\Omega}_l)^{1/2}D_m\ket{u}\big\|
    &\leq \|D_m\ket{u}\| \\
    &= \big\|(\bI - \widtil{\Omega}_m)\Upsilon_{m-1}\ket{u}\big\| \\
    \label{inteq:73}
    &\leq \big\|(\bI - \widtil{\Omega}_m)^{1/2}\Upsilon_{m-1}\ket{u}\big\|,
\end{align}
where $D_m = (\bI - \widtil{\Omega}_m)\Upsilon_{m-1}$.
Substituting Eq.~\eqref{inteq:73} into Eq.~\eqref{inteq:71}, we obtain
\begin{align}
    \big\|(\bI - \widtil{\Omega}_l)^{1/2}\ket{u}\big\| 
    &\leq \big\|(\bI - \widtil{\Omega}_l)^{1/2}\Upsilon_{l-1}\ket{u}\big\| 
    + \sum_{m=1}^{l-1} \big\|(\bI - \widtil{\Omega}_m)^{1/2}\Upsilon_{m-1}\ket{u}\big\| \\
    \label{inteq:74}
    &= \sum_{m=1}^l \big\|(\bI - \widtil{\Omega}_m)^{1/2}\Upsilon_{m-1}\ket{u}\big\|.
\end{align}

Let $x_m \coloneqq \|(\bI - \widtil{\Omega}_m)^{1/2}\Upsilon_{m-1}\ket{u}\|$ for simplicity.
Squaring Eq.~\eqref{inteq:74} and summing over $l=1, \dots, L$, we obtain
\begin{align}
\label{inteq:75}
    \sum_{l=1}^L q_l\big\|(\bI - \widtil{\Omega}_l)^{1/2}\ket{u}\big\|^2 
    \leq \sum_{l=1}^L q_l\Big(\sum_{m=1}^l x_m\Big)^2.
\end{align}
Applying the Cauchy--Schwarz inequality to the inner sum, $(\sum_{m=1}^l x_m)^2 \leq l \sum_{m=1}^l x_m^2$, and exchanging the order of summation, we have
\begin{align}
    \sum_{l=1}^L q_l\Big(\sum_{m=1}^l x_m\Big)^2 
    &\leq \sum_{l=1}^L q_l l \sum_{m=1}^l x_m^2 \\
    &= \sum_{m=1}^L x_m^2 \sum_{l=m}^L q_l l \\
    &\leq \Big(\sum_{l=1}^L q_l l\Big) \sum_{m=1}^L x_m^2,
\end{align}
where we used $\sum_{l=m}^L q_l l \leq \sum_{l=1}^L q_l l$.

We evaluate the sum of $x_m^2$. From the definition of $x_m$,
\begin{align}
    x_m^2 = \bra{u}\Upsilon_{m-1}^\dag (\bI - \widtil{\Omega}_m)\Upsilon_{m-1}\ket{u}.
\end{align}
Since $\widtil{\Omega}_m^2 \leq \widtil{\Omega}_m$, we have 
\begin{align}
    x_m^2 
    &\leq \bra{u}\Upsilon_{m-1}^\dag (\bI - \widtil{\Omega}_m^2)\Upsilon_{m-1}\ket{u} \\
    &= \big\|\Upsilon_{m-1}\ket{u}\big\|^2 - \big\|\widtil{\Omega}_m\Upsilon_{m-1}\ket{u}\big\|^2 \\
    &= \big\|\Upsilon_{m-1}\ket{u}\big\|^2 - \big\|\Upsilon_m\ket{u}\big\|^2.
\end{align}
Summing this from $m=1$ to $L$ gives
\begin{align}
    \sum_{m=1}^L x_m^2 
    &\leq \big\|\Upsilon_0\ket{u}\big\|^2 - \big\|\Upsilon_L\ket{u}\big\|^2 \\
    &= 1 - \mu_\pi^2,
\end{align}
where the last equation comes from $\|\Upsilon_L\ket{u}\| = \mu_\pi$.
Thus, we obtain
\begin{align}
\label{inteq:76}
    \sum_{l=1}^L q_l\big\|(\bI - \widtil{\Omega}_l)^{1/2}\ket{u}\big\|^2 \leq \Big(\sum_{l=1}^L q_l l\Big)(1 - \mu_\pi^2).
\end{align}

Combining Eqs.~\eqref{inteq:72} and~\eqref{inteq:76}, we arrive at the second upper bound:
\begin{align}
\label{eq:upper bound final}
    \nu(\Omega) \leq \ell_q^\pi(1 - \mu_\pi^2).
\end{align}
Since $\pi$ is arbitrary, minimizing the right-hand side over all permutations yields the second upper bound in Eq.~\eqref{eq:lower bound nu general}.

Next, to derive the lower bounds on $\nu(\Omega)$, we evaluate the sum of $\bI-\widtil{\Omega}_l$ as
\begin{align}
\label{inteq:55}
    \sum_{l=1}^L q_l(\bI-\widtil{\Omega}_l) 
    &= q_L(\bI-\Upsilon_L^\dag)(\bI-\Upsilon_L) 
    + \sum_{l=1}^{L-1}(q_l-q_{l+1})(\bI-\Upsilon_l^\dag)(\bI-\Upsilon_l) \notag \\
    &\hspace{5pc} +\sum_{l=2}^L q_l(\bI-\Upsilon_{l-1}^\dag)(\bI-\widtil{\Omega}_l)(\bI-\Upsilon_{l-1})
    + \sum_{l=1}^L q_l\Upsilon_{l-1}^\dag(\bI-\widtil{\Omega}_l)\widtil{\Omega}_l\Upsilon_{l-1},
\end{align}
where we used the identity $(\bI-\widtil{\Omega}_l) = (\bI-\Upsilon_l^\dag)(\bI-\Upsilon_l) - (\bI-\Upsilon_{l-1}^\dag)(\bI-\Upsilon_{l-1}) + (\bI-\Upsilon_{l-1}^\dag)(\bI-\widtil{\Omega}_l)(\bI-\Upsilon_{l-1}) + \Upsilon_{l-1}^\dag(\bI-\widtil{\Omega}_l)\widtil{\Omega}_l\Upsilon_{l-1}$ with $\Upsilon_l = \widtil{\Omega}_l\Upsilon_{l-1}$, together with the summation by parts $\sum_l q_l(B_l - B_{l-1}) = q_L B_L + \sum_{l=1}^{L-1}(q_l-q_{l+1})B_l$ for $B_l \coloneqq (\bI-\Upsilon_l^\dag)(\bI-\Upsilon_l)$, $B_0 = 0$. Note that terms with $\bI-\Upsilon_0$ vanish since $\bI-\Upsilon_0 = 0$.
Substituting Eq.~\eqref{inteq:55} into Eq.~\eqref{inteq:37}, we obtain
\begin{align}
\label{inteq:57}
    \nu(\Omega) 
    &= q_L\big\|(\bI-\Upsilon_L)\ket{v}\big\|^2 
    + \sum_{l=1}^{L-1}(q_l-q_{l+1})\big\|(\bI-\Upsilon_l)\ket{v}\big\|^2 \notag \\
    &\hspace{5pc}+ \sum_{l=2}^L q_l\big\|(\bI-\widtil{\Omega}_l)^{1/2}(\bI-\Upsilon_{l-1})\ket{v}\big\|^2
    + \sum_{l=1}^L q_l\big\|\big(\widtil{\Omega}_l(\bI-\widtil{\Omega}_l)\big)^{1/2}\Upsilon_{l-1}\ket{v}\big\|^2,
\end{align}
where $\ket{v}$ is the eigenvector with unit norm corresponding to the largest eigenvalue of $\sum_{l=1}^L q_l\widtil{\Omega}_l$.
Up to this point, the steps have been exact equalities, and Eq.~\eqref{inteq:57} gives the precise expression. 
We now turn to deriving a lower bound as follows.

The first bound on the left-hand side in Eq.~\eqref{eq:lower bound nu general} follows straightforwardly. 
Since $\bI-\widtil{\Omega}_m \geq 0$ for all $m$, we have $\sum_m q_m(\bI-\widtil{\Omega}_m) \geq \big(\min_l q_l\big) \sum_m(\bI-\widtil{\Omega}_m)$, and hence
\begin{align}
    \nu(\Omega)
    &\geq \Big(\min_{\ket{w} \in \cV_1^\perp: \|\ket{w}\| = 1}\bra{w}\sum_{m=1}^L(\bI-\widtil{\Omega}_m)\ket{w}\Big)\min_l q_l \\
    &\geq \big\|(\bI-\Upsilon_L)\ket{v'}\big\|^2\min_l q_l \\
    &\geq \big(1 - \big\|\Upsilon_L\ket{v'}\big\|\big)^2\min_l q_l \\
    \label{inteq:56}
    &\geq (1-\mu_\pi)^2\min_l q_l,
\end{align}
where $\ket{v'}$ is the unit eigenvector for the largest eigenvalue of $\sum_l\widtil{\Omega}_l$, so that the minimum in the first line is attained at $\ket{w} = \ket{v'}$. The second inequality follows by applying the identity Eq.~\eqref{inteq:55} with all $q_l$ set to $1$ and dropping the last two nonnegative sums. By definition, $\mu_\pi = \|\Upsilon_L\|_\infty$. Since $\pi$ is arbitrary and $(1-x)^2$ is decreasing in $x \in [0, 1]$, choosing $\pi$ attaining $\mu = \min_\pi \mu_\pi$ yields the first lower bound in Eq.~\eqref{eq:lower bound nu general}.

The derivation of the second bound on the left-hand side of Eq.~\eqref{eq:lower bound nu general} proceeds as follows.
The right-hand side of Eq.~\eqref{inteq:37} can be rewritten as follows:
\begin{align}
    \Big\|\sum_l q_l\widtil{\Omega}_l\Big\|_\infty
    &= \big\|(\sqrt{q_1}\widtil{\Omega}_1^{1/2}, \ldots, \sqrt{q_L}\widtil{\Omega}_L^{1/2})(\sqrt{q_1}\widtil{\Omega}_1^{1/2}, \ldots, \sqrt{q_L}\widtil{\Omega}_L^{1/2})^\dag\big\|_\infty \\
    &= \big\|(\sqrt{q_1}\widtil{\Omega}_1^{1/2}, \ldots, \sqrt{q_L}\widtil{\Omega}_L^{1/2})^\dag(\sqrt{q_1}\widtil{\Omega}_1^{1/2}, \ldots, \sqrt{q_L}\widtil{\Omega}_L^{1/2})\big\|_\infty \\
    \label{inteq:31}
    &= \|C\|_\infty,
\end{align}
where $C \coloneqq \big(\sqrt{q_l q_m} \widtil{\Omega}_l^{1/2}\widtil{\Omega}_m^{1/2}\big)_{l, m}$ is an $Ld^2 \times Ld^2$ matrix whose $(l, m)$-block is given by $\sqrt{q_l q_m}\widtil{\Omega}_l^{1/2}\widtil{\Omega}_m^{1/2}$.
Noting that $C$ is positive semidefinite, we further evaluate $\|C\|_\infty$ as
\begin{align}
    \|C\|_\infty 
    &= \max_{x: \|x\| = 1} x^\dag C x \\
    &= \max_{x: \|x\| = 1} \sum_{l, m} \sqrt{q_lq_m}x^{(l) \dag} \widtil{\Omega}_l^{1/2}\widtil{\Omega}_m^{1/2} x^{(m)} \\
    &\leq \max_{x: \|x\| = 1} \sum_{l, m} \sqrt{q_lq_m}\|x^{(l)}\|\|\widtil{\Omega}_l^{1/2}\widtil{\Omega}_m^{1/2} x^{(m)}\| \\
    &\leq \max_{x: \|x\| = 1} \sum_{l, m} \sqrt{q_lq_m}\|x^{(l)}\|\|\widtil{\Omega}_l^{1/2}\widtil{\Omega}_m^{1/2}\|_\infty \|x^{(m)}\| \\
    &= \max_{y: \|y\| = 1, y_l \geq 0} y^\dag \widtil{C} y \\
    \label{inteq:32}
    &\leq \big\|\widtil{C}\big\|_\infty,
\end{align}
where the vector $x^{(l)}$ denotes a $d^2$-dimensional block of the vector $x = (x^{(1) \top}, \ldots, x^{(L) \top})^\top$, and we define an $L$-dimensional vector $y = (y_l)_l^\top = (\|x^{(1)}\|, \ldots, \|x^{(L)}\|)^\top$ and an $L\times L$ matrix $\widtil{C} \coloneqq \big(\sqrt{q_l q_m}\big\|\widtil{\Omega}_l^{1/2}\widtil{\Omega}_m^{1/2}\big\|_\infty\big)_{l, m}$.
In the first inequality, we used the Cauchy--Schwarz inequality, and the second and last inequalities follow from the definition of the operator norm. 
Substituting Eqs.~\eqref{inteq:37} and~\eqref{inteq:31} into Eq.~\eqref{inteq:32}, we obtain
\begin{align}
    \label{inteq:33}
    \lambda_2(\Omega) \leq \big\|\widtil{C}\big\|_\infty.
\end{align}

So far, the calculations are common to the derivation of both terms in Eq.~\eqref{eq:def of S}. 
We now address its first term by applying a particular form of the well-known Gershgorin circle theorem~\cite{Gersgorin1931circletheorem}; see also, e.g., Refs.~\cite{Varga2004Gersgorintheorem, bhatia2013matrix}.

\begin{lemma}[Gershgorin circle theorem~\cite{Gersgorin1931circletheorem}]
\label{lem:Gershgorin thm}
For an $n \times n$ matrix $G = (g_{ij})_{i,j}$, every eigenvalue of $G$ lies in $\bigcup_{i=1}^n \{z \in \bC \mid |z - g_{ii}| \leq R_i\}$, where $R_i = \sum_{j: j\neq i}|g_{ij}|$.
\end{lemma}

We apply Lemma~\ref{lem:Gershgorin thm} to $\widtil{C}$.
The diagonal entries of $\widtil{C}$ are $q_l\|\widtil{\Omega}_l\|_\infty = q_l\gamma_l$, so we have
\begin{align}
\label{inteq:34}
    \big\|\widtil{C}\big\|_\infty 
    \leq \max_l\Big(q_l\gamma_l + \sum_{m: m\neq l} \sqrt{q_lq_m}c_{lm}\Big),
\end{align}
where $c_{lm} = \big\|\widtil{\Omega}_l^{1/2}\widtil{\Omega}_m^{1/2}\big\|_\infty$.
By Eqs.~\eqref{inteq:33} and~\eqref{inteq:34}, we obtain 
\begin{align}
    \label{eq:first bound}
    \nu(\Omega)
    \geq 1 - \max_l\Big(q_l\gamma_l + \sum_{m: m\neq l} \sqrt{q_lq_m}c_{lm}\Big),
\end{align}
which completes the derivation of the first term in Eq.~\eqref{eq:def of S}.

We next address the second term in Eq.~\eqref{eq:def of S}.
Noting that for any matrix $A = (a_{ij})_{i, j}$, $\|A\|_\infty \leq \|A\|_2$ and $\|A\|_2 = \sqrt{\sum_{i, j}|a_{ij}|^2}$ hold, we have
\begin{align}
    \big\|\widtil{C}\big\|_\infty
    &\leq \big\|\widtil{C}\big\|_2 \\
    &= \sqrt{\sum_{l, m} q_lq_m\big\|\widtil{\Omega}_l^{1/2}\widtil{\Omega}_m^{1/2}\big\|_\infty^2} \\
    \label{inteq:36}
    &= \sqrt{\sum_l q_l^2\gamma_l^2 + \sum_{l, m: l\neq m} q_lq_m c_{lm}^2}.
\end{align}
Combining Eqs.~\eqref{inteq:33} and~\eqref{inteq:36}, we obtain
\begin{align}
\label{eq:second bound}
    \nu(\Omega)
    \geq 1 - \sqrt{\sum_l q_l^2\gamma_l^2 + \sum_{l, m: l\neq m} q_lq_m c_{lm}^2},
\end{align}
which gives the second term in Eq.~\eqref{eq:def of S}.

Taken together, Eqs.~\eqref{eq:first bound} and~\eqref{eq:second bound} show that
$\nu(\Omega) \geq 1 - S_q$, where $S_q$ is given by
\begin{align}
\label{inteq:38}
    \min\bigg\{\max_l\Big(q_l\gamma_l + \sum_{m: m \neq l} \sqrt{q_lq_m}c_{lm}\Big), \sqrt{\sum_l q_l^2\gamma_l^2 + \sum_{l, m: l \neq m} q_lq_m c_{lm}^2}\bigg\}.
\end{align}
This provides the second bound on the left-hand side of Eq.~\eqref{eq:lower bound nu general}.

The lower bound on the spectral gap $\nu(\Omega)$ in Eq.~\eqref{eq:lower bound nu general}, therefore, follows from the combination of Eqs.~\eqref{inteq:56} and~\eqref{inteq:38}, completing the proof of Proposition~\ref{prop:little general lower bound spectral gap: global}.

\end{proof}


\begin{figure}
    \centering
    \includegraphics[width=90mm]{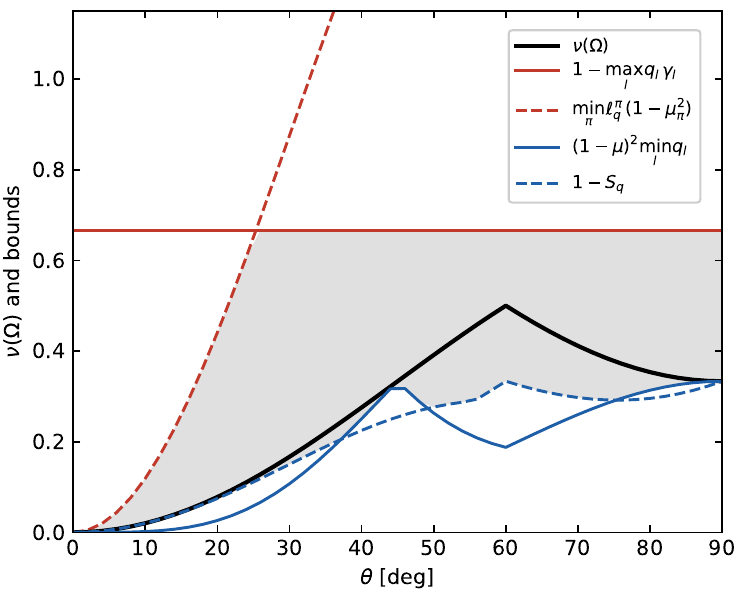}
    \caption{The exact spectral gap $\nu(\Omega)$ (black) and the bounds of Eq.~\eqref{eq:lower bound nu general} (red and blue) for the three bases $\{E_1, E_1^F(\theta), E_2^F(\theta)\}$ with $d=5$ and $q_l = 1/3$ are plotted against $\theta$. The shaded region is the interval certified by the bounds. 
    The two $\theta$-Fourier bases both coincide with the computational basis $E_1$ at $\theta = 0$, where the underlying Hilbert space is reducible, and every bound involving $\mu_\pi$ vanishes.
    At $\theta = 90^\circ$, the $\theta$-Fourier bases coincide with each other, and only two distinct bases remain, namely $E_1$ with weight $1/3$ and the Fourier basis with weight $2/3$.}
    \label{fig:spectral gap bounds}
\end{figure}


\subsection{Numerical evaluation of the spectral gap bounds}
\label{sec:numerical evaluation of spectral gap bounds}

To examine how tight the bounds in \Cref{prop:little general lower bound spectral gap: global} are, we evaluate them numerically for a concrete choice of bases.
We focus here on a parameterized family of bases built from the quantum Fourier transform~\cite{nielsen2010quantum}.

Let $F_d$ denote the $d$-dimensional quantum Fourier transform.
It admits an eigenvalue decomposition of the form $F_d = U \mathrm{diag}(e^{i\phi_1}, \ldots, e^{i\phi_d}) U^\dag$, for some unitary $U$, where $\phi_j$'s take values in $\{0, \pm\pi/2, \pi\}$.
For a parameter $\theta \in \bR$, we define the quantum $\theta$-Fourier transform as $F_d^\theta = e^{i\f{2}{\pi}\theta G}$, where $G = U\mathrm{diag}(\phi_1, \ldots, \phi_d)U^\dag$. By construction, $F_d^\theta$ coincides with $F_d$ at $\theta = \pi/2$. Using this transform, we define the $\theta$-Fourier basis $\{\ket{f_j^\theta}\}_{j=1}^d$ by $\ket{f_j^\theta} = F_d^\theta\ket{j}$, where $\{\ket{j}\}_{j=1}^d$ is the computational basis.
When $\theta = \pi/2$, this basis is mutually unbiased with respect to $\{\ket{j}\}_j$.
From this basis, we form the $\theta$-Fourier basis pair
\begin{align}
\label{eq:theta fourier basis pair}
    E_1^F(\theta) = \{\ket{f_j^\theta}\}_{j=1}^d, \ \ \
    E_2^F(\theta) = \{\ket{f_j^{-\theta}}\}_{j=1}^d.
\end{align}

Writing $E_1 = \{\ket{j}\}_{j=1}^d$ for the computational basis, we evaluate Eq.~\eqref{eq:lower bound nu general} for $\{E_1, E_1^F(\theta), E_2^F(\theta)\}$ ($L=3$) with $d=5$ and $q_l = 1/3$ for $l = 1, 2, 3$, plotting the result as a function of $\theta$ in Fig.~\ref{fig:spectral gap bounds}.
Note that, for the operator $\Omega = \sum_l q_l \Omega_l$, each $\Omega_l$ reduces to the projector $\Pi_l = \sum_{j=1}^d \ketbra{e_{j|l}}{e_{j|l}} \otimes \ketbra{\bar{e}_{j|l}}{\bar{e}_{j|l}}$, and $\gamma_l = 1$ holds for every $l$, since the measurements are onto orthonormal bases.

Neither of the two lower bounds consistently gives a tighter bound than the other: they cross at $\theta \approx 37^\circ$, $49^\circ$, and $75^\circ$.
The two upper bounds show a similar behavior: below $\theta \approx 25^\circ$, the tighter upper bound is
$\min_\pi \ell_q^\pi(1-\mu_\pi^2)$ whereas above $25^\circ$, it is $1 - \max_l q_l\gamma_l = 2/3$, which is independent of $\theta$ since $\gamma_l = 1$ in this case.
Thus, both the maximum and the minimum in Eq.~\eqref{eq:lower bound nu general} are relevant in different regimes.

The role of $\mu$ is most pronounced as $\theta \to 0$, where the three bases become identical. There, $\mu \to 1$ drives the two $\mu$-dependent bounds to zero together with $\nu(\Omega)$, quantitatively reflecting the interpretation of $\mu$ as a measure of irreducibility. At the other endpoint, $\theta=90^\circ$, both lower bounds attain $\nu(\Omega)=1/3$. The upper bound $1-\max_l q_l\gamma_l$ remains at $2/3$ only because the three weights are fixed at $q_l=1/3$; if the two coinciding $\theta$-Fourier bases are merged into a single basis with weight $2/3$, the bound also gives $1/3$.


\section{Derivations of \Cref{thm:vanishing error equivalence,prop:L_1_trivial}}
\label{sec:vanishing error equivalence}

We derive \Cref{thm:vanishing error equivalence,prop:L_1_trivial}, which characterize the equivalence between vanishing errors in one-sided LO transformation to the MES and in state discrimination.

\begin{restated}{corollary}{thm:vanishing error equivalence}
Suppose that the principal eigenspace $\cV_W$ of $W^\sA$ contains exactly one $\bmE$-irreducible subspace $\cK_\star$.
Then, the following are equivalent:
\begin{itemize}
    \item[(i)] $\|W^\sA\|_\infty = 1$.
    \item[(ii)] For every $l=1,\dots,L$, $\big\{E_{j|l}^\sA\big|_{\cK_\star}\big\}_j$ forms a PVM on $\cK_\star$.
    \item[(iii)] For every system $\sB$ and every state $\xi^{\sA\sB}$,
    \begin{align}
    \tag{\ref*{eq:vanish equivalence equation}}
    \label{eq:vanish equivalence equation SI}
    \varepsilon_{\rm opt}^{\Phi_\star}(\xi) = 0  \,\Longrightarrow\,
    \rP_{\rm opt}(\scE_{\xi,E_l}) = 1 \ \text{for all } l = 1, \dots, L.
    \end{align}
\end{itemize}
Here, $\big|_{\cK_\star}$ in (ii) denotes the restriction to $\cK_\star$, $\ket{\Phi_\star}^{\sA\hat{\sB}}$ is the MES on $\cK_\star \otimes \bar{\cK}_\star$, and $\varepsilon_{\rm opt}^{\Phi_\star}(\xi) \coloneqq \min_{\cP}\varepsilon\big(\cP(\xi),\Phi_\star\big)$, where the minimization is taken over all channels $\cP^{\sB\to\hat{\sB}}$.
\end{restated}

\begin{proof}[Proof of \Cref{thm:vanishing error equivalence}]

Since the principal eigenspace $\cV_W$ of $W$ contains exactly one $\bmE$-irreducible subspace $\cK_\star$, we have $\|\Omega\|_\infty = \|W\|_\infty$ and $\nu(\Omega) \neq 0$, so the principal eigenspace $\cV_\Omega$ of $\Omega^{\sA\hat{\sB}}$ is one-dimensional and is spanned by $\ket{\Phi_\star}^{\sA\hat{\sB}}$. This is due to \Cref{prop:qualitative equivalent condition}.
We write $Q_{j|l} \coloneqq E_{j|l}^\sA\big|_{\cK_\star}$.

\medskip

\noindent\textbf{(i) $\Leftrightarrow$ (ii).}
Let $\Pi_{\cK_\star}$ be the projector onto $\cK_\star$. 
The $\bmE$-invariance of $\cK_\star$, i.e., $E_{j|l}^\sA\cK_\star \subseteq \cK_\star$, gives $\Pi_{\cK_\star} E_{j|l}^\sA \Pi_{\cK_\star} = E_{j|l}^\sA \Pi_{\cK_\star}$, and taking the Hermitian conjugate gives $\Pi_{\cK_\star} E_{j|l}^\sA \Pi_{\cK_\star} = \Pi_{\cK_\star}E_{j|l}^\sA$; hence $[\Pi_{\cK_\star}, E_{j|l}^\sA] = 0$.
From the facts that $\sum_j E_{j|l}^\sA = \bI$ and $\sum_l q_l = 1$, we have $\sum_{j, l} q_l Q_{j|l} = \big(\sum_l q_l\sum_j E_{j|l}^\sA\big)\big|_{\cK_\star} = \bI_{\cK_\star}$.
Similarly, $\sum_{j, l} q_l Q_{j|l}^2 = \sum_{j, l} q_l (E_{j|l}^\sA)^2 \big|_{\cK_\star} = W \big|_{\cK_\star} = \|W\|_\infty \bI_{\cK_\star}$, where the last equation follows from $\cK_\star \subseteq \cV_W$.
By subtracting them, we obtain 
\begin{align}
    \sum_{j, l} q_l\big(Q_{j|l} - Q_{j|l}^2\big)
    = (1 - \|W\|_\infty) \bI_{\cK_\star}.
\end{align}
Each summand $Q_{j|l} - Q_{j|l}^2$ is positive semidefinite since $0 \leq Q_{j|l} \leq \bI_{\cK_\star}$.
As $q_l > 0$, the left-hand side vanishes if and only if $Q_{j|l}^2 = Q_{j|l}$ for all $j$ and $l$.
Together with $\sum_j Q_{j|l} = \bI_{\cK_\star}$, obtained by restricting $\sum_j E_{j|l}^\sA = \bI$ to $\cK_\star$, this implies that $\{Q_{j|l}\}_j$ forms a PVM for each $l$, i.e., condition~(ii).
The right-hand side vanishes if and only if $\|W\|_\infty = 1$, which is condition~(i).
Hence, (i) and (ii) are equivalent.

\medskip

\noindent\textbf{(i) $\Leftrightarrow$ (iii).}
We first show (i) $\Rightarrow$ (iii).
As $\cV_\Omega = \rspan\big\{\ket{\Phi_\star}^{\sA\hat{\sB}}\big\}$, the minimization over $\ket{\phi} \in \cV_\Omega$ in Eq.~\eqref{eq:converse relation} in \Cref{prop:converse relation} is trivial, making the second term on its right-hand side $\varepsilon_{\rm opt}^{\Phi_\star}(\xi)$. Under condition~(i), i.e., $\|\Omega\|_\infty = \|W\|_\infty = 1$, if $\varepsilon_{\rm opt}^{\Phi_\star}(\xi) = 0$, Proposition~\ref{prop:converse relation} yields $\big\lag \rP_{\rm opt}(\scE_{\xi, E})\big\rag \geq \|\Omega\|_\infty\big(1 - \varepsilon_{\rm opt}^{\Phi_\star}(\xi)\big) = 1$. 
Since $\big\lag \rP_{\rm opt}(\scE_{\xi, E})\big\rag = \sum_l q_l \rP_{\rm opt}(\scE_{\xi, E_l}) \leq \sum_l q_l = 1$, equality forces $\rP_{\rm opt}(\scE_{\xi, E_l}) = 1$ for all $l$.

To show (i) $\Leftarrow$ (iii), we set the initial state $\xi^{\sA\sB}$ to $\ketbra{\Phi_\star}{\Phi_\star}^{\sA\sB}$.
For this state, $\varepsilon_{\rm opt}^{\Phi_\star}(\xi) = 0$, which is attained by the identity map. By condition~(iii), this implies $\rP_{\rm opt}(\scE_{\xi, E_l}) = 1$ for all $l$, so $\big\lag \rP_{\rm opt}^2(\scE_{\xi, E})\big\rag = \sum_l q_l = 1$. 
Together with the fact that the left-hand side of Eq.~\eqref{eq:equation of main thm1} in \Cref{thm:1} is nonnegative, we obtain $1 = \big\lag \rP_{\rm opt}^2(\scE_{\xi, E})\big\rag \leq \|\Omega\|_\infty = \|W\|_\infty \leq 1$. Therefore, $\|W\|_\infty = 1$.

\end{proof}

Next, we derive \Cref{prop:L_1_trivial}.

\begin{restated}{corollary}{prop:L_1_trivial}
Suppose that the principal eigenspace $\cV_W$ contains exactly one $\bmE$-irreducible subspace $\cK_\star$.
If $\|W\|_\infty = 1$ and $L = 1$, then $d_{\cK_\star} = 1$; that is, $\ket{\Phi_\star}^{\sA\hat{\sB}}$ is a product state. Consequently, for $d_{\cK_\star} \geq 2$, the implication in~\cref{eq:vanish equivalence equation SI} requires $L \geq 2$.
\end{restated}

\begin{proof}[Proof of \Cref{prop:L_1_trivial}]

Since $\|W\|_\infty = 1$, the equivalence (i) $\Leftrightarrow$ (ii) in \Cref{thm:vanishing error equivalence} implies that, for $L=1$, a single POVM forms a PVM $\big\{Q_j \coloneqq E_{j|1}^\sA\big|_{\cK_\star}\big\}_j$ on $\cK_\star$. 
Let $r_j$ be the rank of $Q_j$. Since $\{Q_j\}_j$ is a PVM, $\sum_j r_j = d_{\cK_\star}$ holds. 
Restricting $\Omega^{\sA\hat{\sB}}
= \Omega_1^{\sA\hat{\sB}}$ to $\cK_\star \otimes \bar{\cK}_\star$ gives $\Omega^{\sA\hat{\sB}}\big|_{\cK_\star \otimes \bar{\cK}_\star} = \sum_j Q_j \otimes \bar{Q}_j$, which is an orthogonal projector of rank $\sum_j r_j^2$, as $Q_j Q_k = \delta_{jk} Q_j$, where $\delta_{jk} = 1$ if $j=k$ and $0$ otherwise.
Since all its nonzero eigenvalues equal $1$, the eigenspace $\cV_1$ of $\Omega^{\sA\hat{\sB}}$ corresponding to $\|\Omega\|_\infty = 1$ satisfies $d_{\cV_1} \geq \sum_j r_j^2$, where $d_{\cV_1}$ is the dimension of $\cV_1$.
Note that the inequality accounts for the possibility of additional eigenvectors with eigenvalue $1$ outside of the subspace $\cK_\star \otimes \bar{\cK}_\star$.
However, the condition $\nu(\Omega) \neq 0$ yields $d_{\cV_1} = 1$, so we obtain $\sum_j r_j^2 = 1$. Together with $\sum_j r_j = d_{\cK_\star}$, this implies $d_{\cK_\star} = 1$.

\end{proof}


\section{Details of the send-and-measure protocol}
\label{sec:Details of the send-and-measure protocol}

We provide details of the send-and-measure protocol discussed in Sec.~\ref{SSS:demonstration_ent_distill}. 
In Sec.~\ref{sec:calculation of bounds in a specific case}, the spectral gaps and the error bounds for the two-basis examples in the main text are derived.
In Sec.~\ref{sec:NN-hadamard simulation}, we describe the setup for the numerical study shown in Fig.~\ref{fig:numerical NNH bases simulation}, and in Sec.~\ref{sec:finding basis pairs better than MUBs}, we numerically search for basis pairs that can outperform MUBs.

\subsection{Spectral gaps and error bounds for the two-basis examples}
\label{sec:calculation of bounds in a specific case}

For a given pair of orthonormal bases $E_1 = \{ \ket{e_l} \}_{l=1}^d$ and $E_2 = \{ \ket{f_l} \}_{l=1}^d$, we here investigate the spectral properties of the operator
\begin{equation}
    \Omega_q(E_1, E_2) = q \sum_{l=1}^d \ketbra{e_l}{e_l} \otimes \ketbra{\bar{e}_l}{\bar{e}_l}
    +
    (1- q)
    \sum_{l=1}^d \ketbra{f_l}{f_l} \otimes \ketbra{\bar{f}_l}{\bar{f}_l}, \label{Eq:def_Omega_q}
\end{equation}
where $q \in (0,1)$. When it is clear, we omit $(E_1, E_2)$.

We here particularly consider two choices of basis pairs:
\begin{enumerate}
    \item A pair of MUBs.
    \item A pair of a fixed energy eigenbasis $\{ \ket{j} \}_{j=1}^d$ and the one generated from the energy eigenbasis by the neighboring-level Hadamard (NH) rotations.
\end{enumerate}

In the latter case, depending on the order of applications of the Hadamard rotations, we can define multiple bases. Here, we use the basis that is generated by applying Hadamard rotations to all odd-even nearest-neighbor energy levels $(2j-1, 2j)$, where $j=1,\dots, \lfloor d/2 \rfloor$, and then by applying the rotations to all even-odd nearest-neighboring energy levels $(2j-2, 2j-1)$, where $j=2,\dots, \lceil d/2 \rceil$. Denoting by $H_o$ and $H_e$ the former and the latter ones, the basis is given by $E_2 = \{ H_e H_o \ket{j} \}_{j=1}^d$. More general cases are treated in the next section.

The main statement in this section is the following.
\begin{proposition}\label{Prop:EV_Omega}
The non-zero eigenvalues of $\Omega_q$, defined in~\cref{Eq:def_Omega_q}, are the following.
\begin{enumerate}
    \item If $E_1$ and $E_2$ are MUBs, $\Omega_q$ has eigenvalue $1$ with multiplicity $1$, and $1-q$ and $q$ with multiplicity $d-1$.
    \item If $E_1 = \{\ket{j}\}_{j=1}^d$ and $E_2 = \{H_eH_o\ket{j}\}_{j=1}^d$, $\Omega_q$ has eigenvalue $1$ with multiplicity $1$, eigenvalues
    \begin{equation}
        \frac{1}{2}\left( 1 \pm \sqrt{(1- 2 q)^2 + 4q(1-q) \cos^2\frac{k \pi}{d}} \right),
    \end{equation}
    for $k = 1, \dots, \lfloor \frac{d-1}{2} \rfloor$, with multiplicity $1$, and $1-q$ and $q$ with multiplicity $\lfloor \frac{d}{2}\rfloor$.
\end{enumerate}
When the listed eigenvalues coincide, their multiplicities are added.
\end{proposition}

To show~\cref{Prop:EV_Omega}, the following is a useful lemma, reducing the spectral properties of $\Omega_q$ to singular values of a matrix $G$.

\begin{lemma} \label{Lemma:projectors_eigenvalues}
Let $\{\ket{e_a}\}_{a=1}^r$ and $\{\ket{f_b}\}_{b=1}^s$ be orthonormal sets of vectors, and $\Pi_1 = \sum_{a=1}^r \ketbra{e_a}{e_a}$ and $\Pi_2 = \sum_{b=1}^s \ketbra{f_b}{f_b}$ be corresponding projectors, respectively. Let $G$ be an $r \times s$ matrix with elements $G_{ab} = \braket{e_a}{f_b}$, and $\sigma_k$ be the singular values of $G$. For $q \in (0,1)$ and each $\sigma_k<1$,
\begin{equation}
    \frac{1}{2}\left( 1 \pm \sqrt{(1- 2 q)^2 + 4q(1-q) \sigma_k^2} \right) 
\end{equation}
are eigenvalues of $q \Pi_1 + (1-q)\Pi_2$.
Each singular value $\sigma_k=1$ contributes one eigenvector with eigenvalue $1$.
\end{lemma}

\begin{proof}[Proof of~\cref{Lemma:projectors_eigenvalues}]

Let $G = L \Sigma R$ be the singular value decomposition of $G$ and define 
\begin{equation}
    \ket{v_j} := \sum_{a=1}^r L_{aj} \ket{e_a}, \ \ \text{and} \ \  \ket{w_k} := \sum_{b=1}^s \bar{R}_{kb} \ket{f_b}.
\end{equation}
Using them, we can rewrite $\Pi_1$ and $\Pi_2$ as $\Pi_1 = \sum_{j=1}^r \ketbra{v_j}{v_j}$ and $\Pi_2 = \sum_{k=1}^s \ketbra{w_k}{w_k}$. They also satisfy
\begin{align}
    \braket{v_j}{w_k} &= \sum_{a=1}^{r} \sum_{b=1}^s \bar{L}_{aj} \bar{R}_{kb} \braket{e_a}{f_b}\\
    &=\sum_{a=1}^{r} \sum_{b=1}^s \bar{L}_{aj} G_{ab} \bar{R}_{kb}\\
    &= (L^{\dagger}GR^{\dagger})_{jk}\\
    &= \Sigma_{jk}\\
    &= \sigma_j \delta_{jk}.
\end{align}
This implies that $\Pi_1 \ket{w_k} = \sigma_k \ket{v_k}$ and $\Pi_2\ket{v_j} = \sigma_j \ket{w_j}$. It is also trivial that $\Pi_1\ket{v_j}= \ket{v_j}$ and $\Pi_2\ket{w_k} = \ket{w_k}$.

If $\sigma_j=0$, then $\ket{v_j}$ and $\ket{w_j}$ are eigenvectors with eigenvalues $q$ and $1-q$, respectively.
If $\sigma_j=1$, then $\ket{v_j}=\ket{w_j}$ is an eigenvector with eigenvalue $1$.
It remains to consider $0<\sigma_j<1$.

We then have, for any $j = 1, \dots, \min\{r,s \}$ with $0<\sigma_j<1$ and $\gamma \in \mathbb{R}$,
\begin{align}
    \left(q \Pi_1+ (1-q) \Pi_2\right) \left( \ket{v_j} + \gamma \ket{w_j} \right) &= q \left(1 + \gamma \sigma_j\right)\ket{v_j}  + (1-q) \left(\sigma_j + \gamma\right) \ket{w_j}.
\end{align}
This implies that, if 
\begin{equation}
    (1-q) \left(\sigma_j + \gamma\right) = q \gamma \left(1 + \gamma \sigma_j\right), \label{Eq:EV_eq}
\end{equation} 
then $\ket{v_j} + \gamma \ket{w_j}$ is an eigenvector of $q \Pi_1 + (1- q) \Pi_2$ with eigenvalue $q \left(1 + \gamma \sigma_j\right)$. By solving~\cref{Eq:EV_eq}, we have
\begin{equation}
    \gamma = \frac{1}{2q \sigma_j} \left( 1 - 2 q \pm \sqrt{(1-2 q)^2 + 4 q(1-q)\sigma_j^2} \right).
\end{equation}
Substituting this into the eigenvalue $q(1+\gamma \sigma_j)$, we obtain the desired statement.

\end{proof}

Due to~\cref{Lemma:projectors_eigenvalues}, instead of directly investigating spectral properties of $\Omega_q(E_1, E_2)$, we can check the singular values of $G$, defined by $G_{jk} = \braket{e_j \bar{e}_j}{f_k \bar{f}_k} = | \braket{e_j}{f_k}|^2$, which we call an overlap matrix of $E_1$ and $E_2$.

\begin{proposition}\label{Prop:SV_G}
    Let $E_1 = \{\ket{e_j}\}_{j=1}^d$ and $E_2 = \{\ket{f_j}\}_{j=1}^d$ be two orthonormal bases of a
    $d$-dimensional Hilbert space, and define the overlap matrix $G$ by $G_{jk} = |\braket{e_j}{f_k}|^2$.
    Then the following hold:
    \begin{enumerate}
        \item If $E_1$ and $E_2$ are MUBs, then $G$ has eigenvalue $1$ with multiplicity $1$ and eigenvalue $0$ with multiplicity $d-1$.
        \item If $E_1 = \{\ket{j}\}_{j=1}^d$ and $E_2 = \{H_eH_o\ket{j}\}_{j=1}^d$, then $G$ has singular values $\cos(k\pi/d)$, $k=0,\ldots,\lfloor \frac{d-1}{2}\rfloor$, each with multiplicity $1$. The remaining singular values are all $0$.
    \end{enumerate}
\end{proposition}

\begin{proof}[Proof of~\cref{Prop:SV_G}]

The first case of MUBs trivially follows from the fact that
\begin{equation}
    G_{jk}  = |\braket{e_j}{f_k}|^2 = \frac{1}{d}
\end{equation}
since $E_1$ and $E_2$ are MUBs. Clearly, $G$ has eigenvalue $1$ with multiplicity $1$, and the other eigenvalues are all $0$.

For $d=2$, $G_{jk}=1/2$, so the claim is immediate. Below we assume $d \geq 3$.
In the analysis of the second case, since the Hadamard rotation is real,
\begin{align}
    G &= \sum_{j,k=1}^d \bra{j} H_eH_o \ket{k}^2 \ketbra{j}{k}\\
    &= \sum_{j,k=1}^d \left( \sum_{l=1}^d  \bra{j} H_e \ket{l} \bra{l} H_o \ket{k} \right)^2 \ketbra{j}{k}\\
    &= 
    \sum_{j,k, l=1}^d \bra{j} H_e \ket{l}^2 \bra{l} H_o \ket{k}^2 \ketbra{j}{k},
\end{align}
where the third equality holds as only one term can be nonzero in the sum in the bracket. Using the notation that $A_o = \sum_{j, k} \bra{j} H_o \ket{k}^2 \ketbra{j}{k}$ and $A_e = \sum_{j, k} \bra{j} H_e \ket{k}^2 \ketbra{j}{k}$, we have $G = A_e A_o$. The matrix form of the $A_o$ and $A_e$, in the $E_1$ basis, are
\begin{align}
    A_o = \begin{cases}
        \frac{1}{2} 
    \begin{pmatrix}
        1 & 1 \\ 1 & 1
    \end{pmatrix}
    \oplus \frac{1}{2} 
    \begin{pmatrix}
        1 & 1 \\ 1 & 1
    \end{pmatrix}
    \oplus
    \dots
    \oplus
    \frac{1}{2}
    \begin{pmatrix}
        1 & 1 \\ 1 & 1
    \end{pmatrix}& \text{if $d$ is even,} \\
    \frac{1}{2} 
    \begin{pmatrix}
        1 & 1 \\ 1 & 1
    \end{pmatrix}
    \oplus 
    \frac{1}{2} 
    \begin{pmatrix}
        1 & 1 \\ 1 & 1
    \end{pmatrix}
    \oplus
    \dots 
    \oplus
    \frac{1}{2} 
    \begin{pmatrix}
        1 & 1 \\ 1 & 1
    \end{pmatrix}
    \oplus
    1& \text{if $d$ is odd,}
    \end{cases}
\end{align}
and
\begin{align}
    A_e =
    \begin{cases}
        1 \oplus 
    \frac{1}{2} 
    \begin{pmatrix}
        1 & 1 \\ 1 & 1
    \end{pmatrix}
    \oplus \frac{1}{2} 
    \begin{pmatrix}
        1 & 1 \\ 1 & 1
    \end{pmatrix}
    \oplus
    \dots
    \oplus
    \frac{1}{2}
    \begin{pmatrix}
        1 & 1 \\ 1 & 1
    \end{pmatrix}
    \oplus 
    1& \text{if $d$ is even,} \\
    1 
    \oplus 
    \frac{1}{2} 
    \begin{pmatrix}
        1 & 1 \\ 1 & 1
    \end{pmatrix}
    \oplus 
    \frac{1}{2} 
    \begin{pmatrix}
        1 & 1 \\ 1 & 1
    \end{pmatrix}
    \oplus
    \dots 
    \oplus
    \frac{1}{2} 
    \begin{pmatrix}
        1 & 1 \\ 1 & 1
    \end{pmatrix} & \text{if $d$ is odd.}
    \end{cases}
\end{align}
This necessitates to investigate the two cases, whether $d$ is even or odd, separately.
Below, we provide the eigenvalues of $G^{\dagger}G = A_oA_eA_o$, from which singular values of $G$ follow. Note that $A_o$ and $A_e$ are always projectors.  

When $d$ is even, we consider the action of $G^{\dagger}G$ to a $d$-dimensional ansatz vector in the form of
\begin{equation}
    \vec{x}_{\mathrm{double}} = \left(x_1, x_1, x_2, x_2, \dots, x_{d/2}, x_{d/2} \right)^T.
\end{equation}
A straightforward calculation leads to
\begin{equation}
    G^{\dagger} G  \vec{x}_{\mathrm{double}} = \frac{1}{4} \begin{pmatrix}
        3x_1+x_2\\ 
        3x_1+x_2\\
        x_1 + 2 x_2 + x_3\\
        x_1 + 2 x_2 + x_3\\
        \vdots \\
        x_{d/2-2} + 2 x_{d/2-1} + x_{d/2}\\
        x_{d/2-2} + 2 x_{d/2-1} + x_{d/2}\\
        x_{d/2-1} + 3 x_{d/2}\\
        x_{d/2-1} + 3 x_{d/2}
    \end{pmatrix}.
\end{equation}
This implies that the action of $G^{\dagger}G$ to $\vec{x}_{\mathrm{double}}$ is simply a double of the action of $B$ to a $d/2$-dimensional vector $\vec{x} =(x_1, x_2, \dots, x_{d/2})^T$, where $B$ is
\begin{equation}
    B
    =
    \frac{1}{4}
    \begin{pmatrix}
        3 & 1 & 0 & 0 & \dots & \dots & 0 \\
        1 & 2 & 1 & 0 & 0 & \dots & 0 \\
        0 & 1 & 2 & 1 & 0 & \dots & 0 \\
        \vdots & & &  \ddots & & & \vdots\\
        0 & \dots & 0 & 1 & 2 & 1 & 0\\
        0 & \dots & & 0 & 0 & 1 & 3
    \end{pmatrix}.
\end{equation}
This implies that $G^{\dagger}G$ has eigenvalues $0$ with multiplicity $d/2$ and that its non-zero eigenvalues are equal to the eigenvalues of $B$.

Let us assume that $x_j$ is in the following form $x_j = \cos \left(j - \frac{1}{2}\right) \theta$, for some $\theta \in [0, \pi)$. In this case, by a direct calculation, we can show that
\begin{align}
    &(B \vec{x})_1 = \frac{1}{4} \left( 3x_1 + x_2 \right) = \cos^2 \frac{\theta}{2} x_1,\\
    &(B \vec{x})_j = \frac{1}{4}\left(x_{j-1} + 2 x_j + x_{j+1} \right)
    =
    \cos^2 \frac{\theta}{2} x_j, \qquad \text{for $j = 2,\dots, d/2-1$}.
\end{align}
Hence, if
\begin{equation}
    (B \vec{x})_{d/2} = \frac{1}{4} \left( x_{d/2-1} + 3 x_{d/2} \right) = \cos^2 \frac{\theta}{2} x_{d/2} \label{Eq:trigonometric_EV_eq}
\end{equation}
holds, then it follows that $B\vec{x} = \cos^2 \frac{\theta}{2} \vec{x}$, implying that $\cos^2 \frac{\theta}{2}$ are the eigenvalues of $B$. 

Substituting $x_j = \cos (j-1/2)\theta$,~\cref{Eq:trigonometric_EV_eq} reduces to 
\begin{equation}
    \cos\frac{d-3}{2} \theta +3 \cos \frac{d-1}{2}\theta = 4 \cos^2 \frac{\theta}{2} \cos \frac{d-1}{2}\theta.
\end{equation}
Using $\cos^2 \frac{\theta}{2} = (1+ \cos \theta)/2$ and $2 \cos \theta \cos \frac{d-1}{2}\theta = \cos \frac{d+1}{2}\theta + \cos \frac{d-3}{2}\theta$, this can be simplified to 
\begin{equation}
    \cos \frac{d-1}{2} \theta = \cos \frac{d+1}{2} \theta
    \Longleftrightarrow
    \sin \frac{d}{2}\theta \sin \frac{\theta}{2} = 0.
\end{equation}
Recalling that $\theta \in [0, \pi)$, if $\theta = 2 \frac{k}{d}\pi$ for $k=0, \dots, \frac{d}{2}-1$, this is satisfied and $\cos^2 \frac{\theta}{2}$, or equivalently $\cos^2 \frac{k}{d}\pi$, are eigenvalues of $B$. 
Since the size of $B$ is $d/2$, all the eigenvalues of $B$ have been found.
This further implies that all the non-zero singular values of $G$ are $\cos \frac{k}{d}\pi$. Noting that $\frac{d}{2}-1 = \lfloor \frac{d-1}{2} \rfloor$ for even $d$, we obtain the desired statement for even $d$.

When $d$ is odd, the analysis is in parallel to the even-$d$ case with minor differences. First, we work with an ansatz vector $\vec{x}$ in the form of
\begin{equation}
    \vec{x}_{\mathrm{double}} = \left(x_1, x_1, x_2, x_2, \dots, x_{(d-1)/2}, x_{(d-1)/2}, x_{(d+1)/2} \right)^T.
\end{equation}
Then, a direct analysis shows that $G^{\dagger}G$ has zero eigenvalues with multiplicity $(d-1)/2$. Furthermore, the remaining eigenvalues are equal to those of the matrix $B$ given by 
\begin{equation}
    B
    =
    \frac{1}{4}
    \begin{pmatrix}
        3 & 1 & 0 &  \dots & \dots & 0 \\
        1 & 2 & 1 & 0 &  \dots & 0 \\
        0 & 1 & 2 & 1 &  \dots & 0 \\
        \vdots & & &  \ddots & & \vdots\\
        0 & \dots & 0 & 1 & 2 & 1 \\
        0 & \dots & & 0 & 2 & 2 
    \end{pmatrix},
\end{equation}
whose size is $(d+1)/2$. Note that the last row of $B$, as well as its size, are different from the even-$d$ case. 

This difference is then reflected in the condition on $x_{(d+1)/2}$. Instead of imposing~\cref{Eq:trigonometric_EV_eq}, the following should hold for $\cos^2 \frac{\theta}{2}$ to be eigenvalues of $B$:
\begin{equation}
    \frac{1}{2}\left( x_{(d-1)/2} + x_{(d+1)/2} \right) = \cos^2 \frac{\theta}{2} x_{(d+1)/2},
\end{equation}
where we have again assumed that $x_j = \cos (j-\frac{1}{2})\theta$. An elementary calculation similar to the even-$d$ case shows that, if $\theta= \frac{2k}{d}\pi$, where $k = 0, \dots, \frac{d-1}{2}$, then $\cos^2 \frac{\theta}{2}$, or equivalently, $\cos^2 \frac{k}{d}\pi$ are eigenvalues of $B$.
    
\end{proof}

Finally, the proof of~\cref{Prop:EV_Omega} follows simply by the combination of~\cref{Lemma:projectors_eigenvalues,Prop:SV_G}.

\begin{proof}[Proof of~\cref{Prop:EV_Omega}]

The statement is obtained by substituting the singular values of the overlap matrix $G$ of each pair of $E_1$ and $E_2$, which are provided in~\cref{Prop:SV_G}, into~\cref{Lemma:projectors_eigenvalues}.
    
\end{proof}

\subsection{Setup for the numerical study}
\label{sec:NN-hadamard simulation}

We here describe the detailed setup used to obtain Fig.~\ref{fig:numerical NNH bases simulation}, which evaluates the send-and-measure protocol in Sec.~\ref{SSS:demonstration_ent_distill} using MUBs and bases generated by neighboring-level Hadamard (NH) operations. While we have considered only one basis generated by $H_eH_o$ from the energy eigenbasis $E_1 = \{\ket{j}\}_{j=1}^d$, multiple bases can be similarly defined because nearest-neighbor Hadamard rotations do not commute when they share a common energy level. Hence, different orders of these rotations produce different bases.

Taking all possible orders into account, we obtain at most $1 + (d-1)!$ distinct bases, including $E_1$.
We then select $L$ bases out of them, choose one of these $L$ bases uniformly at random with probability $q_l = 1/L$ for each $l$, and let Alice measure system $\sA$ in the corresponding conjugate basis $\{\ket{\bar{e}_{j|l}}^{\sA}\}_{j=1}^d$.
In this case, the operator $\Omega^{\sA\hat{\sB}}$ reduces to $\f{1}{L}\sum_{l=1}^L \Pi_l^{\sA\hat{\sB}}$, where
\begin{align}
\label{eq:projector definition equation}
    \Pi_l^{\sA\hat{\sB}} = \sum_{j=1}^{d_{\sA}}\ketbra{\bar{e}_{j|l}}{\bar{e}_{j|l}}^\sA\otimes\ketbra{e_{j|l}}{e_{j|l}}^{\hat{\sB}}.
\end{align}
This operator satisfies $\|\Omega\|_\infty = 1$, and the MES $\ket{\Phi}^{\sA\hat{\sB}}$ is a corresponding eigenstate. When $\nu(\Omega) \neq 0$, these imply $\cK_\star = \cH^\sA$, so the target MES has $\log d$ ebits.

For example, consider the case $d=3$, and let 
\begin{align}
    H_{(1, 2)} =
    \begin{pmatrix}
    \f{1}{\sqrt{2}} & \f{1}{\sqrt{2}} & 0 \\
    \f{1}{\sqrt{2}} & -\f{1}{\sqrt{2}} & 0 \\
    0 & 0 & 1
\end{pmatrix}, \ \ \ 
    H_{(2, 3)}=
    \begin{pmatrix}
    1 & 0 & 0 \\
    0 & \f{1}{\sqrt{2}} & \f{1}{\sqrt{2}} \\
    0 & \f{1}{\sqrt{2}} & -\f{1}{\sqrt{2}}
\end{pmatrix},
\end{align}
which act as Hadamard operations on the level pairs (1,2) and (2,3), respectively.
The bases $\{E_1, E_2, E_3\}$ are generated from $E_1 = \{\ket{1}, \ket{2}, \ket{3}\}$ as follows:
\begin{align}
    &E_2 
    = H_{(2, 3)}H_{(1, 2)}E_1 \\
    \label{eq:nn-basis 1}
    &\hspace{1pc}= \Big\{\f{\sqrt{2}\ket{1}+\ket{2}+\ket{3}}{2}, \f{\sqrt{2}\ket{1}-\ket{2}-\ket{3}}{2}, \f{\ket{2}-\ket{3}}{\sqrt{2}}\Big\}, \\
    &E_3 
    = H_{(1, 2)}H_{(2, 3)}E_1 \\
    \label{eq:nn-basis 2}
    &\hspace{1pc}= \Big\{\f{\ket{1}+\ket{2}}{\sqrt{2}}, \f{\ket{1}-\ket{2}+\sqrt{2}\ket{3}}{2}, \f{\ket{1}-\ket{2}-\sqrt{2}\ket{3}}{2}\Big\}.
\end{align}
For this family of bases, we calculate $\nu(\Omega) \approx 0.362$.

As a noise model, we consider the \emph{$r$-mixed $\alpha$-depolarizing and $\beta$-dephasing noise}:
\begin{align}
\label{eq:def of p-mixed depol-dephase noise}
    \cN_r^{\sA'\to\sB}(\cdot)
    \coloneqq (1-r) \cD_\alpha^{\sA'\to\sB}(\cdot) + r\Delta_\beta^{\sA'\to\sB}(\cdot),
\end{align}
where $r, \alpha, \beta \in [0, 1]$ and $\sA'\cong\sB$. The channels $\cD_\alpha$ and $\Delta_\beta$ are the depolarizing and dephasing channels with strengths $\alpha$ and $\beta$, respectively.
They act on a state $\rho$ as $\cD_\alpha(\rho) = (1-\alpha)\rho + \alpha\bI/d$, and $\Delta_\beta(\rho) = (1-\beta)\rho + \beta \diag(\rho)$, where $\diag(\rho)$ is the diagonal part of $\rho$ in the computational basis.
The noisy state $\xi^{\sA\sB}$ is then given by $\xi^{\sA\sB} = \cN_r^{\sA'\to\sB}(\ketbra{\Phi}{\Phi}^{\sA\sA'})$, where $\ket{\Phi}^{\sA\sA'}$ is the MES.

In the send-and-measure implementation, Alice's measurement on $\sA$ in the conjugate basis is replaced by directly sending a uniformly random basis state $\ket{e_{j|l}}^{\sA'}$, and Bob measures the received state in the same basis $F_l^\sB = \{\ket{e_{j|l}}^\sB\}_{j=1}^d$.
The coincidence probability for the basis is
\begin{align}
    p_l = \f{1}{d}\sum_{j=1}^d \bra{e_{j|l}}^\sB \cN_r^{\sA'\to\sB}\big(\ketbra{e_{j|l}}{e_{j|l}}^{\sA'}\big)\ket{e_{j|l}}^\sB,
\end{align}
which can be estimated experimentally.
Following \cref{cor:entanglement_distillation} with exact coincidence probabilities, $\log d$ ebits can be LOCC distilled with infidelity
\begin{align}
    \label{eq:numeric main equation}
    \varepsilon_{\mathrm{distill}}
    \leq \f{1 - C_{\rm coin}(\bmE, \bmF|\xi)}{\nu(\Omega)},
\end{align}
where $C_{\rm coin}(\bmE, \bmF|\xi)= \f{1}{L}\sum_{l=1}^L p_l^2$.
The spectral gap $\nu(\Omega)$ is determined only by the choice of bases and can be computed in advance.

To check the tightness of the bound in Eq.~\eqref{eq:numeric main equation}, we compare the NH-based bound with the MUB-based one and with the optimal error $\varepsilon_{\mathrm{distill}}$ computed via semidefinite programming (SDP).
In Fig.~\ref{fig:numerical NNH bases simulation}, we plot the left- and right-hand sides of Eq.~\eqref{eq:numeric main equation} for $d=5$ as a function of the mixing probability $r \in [0, 1]$, for the noise parameters $(\alpha, \beta) = (0.01, 0.03), (0.02, 0.02), (0.03, 0.01)$.
For each $L = 2, 4, 6$ and each value of $r$, the NH-based bound (solid) is obtained by optimizing the right-hand side of Eq.~\eqref{eq:numeric main equation} over the choice of $L$ bases among the at most $1+(d-1)!$ NH bases. The MUB-based bound (dashed) is obtained analogously from the $d+1$ MUBs, which exist since $d=5$ is prime.
The optimal error $\varepsilon_{\mathrm{distill}}$ is plotted as a black dotted line.
The behavior of these bounds is discussed in Sec.~\ref{SSS:demonstration_ent_distill}.


\subsection{Numerical search for basis pairs outperforming MUBs}
\label{sec:finding basis pairs better than MUBs}

We numerically search for a pair of orthonormal bases that yields a tighter upper bound on the distillation error $\varepsilon_{\mathrm{distill}}$.
We focus on PVMs onto two orthonormal bases with $L=2$ and $q_1=q_2=1/2$ in the single-qubit case ($d=2$), and for a noisy state, we take $\xi^{\sA\sB} = \cN_r^{\sA'\to\sB}(\ketbra{\Phi}{\Phi}^{\sA\sA'})$, where $\cN_r^{\sA'\to\sB}$ is the $r$-mixed $\alpha$-depolarizing and $\beta$-dephasing noise defined in Eq.~\eqref{eq:def of p-mixed depol-dephase noise}.
The parameters are set as $r=0.89$, $\alpha=0.01$, and $\beta=1.0$.

We now consider two pairs of bases: the $\theta$-Fourier basis pair $E_1^F(\theta), E_2^F(\theta)$ defined in Eq.~\eqref{eq:theta fourier basis pair} with $d=2$, and the $Y$-axis-rotated MUB pair given by 
\begin{align}
    E_1^Y = \{e^{i\pi Y/8}\ket{j}\}_{j=1}^2, \ \ \
    E_2^Y = \{e^{-i\pi Y/8}\ket{j}\}_{j=1}^2, \label{eq: Y-axis-rotated MUB pair}
\end{align}
where $Y$ is the qubit Pauli-$Y$ operator. 
The choice of these bases is motivated by the property of the noise $\cN_r^{\sA'\to\sB}$, which affects states near the equatorial plane more strongly than those close to the $Z$ axis. The $Y$-axis-rotated MUB pair balances the noise effects between the two bases, yielding the smallest $u_2$ among MUB pairs.
In contrast, the $\theta$-Fourier basis pair, constructed symmetrically with respect to the $XZ$ plane, allows us to explore whether non-MUB pairs can yield a tighter bound than the $Y$-axis-rotated MUB pair.


\begin{figure}
    \centering
    \includegraphics[width=130mm]{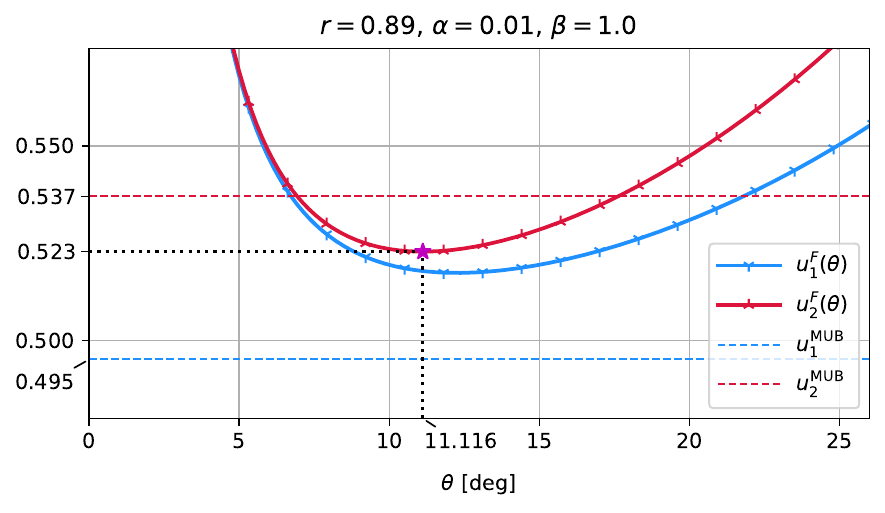}
    \caption{The upper bounds $u_1$ in Eq.~\eqref{eq:u_1} (blue) and $u_2$ in Eq.~\eqref{eq:u_2} (red) are plotted for the $\theta$-Fourier basis pair (solid) and the $Y$-axis-rotated MUB pair (dash-dotted) in the single-qubit case ($d=2$). We set the noisy state to $\xi^{\sA\sB} = \cN_r^{\sA'\to\sB}(\ketbra{\Phi}{\Phi}^{\sA\sA'})$ with $r=0.89$, $\alpha = 0.01$, and $\beta= 1.0$. The $\theta$-Fourier basis pair gives a lower $u_2$ than the $Y$-axis-rotated MUB pair for $\theta \in (6.9^\circ, 17.6^\circ)$, while the latter is always superior for $u_1$.}
    \label{fig:UB comparison}
\end{figure}


We consider two upper bounds $u_1$ and $u_2$ on the distillation error $\varepsilon_{\mathrm{distill}}$:
\begin{itemize}
    \item $u_1$ is an upper bound based on the PGMs $\Gamma_{\scE_{\xi, E_l}}$ for $l=1, 2$ (the definition of the PGM can be found in Supplementary Note~\ref{sec:proof thm1 new}):
\begin{align}
\label{eq:u_1}
    u_1 
    \coloneqq \f{1 - \f{1}{2}\sum_{l=1}^2 \rP\big(\Gamma_{\scE_{\xi, E_l}}|\scE_{\xi, E_l}\big)}{\nu\big(\f{1}{2}\sum_{l=1}^2 \Pi_l\big)},
\end{align}
    \item $u_2$ is an upper bound obtained by applying the Barnum--Knill inequality to $u_1$ (see also Eqs.~\eqref{inteq:2} and~\eqref{inteq:22}):
\begin{align}
\label{eq:u_2}
    u_2 
    \coloneqq \f{1 - \f{1}{2}\sum_{l=1}^2\big(\rP_{\rm opt}(\scE_{\xi, E_l})\big)^2}{\nu\big(\f{1}{2}\sum_{l=1}^2 \Pi_l\big)}.
\end{align}
\end{itemize}
We assume here that the spectral gap $\nu(\Omega)$ is nonzero, in which case $\cK_\star = \cH^\sA$ and the bounds certify that $1$ ebit is LOCC distillable with infidelity $\varepsilon_{\mathrm{distill}} \leq u_i$ for $i = 1, 2$.
For each of the above two basis pairs, $\nu(\Omega)$ is computed by substituting the basis elements into $\Pi_l = \sum_{j=1}^2 \ketbra{e_{j|l}}{e_{j|l}} \otimes \ketbra{\bar{e}_{j|l}}{\bar{e}_{j|l}}$, where $\{\ket{e_{j|l}}\}_{j=1}^2$ denotes the $l$-th basis of the pair, $E_l^F(\theta)$ or $E_l^Y$.

In the above setting, we plot the values of $u_1$ and $u_2$ in Fig.~\ref{fig:UB comparison} as a function of $\theta$.
For the bound $u_2$, the $\theta$-Fourier basis pair outperforms the $Y$-axis-rotated MUB pair for $\theta \in (6.9^\circ, 17.6^\circ)$, attaining its minimum value of $0.523$ at $\theta \approx 11.1^\circ$.
In contrast, for $u_1$, the $Y$-axis-rotated MUB pair is always superior for the parameters used.

As a remark, the $Y$-axis-rotated MUB pair defined in \eqref{eq: Y-axis-rotated MUB pair} achieves the minimum values of both $u_1$ and $u_2$ among the pairs obtained by rotating its bases by the same angle about the $Y$ axis. For any pair of bases obtained in this way, denoted by $E_1$ and $E_2$, a direct calculation gives
 \begin{align}
    &\sum_{l=1}^2
    \rP\big(\Gamma_{\scE_{\xi,E_l}}|\scE_{\xi,E_l}\big)
    = 1+\f{\lambda_x^2+\lambda_z^2}{2},
    \label{eq:PGM common rotation}
    \\
    &\sum_{l=1}^2 \big(\rP_{\rm opt}(\scE_{\xi,E_l})\big)^2
    \leq 2 \bigg(\f{1+\sqrt{(\lambda_x^2+\lambda_z^2)/2}}{2}\bigg)^2
    \label{eq:optimal common rotation}.
\end{align}
Here, $\lambda_x=(1-r)(1-\alpha)+r(1-\beta)$ and $\lambda_z=(1-r)(1-\alpha)+r$ are the contraction factors of the Bloch vector along the $X$ and $Z$ axes, respectively.
We can see that Eq.~\eqref{eq:PGM common rotation} does not depend on the rotation angle. Since the spectral gap is also unchanged when both bases are rotated by the same angle, $u_1$ has the same value for every rotation angle. By contrast, the upper bound in Eq. \eqref{eq:optimal common rotation} is attained by the $Y$-axis-rotated MUB pair defined in \eqref{eq: Y-axis-rotated MUB pair}. Therefore, $u_2$ is minimized by this pair of bases.

\section{Characterization of the one-shot information measures for quantum measurements}
\label{sec:further discussion on information gain disturbance}

We here discuss the one-shot measures of quantum measurements introduced in Sec.~\ref{sec:information gain and irreversibility}.
In Sec.~\ref{sec:upperbound on the info gain}, we prove Proposition~\ref{prop:upperbound infogain} regarding the information gain, and in Sec.~\ref{sec:upperbound on q-distub} we derive an upper bound on the irreversibility. In Sec.~\ref{sec:proof tradeoff onfo gain and q-disturb}, we prove their trade-off relation, stated as Proposition~\ref{thm:Trade-off relation between the information gain and irreversibility}.
Finally, in Sec.~\ref{sec:proof of disturbance noise trade offs}, we prove \Cref{prop:one-shot noise disturbance}, the trade-off relation between the disturbance and noise for an observable.

\subsection{Proof of Proposition~\ref{prop:upperbound infogain}}
\label{sec:upperbound on the info gain}

We here prove that any rank-$1$ projective measurement maximizes the one-shot information gain $\mathrm{IG}(\sigma, \Lambda)$.

\begin{restated}{proposition}{prop:upperbound infogain}
For an arbitrary orthonormal basis $E = \{\ket{e_m}^{\sA'}\}_m$, let $\Lambda_E^{\sA'\to\sM}$ be the rank-1 projective measurement defined by
\begin{align}
\tag{\ref*{eq:rank-1 proj measure}}
\label{eq:rank-1 proj measure SI}
    \Lambda_E^{\sA'\to\sM}(\cdot) 
    = \sum_m \tr_{\sA'}\big[\ketbra{e_m}{e_m}^{\sA'}(\cdot)\big] \ketbra{m}{m}^\sM.
\end{align}
Then, for a state $\ket{\sigma}^{\sA\sA'}$, it achieves the maximal value of the information gain, $\mathrm{IG}(\sigma, \Lambda_E) = -\log{\tr[(\sigma^{\sA'})^2]}$.
\end{restated}

\begin{proof}[Proof of Proposition~\ref{prop:upperbound infogain}]

Recall that the information gain $\mathrm{IG}(\sigma, \Lambda)$ for a state $\ket{\sigma}^{\sA\sA'}$ and a measurement $\Lambda^{\sA'\to\sM} = \tr_\sB\circ\cM^{{\sA'}\to\sB\sM}$ can be written as
\begin{align}
\label{inteq:4}
    \mathrm{IG}(\sigma, \Lambda) = -\log \max_{\tau^\sM} \rF(\Theta^{{\sA}\sM}, \sigma^{\sA} \otimes \tau^\sM),
\end{align}
where $\Theta^{{\sA}\sM} = \Lambda^{\sA'\to\sM}(\ketbra{\sigma}{\sigma}^{{\sA}{\sA'}}) = \sum_m p_m \theta_m^{\sA} \otimes \ketbra{m}{m}^\sM$, $\{p_m\}_m$ is a probability distribution, and $\{\theta_m^{\sA}\}_m$ is a family of states on ${\sA}$.
Note that $\Theta^{\sA} = \sum_m p_m \theta_m^{\sA} = \sigma^{\sA}$.
We first show that the information gain can also be written as
\begin{align}
\label{eq:infogain rewrite}
    \mathrm{IG}(\sigma, \Lambda) = -\log{\sum_m p_m \rF(\theta_m^{\sA}, \sigma)}.
\end{align}

In general, fidelity is monotonically nondecreasing under any quantum channel $\cT$: $\rF\big(\cT(\rho), \cT(\sigma)\big) \geq \rF(\rho, \sigma)$ for any states $\rho$ and $\sigma$.
Since $\Theta^{{\sA}\sM}$ is invariant under the dephasing channel in the basis $\{\ket{m}^\sM\}_m$, the state $\tau^\sM$ that maximizes Eq.~\eqref{inteq:4} can be taken to be diagonal in this basis, i.e., $\tau^\sM = \sum_m \alpha_m \ketbra{m}{m}^\sM$ for some probability distribution $\{\alpha_m\}_m$.
Then, $\max_{\tau^\sM} \rF(\Theta^{{\sA}\sM}, \sigma^{\sA} \otimes \tau^\sM)$ is calculated as
\begin{align}
    \max_{\tau^\sM} \rF(\Theta^{{\sA}\sM}, \sigma^{\sA} \otimes \tau^\sM) 
    = \max_{\{\alpha_m\}_m}\Big(\sum_m \sqrt{\alpha_m p_m} \tr\Big[\sqrt{\sqrt{\theta_m}\sigma\sqrt{\theta_m}}\Big]\Big)^2,
\end{align}
where the maximization over $\{\alpha_m\}_m$ is taken under the conditions $\alpha_m \geq 0$ and $\sum_m \alpha_m = 1$.
Using the Cauchy--Schwarz inequality, we can further proceed as
\begin{align}
    \max_{\{\alpha_m\}_m}\Big(\sum_m \sqrt{\alpha_m p_m} \tr\Big[\sqrt{\sqrt{\theta_m}\sigma\sqrt{\theta_m}}\Big]\Big)^2
    &= \sum_m p_m \Big(\tr\Big[\sqrt{\sqrt{\theta_m}\sigma\sqrt{\theta_m}}\Big]\Big)^2 \\
    &= \sum_m p_m \rF(\theta_m^{\sA}, \sigma^{\sA}),
\end{align}
where $\alpha_m$'s that achieve this equality are given by $\alpha_m = p_m \rF(\theta_m^{\sA}, \sigma^{\sA}) / \sum_m p_m \rF(\theta_m^{\sA}, \sigma^{\sA})$.
By taking $-\log$, we obtain Eq.~\eqref{eq:infogain rewrite}.

We now consider the measurement that maximizes $\mathrm{IG}(\sigma, \Lambda)$. 
The upper bound on $\mathrm{IG}(\sigma, \Lambda)$ is evaluated as
\begin{align}
    \mathrm{IG}(\sigma, \Lambda) 
    &= -\log{\sum_m p_m \rF(\theta_m^{\sA}, \sigma^{\sA})} \\
    \label{inteq:5}
    &\leq -\log \sum_m p_m \tr\big[\theta_m^{\sA}\sigma^{\sA}\big] \\
    &= -\log \tr\big[(\sigma^{\sA})^2\big] \\
    &= -\log \tr\big[(\sigma^{\sA'})^2\big],
\end{align}
where in the inequality we used that $\sqrt{\tr[A^\dag A]} = \|A\|_2 \leq \|A\|_1$ and $\rF(\theta_m, \sigma) = \|\sqrt{\theta_m}\sqrt{\sigma}\|_1^2$.
The last line follows from the fact that $\ket{\sigma}^{{\sA}{\sA'}}$ is pure.
Since $\|A\|_2 = \|A\|_1$ holds only when $A$ has rank at most one, equality in Eq.~\eqref{inteq:5} holds only when $\sqrt{\theta_m}\sqrt{\sigma}$ has rank one. As $\sigma^{\sA} = \sum_m p_m \theta_m^{\sA}$, each $\theta_m^{\sA}$ must be rank one, i.e., a pure state. This condition can be achieved by any rank-$1$ projective measurement on system ${\sA'}$, as we show below.

The initial state $\ket{\sigma}^{{\sA}{\sA'}}$ can be expressed as $\ket{\sigma}^{{\sA}{\sA'}} = \sum_i \sqrt{p_i} \ket{\phi_i}^{\sA} \ket{e_i}^{\sA'}$ by using an arbitrary orthonormal basis $E = \{\ket{e_i}^{\sA'}\}_i$, where the states $\{\ket{\phi_i}^{\sA}\}_i$ are proportional to $\bra{e_i}^{\sA'}\ket{\sigma}^{\sA\sA'}$ but are not necessarily mutually orthogonal.
Then, for a rank-$1$ projective measurement in the basis $E$, i.e.,
\begin{align}
    \Lambda_E^{\sA'\to\sM}(\cdot) 
    = \sum_m \tr_{\sA'}\big[\ketbra{e_m}{e_m}^{\sA'}(\cdot)\big] \ketbra{m}{m}^\sM, 
\end{align}
the post-measurement state $\Theta^{{\sA}\sM}$ is given by
\begin{align}
    \Theta^{{\sA}\sM} 
    &= \sum_m \tr_{\sA'}[\ketbra{e_m}{e_m}^{\sA'} \ketbra{\sigma}{\sigma}^{{\sA}{\sA'}}] \otimes \ketbra{m}{m}^\sM \\
    &= \sum_m p_m \ketbra{\phi_m}{\phi_m}^{\sA} \otimes \ketbra{m}{m}^\sM.
\end{align}
The state on ${\sA}$ after the measurement $\Lambda_E^{\sA'\to\sM}$ is the pure state $\ket{\phi_m}^{\sA}$.  
Thus, the information gain is maximized by this measurement, achieving $\mathrm{IG}(\sigma, \Lambda_E) = -\log{\tr[(\sigma^{\sA'})^2]}$.

\end{proof}


\subsection{Upper bound on the one-shot irreversibility}
\label{sec:upperbound on q-distub}

We show that the irreversibility $\mathrm{Irv}(\sigma, \cM)$ is upper bounded by $-2\log\|\sigma^{\sA'}\|_\infty$.
Recall that the irreversibility is defined by
\begin{align}
\label{inteq:30}
    \mathrm{Irv}(\sigma, \cM) 
    \coloneqq -\log \max_\cP \rF(\cP^{\sB\sM\to\hat{\sB}}(\Theta^{{\sA}\sB\sM}), \ketbra{\sigma}{\sigma}^{{\sA}\hat{\sB}}),
\end{align}
for a state $\ket{\sigma}^{{\sA}{\sA'}}$ and a measurement process $\cM^{{\sA'}\to\sB\sM}$, where $\Theta^{{\sA}\sB\sM} = \cM^{{\sA'}\to\sB\sM}(\ketbra{\sigma}{\sigma}^{{\sA}{\sA'}})$. The maximization is over channels $\cP^{\sB\sM\to\hat{\sB}}$.
Our approach is to consider a measurement intentionally designed to be most disruptive to the system, where the system ${\sA'}$ is discarded to the environment $\sE$, followed by the preparation of states $\{\ket{m}^\sM\}_m$ according to a probability distribution $\{p_m\}_m$.

Let a dilation of the measurement process $\cM^{{\sA'}\to\sB\sM}$ be represented by an isometry $V_\cM^{{\sA'}\to\sB\sM\sE\sM'}$.
That is, when $\cM^{{\sA'}\to\sB\sM}$ is described by the Kraus operators $\{K_{j,m}^{{\sA'}\to\sB}\}_{j,m}$ satisfying $\sum_{j,m} (K_{j,m}^{{\sA'}\to\sB})^\dag K_{j,m}^{{\sA'}\to\sB} = \bI^{{\sA'}}$, so that 
\begin{align}
    \cM^{{\sA'}\to\sB\sM}(\cdot) 
    = \sum_{j,m} K_{j,m}^{{\sA'}\to\sB} (\cdot) (K_{j,m}^{{\sA'}\to\sB})^\dag \otimes \ketbra{m}{m}^{\sM}, 
\end{align}
a corresponding isometric dilation is given by
\begin{align}
\label{eq:dilation of measuremant}
    V_\cM^{{\sA'}\to\sB\sM\sE\sM'} = \sum_{j,m} K_{j,m}^{{\sA'}\to\sB} \otimes \ket{m}^\sM  \ket{j}^{\sE} \ket{m}^{\sM'}.
\end{align}
Then, a purification of the state $\Theta^{{\sA}\sB\sM}$ is 
\begin{align}
\label{eq:purification of post-meas state}
    \ket{\Theta}^{{\sA}\sB\sM\sE\sM'} = V_\cM^{{\sA'}\to\sB\sM\sE\sM'}\ket{\sigma}^{{\sA}{\sA'}}.
\end{align}

Using the purification, the irreversibility defined by Eq.~\eqref{inteq:30} can be equivalently written as 
\begin{align}
    \mathrm{Irv}(\sigma, \cM)
    \label{inteq:6}
    &= -\log \max_{\tau^{\sE\sM'}} \rF(\Theta^{{\sA}\sE\sM'}, \sigma^{\sA} \otimes \tau^{\sE\sM'}) \\
    &= \widetilde{I}_{\f{1}{2}}({\sA}; \sE\sM')_\Theta,
\end{align}
where we note that $\Theta^{\sA}=\sigma^{\sA}$.
This expression is obtained from the following lemma, which is a special case of a result in Ref.~{\cite[Proposition~43]{utsumi2025algorithmsUhlmanntrans}}.

\begin{lemma}[{\cite[Proposition~43]{utsumi2025algorithmsUhlmanntrans}}]
\label{lem:local optimal trans}
For any pure states $\ket{\rho}^{{\sA'}\sB\sC}$ and $\ket{\sigma}^{{\sA'}\sD}$, we have 
\begin{align}
\label{eq:opt ineq local op}
    \max_{\cP^{\sB\to\sD}}\rF\big(\cP^{\sB\to\sD}(\rho^{{\sA'}\sB}), \ketbra{\sigma}{\sigma}^{{\sA'}\sD}\big) 
    = \max_{\tau^\sC} \rF(\rho^{{\sA'}\sC}, \sigma^{\sA'} \otimes \tau^\sC),
\end{align}
where the maximizations on the left- and right-hand sides are taken over all channels from $\sB$ to $\sD$ and all states on $\sC$, respectively.
\end{lemma}

The right-hand side of Eq.~\eqref{eq:opt ineq local op} is known as the decoupling condition~\cite{hayden2008decoupling, dupuis2010decoupling, dupuis2014one}, and the channel that maximizes the left-hand side of Eq.~\eqref{eq:opt ineq local op} is characterized by Uhlmann's theorem~\cite{uhlmann1976transition}.
As an alternative, one can straightforwardly derive Lemma~\ref{lem:local optimal trans} by combining the result in Ref.~\cite[Remark~1]{koning2009operationalmeaningminmax} with a duality relation from Ref.~\cite[Lemma~6]{Hayashi2016Correlationmutualinformation}, both of which are formulated more generally.

The upper bound on the irreversibility, $-2\log\|\sigma^{\sA'}\|_\infty$, is then derived as follows.
We consider the case that the system ${\sA'}$ is discarded to the environment $\sE$, and the states $\{\ket{m}^\sM\}$ are prepared according to a certain probability distribution $\{p_m\}_m$.
In this case, the state on ${\sA}\sE\sM'$ is given by $\ketbra{\sigma}{\sigma}^{{\sA}\sE} \otimes \rho^{\sM'}$, where $\rho^{\sM'} = \sum_m p_m \ketbra{m}{m}^{\sM'}$.
Since 
\begin{align}
    \label{inteq:40}
    \max_{\tau^{\sE\sM'}} \rF(\Theta^{{\sA}\sE\sM'}, \sigma^{\sA} \otimes \tau^{\sE\sM'})
    \geq \max_{\tau^{\sE\sM'}}\rF(\ketbra{\sigma}{\sigma}^{{\sA}\sE}\otimes\rho^{\sM'}, \sigma^{\sA}\otimes\tau^{\sE\sM'})
\end{align}
holds for any state $\Theta^{{\sA}\sE\sM'}$ such that $\Theta^{\sA} = \sigma^{\sA}$, this measurement is the one that maximizes the irreversibility.
We can easily verify Eq.~\eqref{inteq:40} as follows.
Since $\Theta^{\sA} = \sigma^{\sA}$, there exists a channel $\cT^{\sE\to\sE\sM'}$ such that $\Theta^{\sA\sE\sM'} = \cT^{\sE\to\sE\sM'}(\ketbra{\sigma}{\sigma}^{\sA\sE})$~\cite{uhlmann1976transition}. Then, 
\begin{align}
    \max_{\tau^{\sE\sM'}} \rF\big(\Theta^{\sA\sE\sM'}, \sigma^{\sA}\otimes\tau^{\sE\sM'}\big) 
    &\geq \rF\big(\Theta^{{\sA}\sE\sM'}, \sigma^{\sA}\otimes\cT^{\sE\to\sE\sM'}(\til{\tau}^\sE)\big) \\
    &\geq \rF(\ketbra{\sigma}{\sigma}^{\sA\sE}, \sigma^\sA\otimes\til{\tau}^\sE) \\
    &\geq \rF(\ketbra{\sigma}{\sigma}^{\sA\sE} \otimes \rho^{\sM'}, \sigma^\sA\otimes \til{\tau}^{\sE\sM'})
\end{align}
holds for any state $\til{\tau}^{\sE\sM'}$, where $\til{\tau}^\sE = \tr_{\sM'}\big[\til{\tau}^{\sE\sM'}\big]$, so maximizing over $\til{\tau}^{\sE\sM'}$ yields Eq.~\eqref{inteq:40}.

The right-hand side of Eq.~\eqref{inteq:40} can be reduced to 
\begin{align}
    \max_{\tau^{\sE\sM'}} \rF(\ketbra{\sigma}{\sigma}^{{\sA}\sE} \otimes \rho^{\sM'}, \sigma^{\sA} \otimes \tau^{\sE\sM'}) 
    \label{inteq:42}
    = \max_{\tau^\sE} \rF(\ketbra{\sigma}{\sigma}^{{\sA}\sE}, \sigma^{\sA} \otimes \tau^\sE),
\end{align}
which follows from a sandwich argument:
\begin{align}
    \max_{\tau^\sE} \rF(\ketbra{\sigma}{\sigma}^{{\sA}\sE}, \sigma^{\sA} \otimes \tau^\sE)
    &\geq \max_{\tau^{\sE\sM'}} \rF(\ketbra{\sigma}{\sigma}^{{\sA}\sE} \otimes \rho^{\sM'}, \sigma^{\sA} \otimes \tau^{\sE\sM'}) \\
    &\geq \max_{\tau^\sE} \rF(\ketbra{\sigma}{\sigma}^{{\sA}\sE} \otimes \rho^{\sM'}, \sigma^{\sA} \otimes \tau^\sE \otimes \rho^{\sM'}) \\
    &= \max_{\tau^\sE} \rF(\ketbra{\sigma}{\sigma}^{{\sA}\sE}, \sigma^{\sA} \otimes \tau^\sE).
\end{align}
In the first inequality, we used the monotonicity of the fidelity under the partial trace over $\sM'$. The second inequality follows from restricting the maximization to product states of the form $\tau^\sE \otimes \rho^{\sM'}$.
In the last line, we used the multiplicativity of the fidelity: $\rF(\omega_1\otimes\omega_2, \zeta_1\otimes\zeta_2) = \rF(\omega_1, \zeta_1)\rF(\omega_2, \zeta_2)$. 
We further calculate the right-hand side of Eq.~\eqref{inteq:42} as 
\begin{align}
    \max_{\tau^\sE} \rF(\ketbra{\sigma}{\sigma}^{{\sA}\sE}, \sigma^{\sA} \otimes \tau^\sE)
    &= \max_{\tau^\sE} \bra{\sigma}^{{\sA}\sE}(\sigma^{\sA}\otimes\tau^\sE)\ket{\sigma}^{{\sA}\sE} \\
    &= \max_{\tau^\sE} \sum_j \lambda_j^2 \bra{j}\tau\ket{j}^\sE \\
    &= \max_j \lambda_j^2 \\
    \label{inteq:41}
    &= \|\sigma^{\sA'}\|_\infty^2,
\end{align}
where $\{\lambda_j\}_j$ and $\{\ket{j}^\sE\}_j$ are eigenvalues and eigenstates of $\sigma^\sE = \tr_\sA[\ketbra{\sigma}{\sigma}^{\sA\sE}]$, respectively.
In the third equation, we used the fact that $\tau^\sE$ that achieves the maximum is given by the state $\ket{j}^\sE$ corresponding to the index $j$ for which $\lambda_j$ is maximal, and in the last line, we used that $\sE$ is the system discarded from ${\sA'}$.

Hence, from Eqs.~\eqref{inteq:6},~\eqref{inteq:40},~\eqref{inteq:42}, and~\eqref{inteq:41}, we evaluate the irreversibility as $\mathrm{Irv}(\sigma, \cM) \leq -2\log \|\sigma^{\sA'}\|_\infty$.
As is clear from the derivation, this upper bound is achieved by a measurement process that discards the system $\sA'$ into the environment $\sE$.


\subsection{Proof of Proposition~\ref{thm:Trade-off relation between the information gain and irreversibility}}
\label{sec:proof tradeoff onfo gain and q-disturb}

We prove Proposition~\ref{thm:Trade-off relation between the information gain and irreversibility} restated below, which shows a trade-off relation between the information gain and irreversibility.

\begin{restated}{proposition}{thm:Trade-off relation between the information gain and irreversibility}
    For a state $\ket{\sigma}^{\sA\sA'}$, the one-shot information gain and irreversibility of a measurement process $\cM^{{\sA'}\to\sB\sM}$ satisfy the relation: $\mathrm{IG}(\sigma, \Lambda) \leq \mathrm{Irv}(\sigma, \cM)$, where $\Lambda^{\sA'\to\sM} = \tr_\sB \circ \cM^{\sA'\to\sB\sM}$.
\end{restated}

\begin{proof}[Proof of Proposition~\ref{thm:Trade-off relation between the information gain and irreversibility}]

As seen in Sec.~\ref{sec:upperbound on q-distub}, the irreversibility can be equivalently rephrased as $\mathrm{Irv}(\sigma, \cM) = -\log \max_{\tau^{\sE\sM'}} \rF(\Theta^{{\sA}\sE\sM'}, \sigma^{\sA} \otimes \tau^{\sE\sM'})$ using the purification $\ket{\Theta}^{{\sA}\sB\sM\sE\sM'}$ of $\Theta^{{\sA}\sB\sM} = \cM^{{\sA'}\to\sB\sM}(\ketbra{\sigma}{\sigma}^{{\sA}{\sA'}})$.
Our statement then follows from the straightforward calculation:
\begin{align}
    \mathrm{Irv}(\sigma, \cM) 
    &= -\log \max_{\tau^{\sE\sM'}} \rF(\Theta^{{\sA}\sE\sM'}, \sigma^{\sA} \otimes \tau^{\sE\sM'}) \\
    &\geq -\log \max_{\tau^{\sM'}} \rF(\Theta^{{\sA}\sM'}, \sigma^{\sA} \otimes \tau^{\sM'}) \\
    &= -\log \max_{\tau^\sM} \rF(\Theta^{{\sA}\sM}, \sigma^{\sA} \otimes \tau^\sM) \\
    &= \mathrm{IG}(\sigma, \Lambda),
\end{align}
where the inequality is due to the monotonicity of the fidelity under partial trace.  
In the second equation, we used that the systems $\sM$ and $\sM'$ are isomorphic, that is, $\Theta^{{\sA}\sM'} = \cI^{\sM\to\sM'}(\Theta^{{\sA}\sM})$, where $\cI^{\sM\to\sM'}$ denotes the identity channel under the identification of $\sM$ with $\sM'$; see also Eqs.~\eqref{eq:dilation of measuremant} and~\eqref{eq:purification of post-meas state}.

\end{proof}

\subsection{Proof of \Cref{prop:one-shot noise disturbance}}
\label{sec:proof of disturbance noise trade offs}

We prove \Cref{prop:one-shot noise disturbance}, which, for two observables, establishes a trade-off relation between the disturbance of one and the noise on the other. To this end, we use the following norm bound.

\begin{lemma}
\label{lem:monogamy norm bound}
Let $\{Q_j^\sA\}_j$ and $\{R_k^\sA\}_k$ be any PVMs on $\cH^\sA$, and let $\{F_j^\sB\}_j$ and $\{G_k^\sC\}_k$ be any POVMs on $\cH^\sB$ and $\cH^\sC$, respectively. Then, 
\begin{align}
\label{eq:monogamy norm bound}
    \Big\|\sum_j Q_j^\sA \otimes F_j^\sB \otimes \bI^\sC
    + \sum_k R_k^\sA \otimes \bI^\sB \otimes G_k^\sC\Big\|_\infty
    \leq 1 + \max_{j,k}\big\|Q_j^\sA R_k^\sA\big\|_\infty.
\end{align}
\end{lemma}

\begin{proof}[Proof of \Cref{lem:monogamy norm bound}]
We first show that any two projectors $\Pi_1$ and $\Pi_2$ satisfy $\|\Pi_1+\Pi_2\|_\infty \leq 1+\|\Pi_1\Pi_2\|_\infty$.
For a unit vector $\ket{\psi}$, let $t \coloneqq \bra{\psi}(\Pi_1 + \Pi_2)\ket{\psi}$ and $s \coloneqq \|\Pi_1 \Pi_2\|_\infty$.
Noting that $\Pi_l^2 = \Pi_l$, we have $t = \|\Pi_1\ket{\psi}\|^2 + \|\Pi_2\ket{\psi}\|^2 \geq 2 \|\Pi_1\ket{\psi}\|\|\Pi_2\ket{\psi}\|$, which follows from $a^2+b^2 \geq 2ab$ for $a, b \in \bR$.
Combining this with $\Pi_1\Pi_2 = \Pi_1\Pi_1\Pi_2\Pi_2$, we obtain $|\bra{\psi}\Pi_1\Pi_2\ket{\psi}| \leq \|\Pi_1\ket{\psi}\| \|\Pi_1\Pi_2\|_\infty \|\Pi_2\ket{\psi}\| \leq st/2$. Here, we used the Cauchy--Schwarz inequality and the definition of the operator norm.
Hence $t^2 \leq \big\|(\Pi_1+\Pi_2)\ket{\psi}\big\|^2 = t + 2\mathrm{Re}(\bra{\psi}\Pi_1\Pi_2\ket{\psi}) \leq t(1+s)$. If $t=0$, the inequality is trivial; otherwise, dividing by $t$ and maximizing over all unit vectors $\ket{\psi}$ yields $\|\Pi_1+\Pi_2\|_\infty \leq 1+\|\Pi_1\Pi_2\|_\infty$.

We now use Naimark's dilation theorem~\cite{Naimark1943representationadditiveoperator} to apply the above projector bound to the POVMs.
There exist isometries $V^{\sB\to\sB'}$ and $W^{\sC\to\sC'}$ such that $F_j^\sB = (V^{\sB\to\sB'})^\dag \hat{F}_j^{\sB'} V^{\sB\to\sB'}$ and $G_k^\sC = (W^{\sC\to\sC'})^\dag \hat{G}_k^{\sC'} W^{\sC\to\sC'}$, where $\{\hat{F}_j^{\sB'}\}_j$ and $\{\hat{G}_k^{\sC'}\}_k$ are PVMs on the larger systems $\sB'$ and $\sC'$, respectively.
Define an isometry $U^{\sB\sC\to\sB'\sC'} \coloneqq V^{\sB\to\sB'} \otimes W^{\sC\to\sC'}$ and operators $\hat{\Pi}_1^{\sA\sB'} \coloneqq \sum_j Q_j^\sA \otimes \hat{F}_j^{\sB'}$ and $\hat{\Pi}_2^{\sA\sC'} \coloneqq \sum_k R_k^\sA\otimes \hat{G}_k^{\sC'}$.
Both are projectors, since $\{Q_j\}_j$, $\{R_k\}_k$, $\{\hat{F}_j\}_j$, and $\{\hat{G}_k\}_k$ each consist of mutually orthogonal projectors summing to the identity.
Using $V^\dag V = \bI^\sB$ and $W^\dag W = \bI^\sC$, the operator on the left-hand side of Eq.~\eqref{eq:monogamy norm bound} equals $U^\dag(\hat{\Pi}_1+\hat{\Pi}_2)U$, so its norm is at most $\|\hat{\Pi}_1+\hat{\Pi}_2\|_\infty \leq 1 + \|\hat{\Pi}_1\hat{\Pi}_2\|_\infty$.

It remains to evaluate $\|\hat{\Pi}_1\hat{\Pi}_2\|_\infty$. We have $\hat{\Pi}_1^{\sA\sB'}\hat{\Pi}_2^{\sA\sC'} = \sum_{j,k} Q_j^\sA R_k^\sA \otimes \hat{F}_j^{\sB'} \otimes \hat{G}_k^{\sC'}$, and thus
\begin{align}
    (\hat{\Pi}_1\hat{\Pi}_2)^\dag(\hat{\Pi}_1\hat{\Pi}_2)
    = \sum_{j,k} (Q_j R_k)^\dag (Q_j R_k) \otimes \hat{F}_j \otimes \hat{G}_k,
\end{align}
where we used the fact that $\{\hat{F}_j\}_j$ and $\{\hat{G}_k\}_k$ are PVMs.
The projectors $\{\hat{F}_j^{\sB'}\otimes\hat{G}_k^{\sC'}\}_{j,k}$ are mutually orthogonal and sum to the identity on $\sB'\sC'$, so the right-hand side is block diagonal with respect to them, with the $(j,k)$-th block acting as $(Q_j^\sA R_k^\sA)^\dag(Q_j^\sA R_k^\sA)$ on $\sA$. The operator norm of a block-diagonal operator is the largest among those of its blocks, and hence
\begin{align}
    \big\|\hat{\Pi}_1\hat{\Pi}_2\big\|_\infty^2
    &= \big\|(\hat{\Pi}_1\hat{\Pi}_2)^\dag(\hat{\Pi}_1\hat{\Pi}_2)\big\|_\infty \\
    &\leq \max_{j,k}\|Q_j R_k\|_\infty^2.
\end{align}
Here, the inequality arises if $\hat{F}_j = 0$ or $\hat{G}_k = 0$ for some $j$ or $k$, in which case the corresponding blocks are absent from the direct sum.
Combining this with the bound above, we obtain Eq.~\eqref{eq:monogamy norm bound}.

\end{proof}

Using the above lemma, we now prove \Cref{prop:one-shot noise disturbance}.

\begin{restated}{proposition}{prop:one-shot noise disturbance}
For a measurement process $\cM^{\sA'\to\sB\sM}$ with $\Lambda^{\sA'\to\sM} = \tr_\sB\circ\cM^{\sA'\to\sB\sM}$ and any two observables $O_1^{\sA'}$ and $O_2^{\sA'}$, we have
\begin{align}
\tag{\ref*{eq:one-shot noise disturbance}}
\label{eq:one-shot noise disturbance SI}
    2^{-\mathrm{Dst}(O_1, \cM)} + 2^{-\mathrm{Noi}(O_2, \Lambda)}
    \leq 1 + \sqrt{c(O_1, O_2)},
\end{align}
where $c(O_1, O_2) \coloneqq \max_{j,k}\big\|P_{j|1}^{\sA'}P_{k|2}^{\sA'}\big\|_\infty^2$.
\end{restated}

\begin{proof}[Proof of \Cref{prop:one-shot noise disturbance}]
Let $V_\cM^{\sA'\to\sB\sM\sE\sM'}$ be the isometric dilation of $\cM^{\sA'\to\sB\sM}$ in Eq.~\eqref{eq:dilation of measuremant} and let $\ket{\Theta}^{\sA\sB\sM\sE\sM'} = V_\cM^{\sA'\to\sB\sM\sE\sM'}\ket{\Phi}^{\sA\sA'}$ be a purification of the post-measurement state, where $\ket{\Phi}^{\sA\sA'}$ is the MES.
Measuring $\sA$ with the complex-conjugate PVM $\{\bar{P}_{j|l}^\sA\}_j$ indexed by $l$ yields outcome $j$ and a state $P_{j|l}^{\sA'}/d_{j|l}$ on $\sA'$, with probability $\tr[\bar{P}_{j|l}^{\sA}\Phi^{\sA\sA'}] = d_{j|l}/d_{\sA'}$.
This follows from $\tr[\bar{P}_{j|l}] = d_{j|l}$ and $\tr_\sA[\bar{P}_{j|l}^\sA\Phi^{\sA\sA'}]/\tr[\bar{P}_{j|l}^{\sA}\Phi^{\sA\sA'}] = P_{j|l}^{\sA'}/d_{j|l}$.
Then, the ensemble 
\begin{align}
\label{eq:ensemble for dist and noise}
    \scE_{O_l, \cM} = \big\{d_{j|l}/d_{\sA'}; \ \cM^{{\sA'}\to\sB\sM}(P_{j|l}^{\sA'}/d_{j|l})\big\}_{j=1}^{D_l}
\end{align}
is induced on $\sB\sM$.

By the definition of the disturbance and the purification, we have 
\begin{align}
    2^{-\mathrm{Dst}(O_1, \cM)}
    &= \max_{\{F_j\}: \, \mathrm{POVM}}\sum_j \tr\big[(\bar{P}_{j|1}^\sA \otimes F_j^{\sB\sM}) \cM^{\sA'\to\sB\sM}(\Phi^{\sA\sA'})\big]  \\
    \label{inteq:dst as guessing}
    &= \max_{\{F_j\}: \, \mathrm{POVM}}\sum_j \bra{\Theta}^{\sA\sB\sM\sE\sM'}\big(\bar{P}_{j|1}^\sA \otimes F_j^{\sB\sM} \otimes \bI^{\sE\sM'}\big)\ket{\Theta}^{\sA\sB\sM\sE\sM'}.
\end{align}

For the noise on an observable, the guess is made from the register $\sM$ alone. Since the registers $\sM$ and $\sM'$ are perfectly correlated by Eq.~\eqref{eq:dilation of measuremant}, any additional measurement on $\sM$ is reproduced by the corresponding measurement on $\sM'$.
Thus, 
\begin{align}
    2^{-\mathrm{Noi}(O_2, \Lambda)}
    &= \max_{\{G_k\}: \, \mathrm{POVM}}\sum_k \tr\big[(\bar{P}_{k|2}^\sA \otimes G_k^{\sM}) \Lambda^{\sA'\to\sM}(\Phi^{\sA\sA'})\big] \\
    &= \max_{\{G_k\}: \, \mathrm{POVM}}\sum_k \bra{\Theta}^{\sA\sB\sM\sE\sM'}\big(\bar{P}_{k|2}^\sA \otimes \bI^{\sB} \otimes G_k^{\sM} \otimes \bI^{\sE\sM'}\big)\ket{\Theta}^{\sA\sB\sM\sE\sM'} \\
    &= \max_{\{G_k\}: \, \mathrm{POVM}}\sum_k \bra{\Theta}^{\sA\sB\sM\sE\sM'}\big(\bar{P}_{k|2}^\sA \otimes \bI^{\sB\sM\sE} \otimes G_k^{\sM'}\big)\ket{\Theta}^{\sA\sB\sM\sE\sM'} \\
    \label{inteq:noi as guessing}
    &\leq \max_{\{G_k\}: \, \mathrm{POVM}}\sum_k \bra{\Theta}^{\sA\sB\sM\sE\sM'}\big(\bar{P}_{k|2}^\sA \otimes \bI^{\sB\sM} \otimes G_k^{\sE\sM'}\big)\ket{\Theta}^{\sA\sB\sM\sE\sM'},
\end{align}
where the inequality follows because the maximization over POVMs on $\sE\sM'$ includes those acting on $\sM'$ alone.

Adding Eqs.~\eqref{inteq:dst as guessing} and~\eqref{inteq:noi as guessing} and applying \Cref{lem:monogamy norm bound} with $\sB\sM$ and $\sE\sM'$ in place of $\sB$ and $\sC$, we obtain
\begin{align}
    &2^{-\mathrm{Dst}(O_1, \cM)} + 2^{-\mathrm{Noi}(O_2, \Lambda)}  \notag\\
    &\leq  \max_{\{F_j\}, \{G_k\}: \, \mathrm{POVM}}\bra{\Theta}^{\sA\sB\sM\sE\sM'} \Big(\sum_j \bar{P}_{j|1}^\sA \otimes F_j^{\sB\sM} \otimes \bI^{\sE\sM'}
    + \sum_k \bar{P}_{k|2}^\sA \otimes \bI^{\sB\sM} \otimes G_k^{\sE\sM'}\Big)\ket{\Theta}^{\sA\sB\sM\sE\sM'} \\
    &\leq \max_{\{F_j\}, \{G_k\}: \, \mathrm{POVM}}\Big\|\sum_j \bar{P}_{j|1}^\sA \otimes F_j^{\sB\sM} \otimes \bI^{\sE\sM'}
    + \sum_k \bar{P}_{k|2}^\sA \otimes \bI^{\sB\sM} \otimes G_k^{\sE\sM'}\Big\|_\infty \\
    &\leq 1 + \max_{j,k}\big\|\bar{P}_{j|1}^\sA\bar{P}_{k|2}^\sA\big\|_\infty \\
    \label{inteq:dist noise trade off proof}
    &= 1 + \sqrt{c(O_1, O_2)},
\end{align}
where $c(O_1, O_2) \coloneqq \max_{j,k}\big\|P_{j|1}^{\sA'} P_{k|2}^{\sA'}\big\|_\infty^2 = \max_{j,k}\big\|\bar{P}_{j|1}^\sA\bar{P}_{k|2}^\sA\big\|_\infty^2$.
The last equality holds because complex conjugation preserves singular values.
For nondegenerate observables, $c(O_1, O_2)$ reduces to $\max_{j,k}|\braket{e_{j|1}}{e_{k|2}}|^2$, where $\{\ket{e_{j|l}}\}_j$ is the eigenbasis of $O_l^{\sA'}$ for $l = 1, 2$.

\end{proof}

A related prior result is the trade-off relation established in Ref.~\cite{Buscemi2014NoiseandDisturbance}, which lower bounds the sum of the Shannon-entropic noise and disturbance by $-\log c(O_1,O_2)$.
That bound and ours are not directly comparable: the R\'{e}nyi entropic uncertainty relations involve conjugate orders (see~\cite{Coles2017Entropicuncertainty} for details), whereas both our disturbance and noise are defined in terms of min-entropies. Within this min-entropy pairing, our bound is optimal, with equality attained for $O_1 = Z$ and $O_2 = X$ on a qubit.


\section{Proof of Theorem~\ref{thm:Trade-off relation between the information gain and the observable disturbance}}
\label{sec:tradeoff infogain and classical disturb}

We here prove Theorem~\ref{thm:Trade-off relation between the information gain and the observable disturbance}, which bounds the irreversibility in Definition~\ref{def:irreversibility} in terms of the disturbance for an observable in Definition~\ref{def:observable disturbance}. The proof strategy is to directly apply Theorem~\ref{thm:1} to the post-measurement state.

\begin{restated}{theorem}{thm:Trade-off relation between the information gain and the observable disturbance}
For a measurement process $\cM^{{\sA'}\to\sB\sM}$ on the MES $\ket{\Phi}^{\sA\sA'}$, the following trade-off relation holds:
\begin{align}
\tag{\ref*{eq:tradeoff infogain and c-disturbance equation}}
\label{eq:tradeoff infogain and c-disturbance equation SI}
    \nu(\Omega)\big(1-2^{-\mathrm{Irv}(\Phi, \cM)}\big) 
    \leq 1 - \big\langle 2^{-2\mathrm{Dst}(O, \cM)} \big\rangle,
\end{align}
where $\big\langle 2^{-2\mathrm{Dst}(O, \cM)}\big\rangle \coloneqq \sum_{l=1}^L q_l 2^{-2\mathrm{Dst}(O_l, \cM)}$, and $\Omega^{\sA\hat{\sB}}$ is given in Eq.~\eqref{inteq:90}.
\end{restated}

\begin{proof}[Proof of Theorem~\ref{thm:Trade-off relation between the information gain and the observable disturbance}]

The infidelity $\varepsilon(\cP(\xi), \phi)$ and the irreversibility $\mathrm{Irv}(\sigma, \cM)$ defined by Eqs.~\eqref{eq:def quantum error main} and~\eqref{eq:definition of q-disturb}, respectively, are related as
\begin{align}
    \mathrm{Irv}(\Phi, \cM) = -\log \max_\cP\big(1-\varepsilon(\cP(\Theta), \Phi)\big), 
\end{align}
where $\Theta^{\sA\sB\sM} = \cM^{\sA'\to\sB\sM}(\ketbra{\Phi}{\Phi}^{\sA\sA'})$ and $\ket{\Phi}^{\sA\sA'}$ is the MES.

When applying Theorem~\ref{thm:1} to $\Theta^{\sA\sB\sM}$ (with $\sB$ replaced by $\sB\sM$ and $\hat{\sB} \cong \sA'$), we choose the POVMs $\{E_l^\sA\}_{l=1}^L$ to be the PVMs $\big\{\{\bar{P}_{j|l}^\sA\}_{j=1}^{D_l}\big\}_{l=1}^L$, so that $\Omega^{\sA\hat{\sB}}$ is given by Eq.~\eqref{inteq:90}.
Since $\bar{P}_{j|l}^\sA\ket{\Phi}^{\sA\hat{\sB}} = P_{j|l}^{\hat{\sB}}\ket{\Phi}^{\sA\hat{\sB}}$ implies $\Omega^{\sA\hat{\sB}}\ket{\Phi}^{\sA\hat{\sB}} = \ket{\Phi}^{\sA\hat{\sB}}$, the MES $\ket{\Phi}^{\sA\hat{\sB}}$ lies in the principal eigenspace $\cV_\Omega$ with $\|\Omega\|_\infty = 1$.
Moreover, measuring $\sA$ with $\{\bar{P}_{j|l}^\sA\}_j$ yields outcome $j$ with probability $d_{j|l}/d_{\sA'}$ and induces the ensemble $\scE_{\Theta, E_l} = \scE_{O_l, \cM}$, which is explicitly given by Eq.~\eqref{eq:ensemble for dist and noise}.
This ensemble defines $\mathrm{Dst}(O_l, \cM)$ through $\rP_{\rm opt}(\scE_{O_l, \cM})$.
We then obtain 
\begin{align}
    \nu(\Omega)\big(1 - 2^{-\mathrm{Irv}(\Phi, \cM)}\big) 
    &= \nu(\Omega) \min_\cP\varepsilon(\cP(\Theta), \Phi) \\
    &\leq 1- \big\lag \rP_{\rm opt}^2(\scE_{\Theta, E})\big\rag \\
    \label{inteq:7}
    &= 1- \big\lag 2^{-2\mathrm{Dst}(O, \cM)}\big\rag,
\end{align}
where the inequality follows from Eq.~\eqref{eq:equation of main thm1} in Theorem~\ref{thm:1} with $\|\Omega\|_\infty = 1$ and the last equation is due to the definition of the disturbance (\Cref{def:observable disturbance}).

\end{proof}


\section{Numerical evaluation of the information gain--disturbance trade-off}
\label{sec:measurement trade-off simulation}

We numerically evaluate the trade-off inequality between the information gain and disturbance for an observable in the case $L=2$. 
In particular, we examine the tightness of Eqs.~\eqref{eq:info gain and cl distu L=2} and~\eqref{eq:info gain and single cl distu L=2} with respect to the measurement sharpness and the choice of observables.
We restate Eqs.~\eqref{eq:info gain and cl distu L=2} and~\eqref{eq:info gain and single cl distu L=2} below:
\begin{align}
\label{eq:info gain and cl distu L=2 restated in appendix}
    &2^{-2\mathrm{Dst}(O_1, \cM)} + 2^{-2\mathrm{Dst}(O_2, \cM)}
    \leq 2\left(1 - \nu\left(\frac{\Omega_1 + \Omega_2}{2}\right)\left(1-2^{-\mathrm{IG}(\Phi, \Lambda)}\right)\right), \\
\label{eq:info gain and single cl distu L=2 restated in appendix}
    &\mathrm{Dst}(O_2, \cM)
    \geq -\f{1}{2}\log{\left(1-2\left(1-\f{1}{d_{\sA'}}\right)\nu\left(\f{\Omega_1 + \Omega_2}{2}\right)\right)}.
\end{align}

In the evaluation of Eq.~\eqref{eq:info gain and cl distu L=2 restated in appendix}, we use the $\theta$-Fourier transform $F_d^\theta$, which we defined in Sec.~\ref{sec:numerical evaluation of spectral gap bounds}, and particularly consider two observables: 
\begin{align}
    &O_1 = F_d^\theta Z (F_d^{\theta})^\dag = \sum_j \omega^j \ketbra{f_j^\theta}{f_j^\theta}, \\
    &O_2 = F_d^{-\theta} Z (F_d^{-\theta})^\dag = \sum_j \omega^j \ketbra{f_j^{-\theta}}{f_j^{-\theta}},    
\end{align}
where $\ket{f_j^\theta} = F_d^\theta\ket{j}$, and $Z = \sum_j \omega^j \ketbra{j}{j}$ is the qudit Pauli-$Z$ operator, whose eigenvalues are given by the $d$-th roots of unity $\omega^j = e^{i 2\pi j/d}$.
Although the qudit Pauli operators are not Hermitian for $d > 2$, they only specify the measurement bases here; any nondegenerate Hermitian observable with the same eigenbasis has the same disturbance.
We fix a measurement process $\cM$ to be a Pauli-$X$ measurement mixed with white noise: 
\begin{align}
    \cM(\cdot) = \sum_j\sqrt{M_j}(\cdot)\sqrt{M_j}\otimes\ketbra{j}{j},
    \ \ \ \ \ \ 
    M_j = (1-\eta) F_d \ketbra{j}{j} (F_d)^\dag + \eta\f{\bI}{d},
\end{align}
where $\{M_j\}_j$ is a POVM with the noise strength $\eta \in [0,1]$ that quantifies the degree of measurement unsharpness. If $\eta = 0$, the measurement is sharp, whereas $\eta = 1$ corresponds to a completely noisy and uninformative measurement.
We then plot, in Fig.~\ref{fig:double_obs_dist_vs_info_gain}, the left- and right-hand sides of Eq.~\eqref{eq:info gain and cl distu L=2 restated in appendix} for $d=2, 3, 4$ as functions of $\theta$, for several values of $\eta$.
To assess the tightness of the bound, we use the nonnegative gap between these two sides, defined as the right-hand side minus the left-hand side and shown in the inset. A smaller gap indicates a tighter bound, although a vanishing gap may correspond to a trivial case.


\begin{figure*}[t]
  \centering

  \begin{minipage}[t]{0.31\textwidth}
    \centering
    \includegraphics[width=\linewidth]{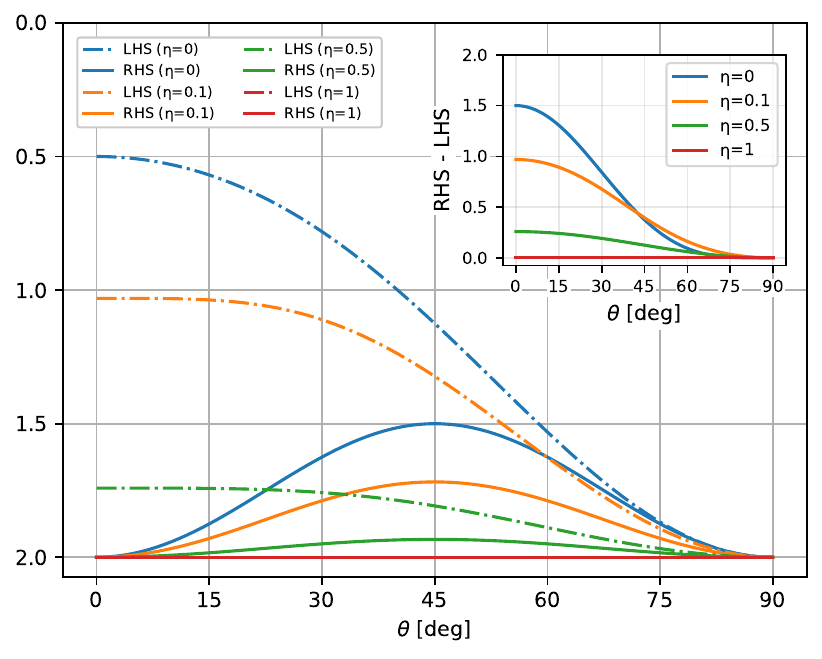}

    \vspace{1mm}
    \small\textbf{(a)} $d=2$
  \end{minipage}\hfill
  \begin{minipage}[t]{0.31\textwidth}
    \centering
    \includegraphics[width=\linewidth]{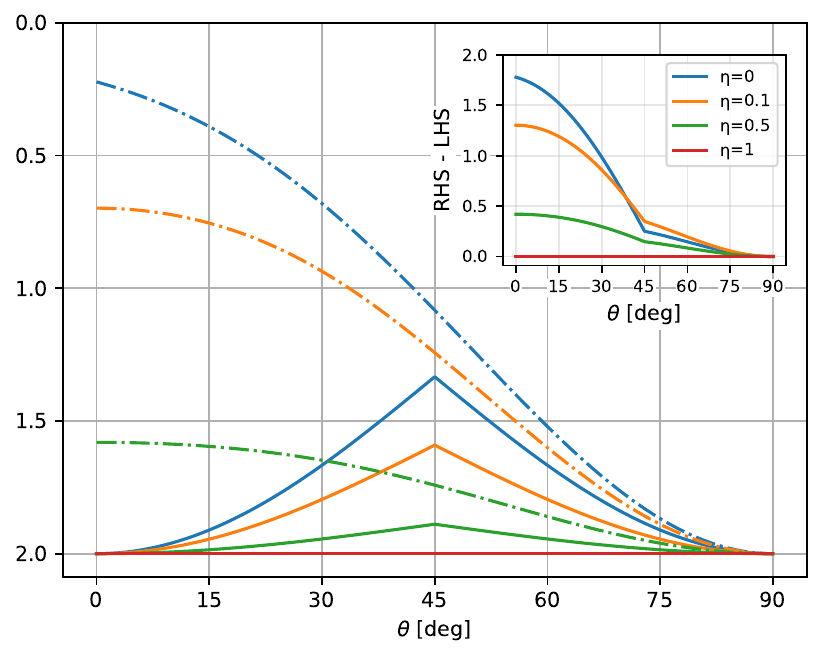}

    \vspace{1mm}
    \small\textbf{(b)} $d=3$
  \end{minipage}\hfill
  \begin{minipage}[t]{0.31\textwidth}
    \centering
    \includegraphics[width=\linewidth]{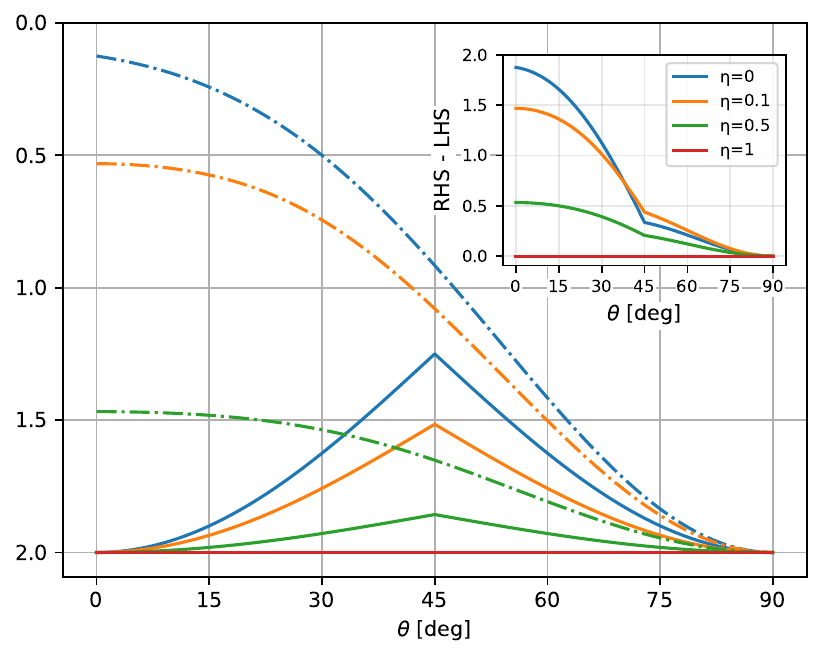}

    \vspace{1mm}
    \small\textbf{(c)} $d=4$
  \end{minipage}

  \caption{The left- and right-hand sides of Eq.~\eqref{eq:info gain and cl distu L=2 restated in appendix} 
  (dash-dotted and solid, respectively) are plotted for several values of the measurement 
  unsharpness parameter $\eta$, with $\theta$ on the horizontal axis. Note that smaller values 
  on the left-hand side indicate larger disturbance, while, for a fixed pair of observables, smaller values on the right-hand side indicate larger information gain; thus, the vertical axis is 
  inverted. The inset plots the difference between the right-hand side and the left-hand side.}
  \label{fig:double_obs_dist_vs_info_gain}
\end{figure*}

\begin{figure}[h]
  \centering
  \begin{minipage}[t]{0.30\columnwidth}
    \centering
    \includegraphics[width=\linewidth]{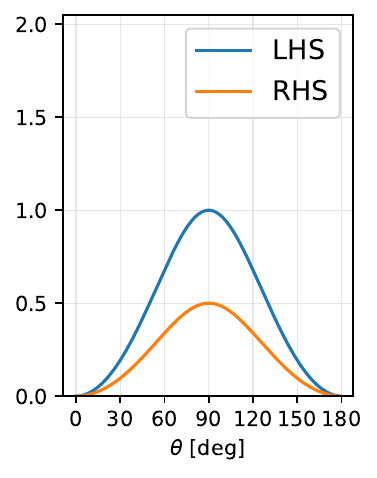}
    \vspace{1mm}
    \small\textbf{(a)} $d=2$
  \end{minipage}\hfill
  \begin{minipage}[t]{0.30\columnwidth}
    \centering
    \includegraphics[width=\linewidth]{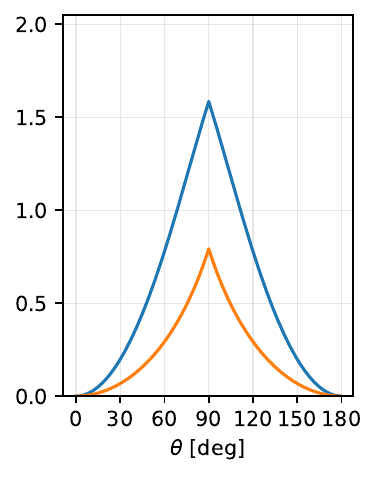}
    \vspace{1mm}
    \small\textbf{(b)} $d=3$
  \end{minipage}\hfill
  \begin{minipage}[t]{0.30\columnwidth}
    \centering
    \includegraphics[width=\linewidth]{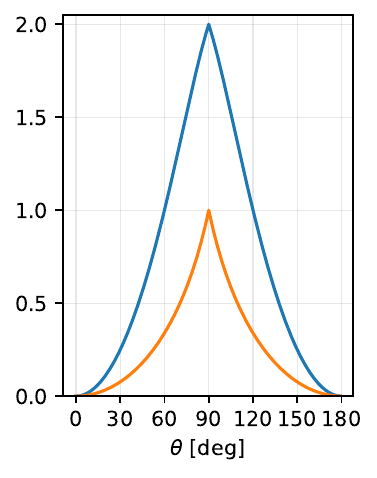}
    \vspace{1mm}
    \small\textbf{(c)} $d=4$
  \end{minipage}
  \caption{The figures show the left- and right-hand sides of Eq.~\eqref{eq:info gain and single cl distu L=2 restated in appendix} as $\theta$ varies. The behavior around $\theta = 90^\circ$ differs between $d=2$ and $d>2$.}
  \label{fig:single_obs_dist_vs_info_gain}
\end{figure}


From Fig.~\ref{fig:double_obs_dist_vs_info_gain}, we first observe that, for each value $\eta < 1$, the gap between the left- and right-hand sides of Eq.~\eqref{eq:info gain and cl distu L=2 restated in appendix} is generally largest near $\theta = 0$ and decreases as $\theta$ approaches $45^\circ$. 
For fixed $\eta$, the measurement process, and thus the information gain, is independent of $\theta$; the $\theta$-dependence of the right-hand side originates from the spectral gap associated with the two observables. 
At $\theta = 0$, the two observables coincide with $Z$ and $\nu((\Omega_1+\Omega_2)/2)$ vanishes. The right-hand side then reduces to the trivial upper bound $2$, even though a sharp Pauli-$X$ measurement strongly disturbs the common $Z$ observable. 
By contrast, at $\theta = 45^\circ$, the eigenbases of the observables become MUBs and the spectral gap reaches its maximum $1/2$. Since the information gain is independent of $\theta$ for fixed $\eta$, this minimizes the right-hand side of Eq. \eqref{eq:info gain and cl distu L=2 restated in appendix} for $\eta<1$ among the observable pairs obtained by varying $\theta$.
The approach to the MUB point differs between $d=2$ and $d>2$. For $d=2$, $\nu((\Omega_1+\Omega_2)/2)$ varies smoothly with $\theta$, while for $d>2$, multiple eigenvalues of $(\Omega_1+\Omega_2)/2$ below the maximum become degenerate at the MUB point and split into branches with different slopes as $\theta$ deviates. 
Consequently, $\nu((\Omega_1+\Omega_2)/2)$, which is determined by the leading branch among these split eigenvalues, exhibits a kink at $\theta = 45^\circ$.

The dependence on $\eta$ clarifies the physical significance of a small gap between the right- and left-hand sides. When $\eta=1$, the measurement yields no information and causes no disturbance, so both sides equal the trivial value $2$ for every $\theta$. In contrast, for small $\eta$, the measurement is nearly sharp and enables us to acquire substantial information. The reduction of the gap near the MUB point therefore shows that the bound remains informative in a regime with nonzero information gain and disturbance, rather than being saturated only in the trivial no-information limit.

Next, we evaluate Eq.~\eqref{eq:info gain and single cl distu L=2 restated in appendix}.
To this end, we take two observables as $O_1 = Z$ and $O_2 = F_d^\theta Z (F_d^{\theta})^\dag$, where
we note that $O_2$ corresponds to the qudit Pauli-$X$ operator at $\theta = 90^\circ$.
The measurement is taken to be the projective measurement in the computational basis: $\{M_j = \ketbra{j}{j}\}_j$.
In this setting, Fig.~\ref{fig:single_obs_dist_vs_info_gain} shows the left- and right-hand sides of Eq.~\eqref{eq:info gain and single cl distu L=2 restated in appendix} for $d=2,3,4$.

As shown in Fig.~\ref{fig:single_obs_dist_vs_info_gain}, Eq.~\eqref{eq:info gain and single cl distu L=2 restated in appendix} is saturated at $\theta=0$. This saturation is trivial: at this point, the two observables coincide, $O_2=O_1=Z$, and the measurement in the $Z$ basis leaves their eigenstates undisturbed, so that $\mathrm{Dst}(O_2,\cM) = 0$. 
The spectral gap also vanishes, and hence the right-hand side of Eq.~\eqref{eq:info gain and single cl distu L=2 restated in appendix} is zero. 
As $\theta$ increases, the eigenbasis of $O_2$ moves away from the measured $Z$ basis, and the disturbance and its lower bound become nonzero. The different behavior around $\theta = 90^\circ$ for $d=2$ and $d>2$ originates from the dimension-dependent spectral structure discussed above. At $\theta = 90^\circ$, $O_2$ is the qudit Pauli-$X$ observable and is mutually unbiased with respect to the $Z$ basis. The $Z$-basis measurement then completely erases the $X$-eigenstate label, giving $\mathrm{Dst}(O_2,\cM)=\log d$, whereas Eq.~\eqref{eq:info gain and single cl distu L=2 restated in appendix} gives the lower bound $\frac{1}{2}\log d$. Thus, the bound is nontrivial at the MUB point but differs from the actual disturbance by a factor of two for all $d$.

\putbib[ref]
\end{bibunit}

\end{document}